%% file: BCG.tex
\ifdefined\IsLLNCS
\documentclass{llncs}
\else
\documentclass[11pt]{article}
\usepackage{fullpage}
\fi

\usepackage{amsfonts,amsmath,amssymb,boxedminipage,color,url,nccmath,tikz}
\usepackage[bottom]{footmisc} 
\usepackage[pageanchor=false,pdfstartview=FitH,colorlinks,linkcolor=blue,filecolor=blue,citecolor=blue,urlcolor=blue]{hyperref} 

\ifdefined\IsLLNCS\else
\usepackage{amsthm} 
\fi

\usepackage{booktabs}
\usepackage[T1]{fontenc}
\usepackage{lmodern} \normalfont 
\DeclareFontShape{T1}{lmr}{bx}{sc} { <-> ssub * cmr/bx/sc }{}
\usepackage{xcolor}
\usepackage{enumerate,paralist}
\usepackage{enumitem} 
\usepackage{varwidth}
\usepackage{aliascnt}
\usepackage[numbers]{natbib} 
\usepackage{cleveref} 
\usepackage{xspace}
\usepackage{xstring}
\usepackage[nottoc]{tocbibind} 
\usepackage[hypcap=false]{caption}
\usepackage{mathtools}
\usepackage[usestackEOL]{stackengine}
\setstackgap{L}{\normalbaselineskip}

\input{macros}

\title{Breaking the $O(\sqrt n)$-Bit Barrier:\\ Byzantine Agreement with Polylog Bits Per Party}

\ifdefined\IsAnon
    \author{}
    \date{}
    \LLNCS{\institute{}}
\else

    \ifdefined\IsLLNCS
        \author{Elette Boyle\inst{1}
        \and Ran Cohen\inst{2}
        \and Aarushi Goel\inst{3}
        }
        \institute{Reichman University and NTT Research\\ \email{elette.boyle@runi.ac.il}
        \and Reichman University\\ \email{cohenran@runi.ac.il}
        \and NTT Research\\ \email{aarushi.goel@ntt-research.com}
        }
    \else
        \author{Elette Boyle\thanks{Reichman University and NTT Research. E-mail: \texttt{elette.boyle@runi.ac.il}.}
        \and Ran Cohen\thanks{Reichman University. E-mail: \texttt{cohenran@runi.ac.il}. }
        \and Aarushi Goel\thanks{NTT Research. E-mail: \texttt{aarushi.goel@ntt-research.com}.}
        }
    \fi
\fi

\begin{document}
\sloppy

\maketitle
\thispagestyle{empty}

\input{abstract}

\ifdefined\IsLLNCS\else
\Tableofcontents
\fi

\input{Introduction}

\input{Preliminaries}

\input{SRDS}

\input{balanced_ba}

\input{SRDS_Constructions}

\input{relation-snargs}

\paragraph{Acknowledgements.}
E.\ Boyle’s research is supported in part by ISF grant 1861/16 and AFOSR Award FA9550-17-1-0069 and ERC project HSS (852952).
R.\ Cohen’s research was done in part while the author was at Northeastern University and supported by NSF grant 1646671.
A.\ Goel’s work was done in part while visiting the FACT Center at IDC Herzliya, Israel, and while the author was at Johns Hopkins University, supported in part by an NSF CNS grant 1814919, NSF CAREER award 1942789 and Johns Hopkins University Catalyst award.

{\small{
\ifdefined\IsLLNCS
\begingroup
\let\clearpage\relax
\bibliographystyle{abbrvnat}
\bibliography{crypto_lncs}
\endgroup
\else
\bibliographystyle{abbrvnat}
\bibliography{crypto}
\fi
}}

\appendix

\input{Preliminaries_cont}

\input{balanced_ba_cont}

\input{SRDS_cont}

\section{Connection with Succinct Arguments (Cont'd)}
In this section, we provide supplementary material for \cref{sec:succinct_arguments}.
In \cref{sec:sym-poly}, we prove \cref{thm:sym-poly-nphard}, and in \cref{sec:snarg-compliant} we formally define SNARG-compliant multi-signature schemes.
\input{sym-poly}
\input{general_multisigs}

\end{document}

%% file: macros.tex
\newcommand{\remove}[1]{}

\newcommand{\LLNCS}[1]{\ifdefined\IsLLNCS #1 \fi}

\newcommand{\TLLNCS}[2]{\ifdefined\IsLLNCS#1\else #2 \fi}

\ifdefined\IsDraft
\newcommand{\authnote}[2]{[{\color{red}\textbf{#1:}}~{\color{blue} #2}]}
\else
\newcommand{\authnote}[2]{}
\fi

\ifdefined\IsTrackChanges
\newcommand{\stkout}[1]{\ifmmode\text{\sout{\ensuremath{#1}}}\else\sout{#1}\fi}

\newcommand{\deleted}[2]{{\textbf{Deleted:}~{\color{red} \stkout{#2} }}}
\newcommand{\deletedtwo}[2]{{{\color{red} #2 }}}

\else

\newcommand{\deleted}[2]{}
\newcommand{\deletedtwo}[2]{}

\newcommand{\changedtwo}[2]{}

\fi

\newcommand{\Tableofcontents}{
\thispagestyle{empty}
\pagenumbering{gobble}
\clearpage
{\small{
\tableofcontents
}}
\thispagestyle{empty}
\clearpage
\pagenumbering{arabic}
}

\ifdefined\IsLLNCS
\newcommand{\QED}{\qed}
\else
\newcommand{\QED}{}
\fi

\ifdefined\IsLLNCS
\newcommand{\SUBSUBSEC}{.}
\else
\newcommand{\SUBSUBSEC}{}
\fi

\newcommand{\resp}{resp.,\xspace}
\newcommand{\ie}  {i.e.,\xspace}
\newcommand{\eg}  {e.g.,\xspace}

\newcommand{\wrt} {with respect to\xspace}

\newcommand{\ceil}[1]{\lceil #1 \rceil}
\newcommand{\sset}[1]{\{#1\}}

\newcommand{\ssize}[1]{|#1|}

\renewcommand{\Pr}{{\mathrm {Pr}}}
\newcommand{\pr}[1]{\Pr\left[#1\right]}

\newcommand{\characteristic}{\mathsf{\textsf{char}}}
\newcommand{\field}{{\FF}}

\newcommand{\GG}{{\mathbb{G}}}
\newcommand{\FF}{{\mathbb{F}}}
\newcommand{\NN}{{\mathbb{N}}}

\newcommand{\ZZ}{{\mathbb{Z}}}
\newcommand{\zo}{\{0,1\}}
\newcommand{\zn}{{\zo^n}}

\newcommand{\is}{{i^\ast}}

\newcommand{\js}{{j^\ast}}

\newcommand{\zs}{{z^\ast}}

\newcommand{\poly}{\mathsf{poly}}
\newcommand{\polylog}{\mathsf{polylog}}

\newcommand{\ptp}{{point-to-point}\xspace}

\newcommand{\func}[1]{\ensuremath{f_\mathsf{#1}}\xspace}

\newcommand{\faecomm}{\func{ae\mhyphen comm}}
\newcommand{\faggrsig}{\func{aggr\mhyphen sig}}

\newcommand{\fba}{\func{ba}}
\newcommand{\fct}{\func{ct}}
\newcommand{\fcrs}{\func{crs}}
\newcommand{\fpki}{\func{pki}}

\newcommand{\fagg}{\func{agg}}

\newcommand{\prot}[1]{\pi_{\ensuremath{\mathsf{#1}}}\xspace}
\newcommand{\protba}{\prot{ba}}

\newcommand{\party}{\mathsf{party}}
\renewcommand{\cref}{\Cref}

\TLLNCS{
\newaliascnt{claiml}{theorem}
\newtheorem{claiml}[claiml]{Claim}
\aliascntresetthe{claiml}

\renewenvironment{claim}{\begin{claiml}}{\end{claiml}}
}
{
\newtheorem{theorem}{Theorem}[section]

\newaliascnt{lemma}{theorem}
\newtheorem{lemma}[lemma]{Lemma}
\aliascntresetthe{lemma}

\newaliascnt{claim}{theorem}
\newtheorem{claim}[claim]{Claim}
\aliascntresetthe{claim}

\newaliascnt{corollary}{theorem}
\newtheorem{corollary}[corollary]{Corollary}
\aliascntresetthe{corollary}

\newaliascnt{proposition}{theorem}

\aliascntresetthe{proposition}

\newaliascnt{conjecture}{theorem}

\aliascntresetthe{conjecture}

\newaliascnt{definition}{theorem}
\newtheorem{definition}[definition]{Definition}
\aliascntresetthe{definition}

\newaliascnt{remark}{theorem}

\aliascntresetthe{remark}

\newaliascnt{example}{theorem}

\aliascntresetthe{example}
}

\newaliascnt{construction}{theorem}
\newtheorem{construction}[construction]{Construction}
\aliascntresetthe{construction}

\crefname{lemma}{Lemma}{Lemmas}
\crefname{figure}{Figure}{Figures}
\crefname{claim}{Claim}{Claims}
\crefname{corollary}{Corollary}{Corollaries}
\crefname{proposition}{Proposition}{Propositions}
\crefname{conjecture}{Conjecture}{Conjectures}
\crefname{definition}{Definition}{Definitions}
\crefname{remark}{Remark}{Remarks}
\crefname{exmaple}{Example}{Examples}

\crefname{equation}{Equation}{Equations}

\newaliascnt{proto}{theorem}

\aliascntresetthe{proto}
\crefname{proto}{protocol}{protocols}

\newaliascnt{algo}{theorem}

\aliascntresetthe{algo}
\crefname{algo}{algorithm}{algorithms}

\newaliascnt{expr}{theorem}

\aliascntresetthe{expr}
\crefname{experiment}{experiment}{experiments}

\newcommand \mycaption {\small }     
\newcommand \mylabel {}

\ifdefined\IsLLNCS
    \newenvironment{nfbox}[3]{
    \renewcommand \mycaption {#1}
    \renewcommand \mylabel {#2}
    \begin{center}
    \begin{small}
    \begin{tabular}{|ll|}
    \hline
    \hspace{.3ex}
    \begin{minipage}{.97\linewidth}
         \vspace{0.5ex}
         #3
         }
         {
         \vspace{-2ex}
         \captionof{figure}{\mycaption}
         \label{\mylabel}
     \end{minipage}
     &\hspace{.3ex} \\
     \hline
     \end{tabular}
     \end{small}
     \end{center}
    }
\else
    \newenvironment{nfbox}[3]{
    \renewcommand \mycaption {#1}
    \renewcommand \mylabel {#2}
    \begin{center}
    \begin{small}
    \begin{tabular}{|ll|}
    \hline
    \hspace{0.3ex}
    \begin{minipage}{.94\linewidth}
         \smallskip
         #3
         }
         {\smallskip
    \end{minipage}
    &\hspace{-0.3ex} \\
    \hline
    \end{tabular}
   \captionsetup{type=figure}
   \captionof{figure}{\mycaption}
    \label{\mylabel}
    \end{small}
    \end{center}
    }
\fi

\newcommand{\srds}{succinctly reconstructed distributed signatures\xspace}
\newcommand{\Srds}{Succinctly reconstructed distributed signatures\xspace}
\newcommand{\SRDS}{Succinctly Reconstructed Distributed Signatures\xspace}

\newcommand{\PCDsigma}{\ensuremath{\sigma_{\mathsf{pcd}}}\xspace}
\newcommand{\PCDtau}{\ensuremath{\tau_{\mathsf{pcd}}}\xspace}
\newcommand{\PCDGen}{\ensuremath{\mathsf{PCD.Gen}}\xspace}
\newcommand{\PCDProve}{\ensuremath{\mathsf{PCD.Prover}}\xspace}
\newcommand{\PCDVerify}{\ensuremath{\mathsf{PCD.Verify}}\xspace}
\newcommand{\dcgen}{\ensuremath{\mathsf{DC_{\mathsf{Gen}}}}\xspace}

\newcommand{\MSSetup}{\ensuremath{\mathsf{MS.Setup}}\xspace}
\newcommand{\MSGen}{\ensuremath{\mathsf{MS.KeyGen}}\xspace}
\newcommand{\MSSign}{\ensuremath{\mathsf{MS.Sign}}\xspace}
\newcommand{\MSVerify}{\ensuremath{\mathsf{MS.Verify}}\xspace}
\newcommand{\MSCombine}{\ensuremath{\mathsf{MS.Combine}}\xspace}
\newcommand{\MSMverify}{\ensuremath{\mathsf{MS.MVerify}}\xspace}
\newcommand{\MSVerifyAgg}{\ensuremath{\MSVerify_\mathsf{agg\mhyphen key}}\xspace}
\newcommand{\MSSignDeg}{\ensuremath{\MSSign_\mathsf{deg\mhyphen key}}\xspace}
\newcommand{\MSVerifyInv}{\ensuremath{\mathsf{MS.MVerifyInv}}\xspace}

\newcommand{\TSSetup}{\ensuremath{\mathsf{Setup}}\xspace}
\newcommand{\TSCR}{\ensuremath{\mathsf{Setup}}\xspace}
\newcommand{\TSGen}{\ensuremath{\mathsf{KeyGen}}\xspace}

\newcommand{\TSVerify}{\ensuremath{\mathsf{Verify}}\xspace}
\newcommand{\TSSignShare}{\ensuremath{\mathsf{Sign}}\xspace}

\newcommand{\TSAggr}{\ensuremath{\mathsf{Aggregate}}\xspace}

\newcommand{\DSGen}{\ensuremath{\mathsf{DS.KeyGen}}\xspace}
\newcommand{\DSOGen}{\ensuremath{\mathsf{DS.OKeyGen}}\xspace}
\newcommand{\DSSign}{\ensuremath{\mathsf{DS.Sign}}\xspace}
\newcommand{\DSVerify}{\ensuremath{\mathsf{DS.Verify}}\xspace}

\newcommand{\naive}{na\"{i}ve\xspace}

\newcommand{\ProofSetup}{\ensuremath{\mathsf{S.Setup}}\xspace}
\newcommand{\ProofProve}{\ensuremath{\mathsf{S.Prove}}\xspace}
\newcommand{\ProofVerify}{\ensuremath{\mathsf{S.Verify}}\xspace}
\newcommand{\ProofGen}{\ensuremath{\mathsf{Proof}_{\mathsf{Gen}}}\xspace}

\newcommand{\vect}[1]{{ \boldsymbol #1}}

\newcommand{\Adv}{{\ensuremath{\mathcal{A}}}\xspace}

\newcommand{\Sim}{{\ensuremath{\cal S}}\xspace}

\newcommand{\Party}{{\ensuremath{P}}\xspace}
\newcommand{\secParam}{{\ensuremath{\kappa}}\xspace}
\newcommand{\aux}{{\mathsf{aux}}\xspace}
\newcommand{\negl}{\mathsf{negl}}
\newcommand{\Inv}{\mathsf{Inv}}
\newcommand{\PInv}{\mathsf{PInv}}

\newcommand{\leafnode}{{v}\xspace}

\newcommand{\inputlen}{{\ensuremath{\ell_\mathsf{in}}}\xspace}
\newcommand{\outputlen}{{\ensuremath{\ell_\mathsf{out}}}\xspace}

\newcommand{\heads}{{\ensuremath{\mathsf{heads}}}\xspace}

\newcommand{\Exp}{{\ensuremath{\mathbb{E}}}\xspace}

\newcommand{\LS}{{\ensuremath{\mathcal{L}}}\xspace}
\newcommand{\BS}{{\ensuremath{\mathcal{B}}}\xspace}
\newcommand{\CS}{{\ensuremath{\mathcal{C}}}\xspace}
\newcommand{\DS}{{\ensuremath{\mathcal{D}}}\xspace}
\newcommand{\E}{{\ensuremath{\mathcal{E}}}\xspace}
\newcommand{\FS}{{\ensuremath{\mathcal{F}}}\xspace}

\newcommand{\HS}{{\ensuremath{\mathcal{H}}}\xspace}
\newcommand{\IS}{{\ensuremath{\mathcal{I}}}\xspace}
\newcommand{\JS}{{\ensuremath{\mathcal{J}}}\xspace}

\newcommand{\MS}{{\ensuremath{\mathcal{M}}}\xspace}
\newcommand{\PS}{{\ensuremath{\mathcal{P}}}\xspace}

\newcommand{\XS}{{\ensuremath{\mathcal{X}}}\xspace}

\newcommand{\validset}{\ensuremath{S_{\mathsf{val}}}\xspace}

\newcommand{\verifyset}{{\ensuremath{S_{\mathsf{ver}}}}\xspace}
\newcommand{\compliance}{{\ensuremath{\mathtt{C}}}\xspace}

\newcommand{\yesD}{{\ensuremath{\mathcal{D}_\yes}}\xspace}
\newcommand{\noD}{{\ensuremath{\mathcal{D}_\no}}\xspace}

\newcommand{\yes}{{\mathsf{yes}}\xspace}
\newcommand{\no}{{\mathsf{no}}\xspace}
\newcommand{\indexMax}{{\ensuremath{\mathsf{max}}}\xspace}
\newcommand{\indexMin}{{\ensuremath{\mathsf{min}}}\xspace}

\newcommand{\goodP}{{\ensuremath{\mathcal{N}}}\xspace}
\newcommand{\setP}{{\ensuremath{\mathcal{S}}}\xspace}

\newcommand{\init}{\mathsf{init}}
\newcommand{\crs}{\mathsf{crs}}
\newcommand{\pp}{\mathsf{pp}}
\newcommand{\vk}{\mathsf{vk}}
\newcommand{\ovk}{\mathsf{ovk}}
\newcommand{\sk}{\mathsf{sk}}
\newcommand{\sig}{\mathsf{sig}}
\newcommand{\skdeg}{\sk_\mathsf{deg}}
\newcommand{\vkdeg}{\vk_\mathsf{deg}}
\newcommand{\vkagg}{\vk_\mathsf{agg}}
\newcommand{\signsig}{\mathsf{sig}}

\newcommand{\rroot}{\mathsf{root}}
\newcommand{\height}{\mathsf{height}}
\newcommand{\setsig}{S_{\mathsf{sig}}}

\newcommand{\Expt}{\ensuremath{\mathsf{Expt}}\xspace}
\newcommand{\ExptTSrobust}{\Expt^\mathsf{robust}\xspace}
\newcommand{\ExptTSforge}{\Expt^\mathsf{forge}\xspace}

\mathchardef\mhyphen="2D

\newcommand{\mode}{\mathsf{mode}\xspace}
\newcommand{\bbpki}{\mathsf{b\mhyphen pki}\xspace}
\newcommand{\trpki}{\mathsf{tr\mhyphen pki}\xspace}

\newcommand{\iith}[1] {$#1$\textsuperscript{th}\xspace}
\newcommand{\ith}           {\iith{i}}
\newcommand{\jth}           {\iith{j}}
\newcommand{\kth}           {\iith{k}}

\newcommand{\child}{\ensuremath{\mathsf{child}}}
\newcommand{\parent}{\ensuremath{\mathsf{parent}}}
\newcommand{\graph}{G\xspace}
\newcommand{\tree}{T\xspace}
\newcommand{\vertex}{V\xspace}
\newcommand{\edge}{E\xspace}
\newcommand{\inputvar}{\textsf{in}}
\newcommand{\outputvar}{\textsf{out}}
\newcommand{\commtree}{\tree}
\newcommand{\vertexCommtree}{\vertex}
\newcommand{\edgeCommtree}{\edge}

\newcommand{\trans}{\ensuremath{\mathsf{trans}}}
\newcommand{\linp}{\ensuremath{\mathsf{linp}}}
\newcommand{\data}{\ensuremath{\mathsf{data}}}
\newcommand{\prove}{\ensuremath{\mathsf{proof}}}
\newcommand{\inputs}{\ensuremath{\mathsf{inputs}}}
\newcommand{\subsetsize}{\ensuremath{s}}

\newcommand{\progP}{\ensuremath{\mathsf{P}}}

\newcommand{\ms}{\mathsf{ms}}
\newcommand{\sigmams}{\sigma_\ms}
\newcommand{\ppms}{\pp_\ms}
\newcommand{\XSms}{\XS_\ms}
\newcommand{\ppsrds}{\pp_{\mathsf{srds}}}
\newcommand{\fmsg}{f_{\mathsf{msg}}}
\newcommand{\hash}{H}
\newcommand{\sibling}{\mathsf{sibling}}
\newcommand{\msetup}{\mathsf{Merkle.Setup}}
\newcommand{\mhash}{\mathsf{Merkle.Hash}}
\newcommand{\mproof}{\mathsf{Merkle.Proof}}
\newcommand{\mverify}{\mathsf{Merkle.Verify}}
\newcommand{\seed}{\mathsf{seed}}

\newcommand{\mapping}{\mathsf{idmap}}
\newcommand{\Range}{\mathsf{range}}

\newcommand{\plus}{\raisebox{.4\height}{\scalebox{.6}{+}}}

%% file: abstract.tex

\begin{abstract}
\emph{Byzantine agreement} (BA), the task of $n$ parties to agree on one of their input bits in the face of malicious agents, is a powerful primitive that lies at the core of a vast range of distributed protocols. Interestingly, in BA protocols with the best overall communication, the demands of the parties are highly \emph{unbalanced}: the amortized cost is $\tilde O(1)$ bits per party, but some parties must send $\Omega(n)$ bits.
In best known \emph{balanced} protocols, the overall communication is sub-optimal, with each party communicating $\tilde O(\sqrt{n})$.

In this work, we ask whether asymmetry is inherent for optimizing total communication. In particular, is BA possible where \emph{each} party communicates only $\tilde O(1)$ bits? Our contributions in this line are as follows:

\begin{itemize}[leftmargin=*]
\item
We define a cryptographic primitive---\emph{succinctly reconstructed distributed signatures} (SRDS)---that suffices for constructing $\tilde O(1)$ balanced BA. We provide two constructions of SRDS from different cryptographic and Public-Key Infrastructure (PKI) assumptions.

\item
The SRDS-based BA follows a paradigm of boosting from ``almost-everywhere'' agreement to full agreement, and does so in a single round. Complementarily, we prove that PKI setup and cryptographic assumptions are necessary for such protocols in which every party sends $o(n)$ messages.

\item
We further explore connections between a natural approach toward attaining SRDS and average-case succinct non-interactive argument systems (SNARGs) for a particular type of NP-Complete problems (generalizing Subset-Sum and Subset-Product).
\end{itemize}

\noindent
Our results provide new approaches forward, as well as limitations and barriers, towards minimizing per-party communication of BA. In particular, we construct the first two BA protocols with $\tilde O(1)$ balanced communication, offering a tradeoff between setup and cryptographic assumptions, and answering an open question presented by King and Saia (DISC'09).
\end{abstract}

%% file: Introduction.tex

\section{Introduction}

The problem of \emph{Byzantine agreement (BA)}~\cite{PSL80,LSP82} asks for a set of $n$ parties to agree on one of their input bits, even facing malicious corruptions. BA is a surprisingly powerful primitive that lies at the core of virtually every interactive protocol tolerating malicious adversaries, ranging from other types of consensus primitives such as broadcast~\cite{PSL80,LSP82} and blockchain protocols (\eg \cite{CM19}), to secure multiparty computation (MPC) \cite{Yao82,GMW87,BGW88,CCD88,RB89}.
In this work, we study BA in a standard context, where a potentially large set of $n$ parties runs the protocol within a synchronous network, and security is guaranteed facing a constant fraction of statically corrupted parties.

Understanding the required communication complexity of BA as a function of $n$ is the subject of a rich line of research.
For the relaxed goal of \emph{almost-everywhere agreement}~\cite{DPPU88}, \ie agreement of all but $o(1)$ fraction of the parties, the full picture is essentially understood. The influential work of \citet{KSSV06} showed a solution roughly ideal in every dimension: in which each party speaks to $\tilde O(1)$ other parties (\ie polylog degree of communication graph, a.k.a.\ communication \emph{locality}~\cite{BGT13}), and communicates a total of $\tilde O(1)$ bits throughout the protocol, in $\tilde O(1)$ rounds;\footnote{We follow the standard practice in large-scale cryptographic protocols, where $\tilde{O}$ hides polynomial factors in $\log{n}$ and in the security parameter $\secParam$, see \eg \cite{DI06,DIKNS08}.}
further, the solution does not require cryptographic and/or trusted setup assumptions and is given in the full-information model.
The main challenge in BA thus becomes extending almost-everywhere to full agreement.

In this regime, our current knowledge becomes surprisingly disconnected.
While it is known how to employ cryptography and setup assumptions to compute BA with $\tilde{O}(1)$ locality~\cite{BGT13,CCGGOZ15,BCDH18}, the number of \emph{bits} that must be communicated by each party is large, $\Omega(n)$.\footnote{In fact, the constructions in \cite{BGT13,CCGGOZ15,BCDH18} are for MPC protocols that enable secure computation of any function with $\tilde{O}(1)$ locality; these protocols are defined over point-to-point networks, and so also provide a solution for the specific task of BA.}
BA with \emph{amortized} $\tilde{O}(1)$ per-party communication (and computation) can be achieved~\cite{BGH13,CM19,ACDNPRS19}; however, the structure of these protocols is wildly unbalanced: with some parties who must each communicate with $\Theta(n)$ parties and send $\Omega(n)$ bits.
The existence of ``central parties'' who communicate a large amount facilitates fast convergence in these protocols.
When optimizing per-party communication, the best BA solutions degrade to $\tilde\Theta(\sqrt n)$ bits/party, with suboptimal $\tilde O(n^{3/2})$ overall communication~\cite{KS11,KLST11}.

This intriguing  gap leads us to the core question studied in this paper: Is such an imbalance inherent?
More specifically:
\begin{quote}
	\centering
	\emph{Is it possible to achieve Byzantine agreement with \emph{(balanced)} \\ per-party communication of $\tilde{O}(1)$?}
\end{quote}	

Before addressing our results, it is beneficial to consider the relevant lower bounds.
It is well known that any \emph{deterministic} BA protocol requires $\Omega(n^2)$ communication~\cite{DR85} (and furthermore, the connectivity of the underlying communication graph must be $\Omega(n)$~\cite{DOlev82,FLM86}).
This result extends to randomized BA protocols, in the special case of very \emph{strong adversarial} (adaptive, strongly rushing\footnote{A \emph{strongly rushing} adversary in~\cite{ACDNPRS19} can adaptively corrupt a party that has sent a message $m$ and replace the message with another $m'$, as long as no honest party received $m$.}) capabilities~\cite{ACDNPRS19}.
Most closely related is the lower bound of \citet{HKK08}, who showed that without trusted setup assumptions, at least one party must send $\Omega(\sqrt[3]{n})$ messages.\footnote{The lower bound in \cite{HKK08} easily extends to a public setup such as a common reference string.} But, the bound in \cite{HKK08} applies only to a restricted setting of protocols with \emph{static message filtering}, where every party decides on the set of parties it will listen to before the beginning of each round (as a function of its internal view at the end of the previous round).
We note that while the almost-everywhere agreement protocol in~\cite{KSSV06} falls into the static-filtering model, all other scalable BA protocols mentioned above crucially rely on \emph{dynamic message filtering} (which is based on incoming messages' content).
This leaves the feasibility question open.

\subsection{Our Results}\label{sec:intro:results}

We perform an in-depth investigation of boosting from almost-everywhere to full agreement with $\tilde O(1)$ communication per party.
Motivated by the $\tilde O(1)$-locality protocol of \citet*{BGT13}, we first achieve an intermediate step of \emph{certified almost-everywhere agreement}, where almost all of the parties reach agreement, and, in addition, hold a certificate for the agreed value. \citet{BGT13} showed how to boost certified almost-everywhere agreement to full agreement in a single round, where every party communicates with $\tilde{O}(1)$ parties.

Our initial observation is that the protocol from \cite{BGT13} achieves low communication aside from one expensive piece: the distributed generation of the certificate, which is of size $\Theta(n)$, and its dissemination.
We thus target this step and explore.

Our contributions can be summarized as follows.

\begin{itemize} [leftmargin=*]
\item \textbf{SRDS and balanced BA.}
We define a minimal cryptographic primitive whose existence implies $\tilde O(1)$ balanced BA: \emph{succinctly reconstructed distributed signatures} (SRDS).

We provide two constructions of SRDS, each based on a different flavor of a public-key infrastructure (PKI): (1) from one-way functions in a ``trusted-PKI'' model, and (2) from collision-resistant hash functions (CRH) and a strong form of succinct non-interactive arguments of knowledge (SNARKs)\footnote{A SNARK~\cite{Micali94,BCCGLRT17} is a proof system that enables a prover holding a witness $w$ to some public NP statement $x$ to convince a verifier that it indeed knows $w$ by sending a single message. The proof string is succinct in the sense that it is much shorter than the witness $w$, and knowledge is formalized via an efficient extractor that succeeds extracting $w$ from a malicious prover $P^*$ with roughly the same probability that $P^*$ convinces an honest verifier.} in a model with a ``bare PKI'' and a common random string (CRS). Roughly, trusted-PKI setup assumes that parties' keys are generated properly, whereas bare PKI further supports the case where corrupt parties may generate keys maliciously. We elaborate on the difference between the PKI models in \cref{sec:intro:technique}.

\item \textbf{Necessity of setup for one-shot ``boost.''}
Our SRDS-based BA follows a paradigm of boosting from almost-everywhere to full agreement, and does so in a single communication round.
Complementarily, we prove two lower bounds for any such one-shot boost in which every party sends $o(n)$ messages. The first shows that some form of PKI (or stronger setup, such as correlated randomness\footnote{In the \emph{correlated-randomness} model a trusted dealer samples $n$ secret strings from a joint distribution and delivers to each party its corresponding secret string, \eg a setup for threshold signatures.}) is \emph{necessary} for this task. The second shows that given only PKI setup (as opposed to stronger, correlated-randomness setup), then \emph{computational assumptions} (namely, at least one-way functions) are additionally required.

In contrast to prior lower bounds (\eg \cite{HKK08,ACDNPRS19}), this holds even against a static adversary, and where parties can exercise dynamic filtering (\ie without placing limitations on how parties can select to whom to listen).

\item \textbf{Connections to succinct arguments.}
We further explore connections between a natural approach toward attaining SRDS in weaker PKI models and \emph{average-case} succinct non-interactive argument (SNARG) systems\footnote{Similarly to a SNARK, a SNARG allows a prover holding a witness $w$ to some public NP statement $x$ to convince a verifier that $x$ belongs to the language; however, as opposed to a SNARK, here the prover does not prove that it knows $w$ (only that such a witness exists), hence there is no requirement to extract the witness from a cheating prover.} for a particular type of NP-Complete problems (generalizing Subset-Sum and Subset-Product). This can be interpreted as a barrier toward this approach for constructing SRDS without heavy ``SNARG-like'' tools.
\end{itemize}

\noindent
Collectively, our results provide an initial mapping for the feasibility landscape of BA with $\tilde O(1)$ per-party communication, including new approaches forward, as well as limitations and barriers.
Our approach yields two BA protocols with $\tilde O(1)$ communication per party, offering a tradeoff between the setup assumptions and the cryptographic assumptions.
These results answer an open question presented by \citet{KS09}, asking whether cryptography can be used to construct BA with $o(\sqrt{n})$ communication per party.
Our BA results are summarized in \cref{tbl:intro:ae_to_e} alongside other almost-everywhere to everywhere agreement protocols.

\input{tbl_ae_to_e}

\subsection{Technical Overview}\label{sec:intro:technique}

We now proceed to present our results in greater detail.

\paragraph{\Srds.}
Our first contribution is identifying and formalizing a cryptographic primitive that enables boosting from almost-everywhere agreement to full agreement on a value, with low per-party communication.

The primitive---\emph{\srds} (SRDS)---is a new type of a distributed signature scheme, with a natural motivation: allowing a set of parties to jointly produce a signature on some message $m$, which can serve as a succinct certificate for proving that a \emph{majority} of the parties agree on $m$.
Interestingly, this task does not seem to be attained by existing distributed signature notions, such as \emph{multi-signatures}~\cite{IN83}, \emph{aggregate signatures}~\cite{BGLS03}, or \emph{threshold signatures}~\cite{DF89}. For example, while multi-signatures (and, similarly, aggregate signatures) can succinctly combine signatures of many parties, to \emph{verify} the signature, the (length-$\Theta(n)$!)\ vector of contributing-parties identities must also be communicated.\footnote{Indeed, the verification algorithm of multi-signatures (and aggregate signatures) must receive the set of parties who signed the message. This is precisely the culprit for the large $\tilde \Theta(n)$ per-party communication within the low-locality protocol of~\cite{BGT13}.}
As discussed in the related-work section (\cref{sec:relatedWork}), threshold signatures are implied by SRDS but also do not suffice: while identities of the signers are no longer needed to verify a combined signature, this information is necessary to \emph{reconstruct} the combined signature in the first place (even within specific existing schemes, \eg \cite{GJKR01,Boldyreva03}).
We provide a more detailed comparison to different signature notions in \cref{sec:relatedWork}.

An SRDS scheme is based on a PKI for signatures, where every party is set with a secret signing key and a public verification key.\footnote{As mentioned, we will distinguish between a \emph{bare PKI}, where every party locally chooses its keys and corrupted parties can set their keys as a function of all verification keys (and any additional public information), and a \emph{trusted PKI}, which is honestly generated (either locally or by a trusted party) and where corrupted parties cannot change their verification keys. See further discussion below.}
The parties may receive additional setup information that may contain, for example, public parameters for the signature scheme or a common random string (CRS), depending on the actual construction.
Given a message $m$, every party can locally generate a signature on $m$, and signatures on the same message can be succinctly aggregated into a new signature.
The new aspect is that given a combined signature and a message $m$, it is possible to verify whether is was aggregated from a ``large'' number of ``base'' signatures on $m$, and both aggregation and verification can be done \emph{succinctly}.

Three properties are required from an SRDS scheme:
\emph{robustness} means that an adversary cannot prevent the honest parties from generating an accepting signature on a message; \emph{unforgeability} prevents an adversary controlling a minority from forging a signature; and \emph{succinctness} requires that the ``final'' signature (including \emph{all} information needed for verification) is short (of size $\tilde{O}(1)$) and can be incrementally reconstructed from ``base'' signatures in small batches of size $\polylog(n)$.\footnote{$\polylog(n)$ denotes $\log^c(n)$ for some constant $c>1$.}
An SRDS scheme is \emph{$t$-secure} if it satisfies the above properties even facing $t$ colluding adversarial parties.

\paragraph{Balanced BA from SRDS.}
We demonstrate how to attain $\tilde O(1)$-balanced BA against $\beta n$ corruptions (for $\beta<1/3$) given black-box access to any $\beta n$-secure SRDS scheme.
We begin by presenting a distilled version of the ``certified almost-everywhere agreement'' approach from~\cite{BGT13} that we tailor for Byzantine agreement, where only correctness matters and privacy is not required.\footnote{The focus of \cite{BGT13} was on MPC and required stronger assumptions and additional rounds; in particular, a \naive use of their MPC protocol \emph{cannot} lead to communication-balanced BA as it requires all parties to send information to a designated $\polylog(n)$-size set, the so-called \emph{supreme committee}.}
\begin{enumerate}[leftmargin=*]
\item
The parties execute the almost-everywhere agreement protocol of King et al.~\cite{KSSV06}; this establishes a $\polylog(n)$-degree communication tree (which is essentially a sparse overlay network) in which each node is assigned with a committee of $\polylog(n)$ parties. The guarantees are that the $\polylog(n)$-size \emph{supreme committee} (\ie the committee assigned to the root) has a $2/3$ honest majority and almost all of the parties are connected to the supreme committee via the communication tree.
\item
The supreme committee executes a BA protocol on their inputs to agree on the output $y$, and, in addition, runs a coin-tossing protocol to agree on a random seed $s$. Next, the supreme committee propagates the pair $(y,s)$ to \emph{almost} all of the parties.
\item
Once a party receives the pair $(y,s)$, the party signs it (in~\cite{BGT13}, using a multi-signature scheme), and sends the signature back to the supreme committee that aggregates all the signatures. The aggregated signature attesting to $(y,s)$ is then distributed to \emph{almost} all of the parties.
\end{enumerate}

Once this form of \emph{certified} almost-everywhere agreement on $(y,s)$ is reached, full agreement can be obtained in one round. Every party $\Party_i$ that receives the signed pair $(y,s)$, evaluates a pseudorandom function (PRF) on the seed $s$ and its identity $i$ to determine a set of (sufficiently random) $\polylog(n)$ parties, and sends the signed $(y,s)$ to every party in that set. A party that receives such a signed pair, can verify that a majority of the parties agree on $(y,s)$ (by the guarantees of multi-signatures) and that it was supposed to receive a message from the sender (by evaluating the PRF on $s$ and the sender's identity). In this case, it can output $y$ and halt.

The protocol from~\cite{BGT13} achieves $\tilde{O}(1)$ locality. However, recall that even though the size of a multi-signature might itself be ``small,'' the verification algorithm additionally requires a list of contributing parties, where the description size of this list will need to be proportional to $n$. Hence, the effective size of the aggregated signature, and thus per-party communication, is stuck at $\Theta(n)$.

\medskip
At this point the new notion of SRDS comes into the picture.
We use the \emph{succinctness} property of SRDS combined with the communication tree established by the protocol from~\cite{KSSV06} to bound the size of the aggregated signatures by $\tilde{O}(1)$.
In essence, the parties aggregate the signatures in a recursive manner up the communication tree such that in each step at most $\polylog(n)$ signatures are aggregated.

This technique introduces additional subtleties that must be addressed. For example, since the partially aggregated signature can no longer afford to describe the set of contributing parties, it is essential to make sure that the same ``base'' signature is not aggregated multiple times (this may allow the adversary to achieve more influence on the final aggregated signature than its proportional fraction of ``base'' signatures).
To address this, we assign $\polylog(n)$ (virtual) identities to every party, one identity for each path from that party to the supreme committee in the communication tree; this ensures that the fraction of signatures that are generated by corrupted parties is equal to the corruption threshold. We refer the reader to \cref{sec:ba_from_srds} for more details.

\begin{theorem}[balanced BA, informal]\label{thm:intro:ba}
Let $\beta<1/3$ be a constant. Assuming the existence of $\beta n$-secure SRDS, there exists an $n$-party, $\beta n$-resilient BA protocol that terminates after $\polylog(n)$ rounds, and where every party communicates $\polylog(n)\cdot\poly(\secParam)$ bits.
\end{theorem}

We note that our BA protocol is the first to establish a $\polylog(n)$-degree communication graph where \emph{every} party has an ``honest path'' to a $2/3$-honest committee, such that the per-party communication required for establishing it is $\tilde{O}(1)$. Thus, we can obtain the following corollaries.

\begin{corollary}[informal] \label{cor:intro:mpc}
Let $\beta<1/3$ be a constant. Assuming the existence of $\beta n$-secure SRDS:
\begin{enumerate}[leftmargin=*]
    \item {\bf Broadcast}:
    There exists a $\beta n$-resilient 1-bit broadcast protocol such that $\ell$ protocol executions (potentially with different senders) require $\ell\cdot\polylog(n)\cdot\poly(\secParam)$ bits of communication per party.
    \item {\bf MPC}:
    Assuming fully homomorphic encryption, a function $f:(\zo^\inputlen)^n\to\zo^\outputlen$ can be securely computed with guaranteed output delivery tolerating a static, malicious $\beta n$-adversary, such that the total communication complexity (of all parties) is $n\cdot\polylog(n)\cdot\poly(\secParam)\cdot(\inputlen+\outputlen)$ bits.
\end{enumerate}
\end{corollary}

One remark regarding the corruption model is in place.
In this work, we consider \emph{static} adversaries that choose the set of corrupted parties before the beginning of the protocol. As mentioned above, our constructions are based on some form of trusted setup, which, as we prove below, is necessary. We emphasize that (as standard) we avoid trivialized settings, \eg where the trusted setup determines a $\polylog(n)$-degree communication tree for achieving \emph{full} agreement,\footnote{If the communication tree is sampled \emph{after} the corruptions have been fixed and is given to the parties as setup, then with overwhelming probability all nodes are \emph{good} (i.e., with more than two-thirds honest majority in the parties assigned). This ensures that the path from all leaf nodes to the root only contains good nodes and hence all parties can communicate with the committee assigned to the root node without disruptions to reach full agreement.} by considering the adversarial model where the adversary can corrupt the parties adaptively during the setup phase \emph{given} the setup information of the corrupted parties and any public setup information. During the online phase, the adversary is static and cannot corrupt additional parties.

\paragraph{Constructing SRDS.}
We present two constructions of SRDS, offering a tradeoff between setup assumptions and cryptographic assumptions.

Our first construction is influenced by the ``sortition approach'' of Algorand~\cite{CM19} and merely requires one-way functions (OWF); however, the public-key infrastructure (PKI) is assumed to be \emph{honestly} generated (either by the parties themselves or by an external trusted third party), and corrupted parties cannot alter their keys. The construction is based on digital signatures augmented with an oblivious key-generation algorithm for sampling a verification key without knowing the corresponding signing key.\footnote{We note that standard signatures can be used if we strengthen the model assumptions, \eg by assuming that a party can securely erase its signature key, or by considering a trusted party that only provides the verification keys to some parties. We opted not to rely on stronger model assumption since we can establish signatures with oblivious key generation from the minimal assumption of one-way functions.} Lamport's signatures~\cite{Lamport79}, which are based on OWF, can easily be adjusted to support this property.
To establish the PKI, every party decides whether to generate its public verification key obliviously or together with a signing key by tossing a biased coin, such that with overwhelming probability all but $\polylog(n)$ keys are generated obliviously.
Since those with the ability to sign are determined at random (as part of the trusted PKI), only parties who hold a signing key can sign messages. The
oblivious key-generation algorithm ensures that an adversary who
only sees a list of verification keys, cannot distinguish between the
keys that have a corresponding signing key and ones that do not.
As a result, even if the adversary chooses the set of corrupt parties
after the keys are sampled, with a high probability, the fraction
of honest parties will be preserved in the signing subset. SRDS
signature-aggregation is done by concatenation, and verification of
an SRDS signature requires counting how many valid signatures
were signed on the message.

It would be desirable to reduce the trust assumption in establishing the PKI, \eg by using verifiable pseudorandom functions (VRF)~\cite{MRV99} as done in~\cite{CM19}.
However, this approach~\cite{CM19} is defined within a blockchain model where a fresh random string (the hash of the recent block) is assumed to be consistently available to all parties later in the protocol and serves as the seed for the sortition;
equivalently, that parties have access to a common random string (CRS) \emph{independent} of corrupted parties' public keys.
Without this extra model assumption, their VRF approach does not apply.
We note that several recent consensus protocols \cite{ACDNPRS19,CPS19,CHMOS19,CPS20,CKS20,BKLL20,WXSD20,WXDS20} also follow the sortition approach of~\cite{CM19}; however, similar to our first construction, their PKI is assumed to be honestly generated by a trusted third party.

\begin{theorem}[SRDS from OWF and trusted PKI, informal]\label{thm:intro:srds_owf}
Let $\beta<1/3$ and assume that OWF exist. Then, there exists a $\beta n$-secure SRDS in the trusted-PKI model.
\end{theorem}

Our second construction is based on a weaker bare-PKI setup, in which each party locally computes its signature keys, and the adversary can corrupt parties and change their keys as a function of honest parties' public keys.
To illustrate the underlying ideas, consider a simplified case where all of the nodes in the almost-everywhere communication tree are honest (each node will essentially be realized by some committee of parties). A \naive construction would be to have all parties sign the message and send the signature to their respective leaf nodes. Every leaf node would then count the number of verified signatures received and send the message and the counter to its parent. In a recursive way, every node would simply add the counters received from its children and send it to its parent. At the end of this process, the root node would get a final count of the total number of verified signatures. This approach completely breaks down, however, if even one node is not honest as it can lie about its count. To enforce an honest behavior of the nodes, we need to make sure that the aggregation is done in a verifiable way, \ie ensure that bad nodes send a valid count of the number of signatures aggregated so far.

Toward this, our first idea is to require each node to attach a ``succinct proof'' of honest behavior to their messages. In particular, in addition to the message $m$ and count $c$ that a leaf node sends to its parent, it must also send a proof to convince the parent that it knows $c$ distinct signatures on the message $m$. Similarly, every node must prove that they
received sufficiently many valid proofs backing up its count. To verify, it is sufficient to check at the root node, whether sufficiently many ``base'' signatures were aggregated. This approach requires proof systems that support \emph{recursive composition};
for this reason, we use \emph{proof-carrying data} (PCD) systems \cite{CT10}.

A PCD system extends the notion of SNARKs to the distributed setting by allowing recursive composition in a succinct way.
Informally, every party can generate a succinct proof on some statement, certifying that it satisfies a given local property with respect to its private input and previously received messages (statements and their proofs).
\citet{BCCT13} proved that PCD systems for logarithmic-depth DAGs exist assuming SNARKs with \emph{linear extraction}, \ie where the size of the extractor is linear in the size of the prover.\footnote{We note that although SNARKs with linear extraction are a stronger assumption than standard SNARKs (with polynomial extraction), standard SNARKs techniques do not separate the two notions.} Extractability assumptions of this kind have been previously considered in, \eg \cite{Val08,DFH12,GS14,BJPY18}.
Since PCD systems allow for propagation of information up a communication tree in a succinct and publicly verifiable way, they seem to exactly capture our requirements for SRDS.

This simple idea, however, is vulnerable to an adversary that generates a valid-looking aggregate signature by using multiple copies of the same signature. Indeed, since the partially aggregated signature must be succinct, the parties cannot afford to keep track of \emph{which} base signatures were already incorporated, leaving them vulnerable to a repeat occurrence.
We protect against such an attack by encoding additional information in the partially aggregated signatures using collision-resistant hash functions (CRH). We refer the reader to \cref{sec:SRDS:snarks} for the detailed solution.

\begin{theorem}[SRDS from CRH, SNARKs and bare PKI, informal]\label{thm:intro:srds_snarks}
Let $t<n/3$ and assume that CRH and SNARKs with linear extraction exist. Then, there exists a $t$-secure SRDS in the CRS and bare-PKI model.
\end{theorem}

\paragraph{Necessity of PKI for single-round boost of almost-everywhere agreement.}
Our SRDS-based BA protocol (\cref{thm:intro:ba}) shows how to boost almost-everywhere agreement to full agreement in a single round with small communication. Both our constructions crucially rely on a public-key infrastructure (PKI) that enables each party to publish its verification key on a bulletin board. We show that this setup assumption is necessary for this task. That is, given \emph{only} public setup---\ie the common reference string (CRS) model---this task is not possible.

We note that the lower bound of \citet{HKK08} does not translate to our setting, as it considers \emph{static} message filtering, where every party chooses to whom to listen in a given round based on its view prior to that round (and then may perform additional sanity checks on the incoming messages from these parties to ensure they are not malformed).
The lower bound in~\cite{HKK08} shows that \emph{dynamic} filtering, \ie where in any given round, parties may decide to accept or ignore a message from other parties based on the incoming messages received in that round, is required (at least in the CRS model).
We present the first such lower bound in the dynamic-filtering model.

\begin{theorem}[no single-shot boost in CRS model, informal]\label{thm:intro:lb_crs}
There is no single-round protocol from almost-everywhere to everywhere agreement in the CRS model where every party sends sublinear (\ie $o(n)$) many messages.
\end{theorem}

Recall that almost-everywhere agreement guarantees that all parties agree on the common output aside from a $o(n)$-size set of \emph{isolated} parties, whose identities are unknown to the remaining honest parties.
In the setting of static filtering, one can prove continued isolation of these parties for any low-communication protocol in a relatively clean manner~\cite{HKK08}: The probability that an honest party $\Party_i$ will send messages to an honest isolated $\Party_j$ is \emph{independent} of the event that $\Party_j$ will choose to process messages from $\Party_i$ in this round, thus placing a birthday-type bound on information successfully being conveyed.
With dynamic filtering, however, $\Party_j$ may process messages \emph{dependent} on some property of this message, \eg whether it contains particular authentication, which may only be contained in honest messages.\footnote{In general, message filtering should be via a simple and ``light'' test, \eg counting how many messages arrived, or verifying a signature. We refer to \cite{BCDH18} for a discussion on message filtering in protocols over incomplete graphs.}
In such case, there is strong bias toward accepting honest messages, and one must work harder to ensure that isolated parties do not reach agreement.

At a high level, the idea of our lower bound is to make a linear set of corrupted parties emulate the role of isolated parties during the first part of the protocol (reaching almost-everywhere agreement). This way, the honest parties cannot distinguish between isolated honest parties and faking corrupted parties, and must attempt to communicate the output value to all such parties. However, if each honest party only sends a sublinear number of messages, then with a very high probability, most isolated honest parties (and faking corrupted parties) only receive messages from a sublinear number of non-isolated parties in the last round. The adversary can use this fact to keep an isolated honest party confused in the following sense. In the last round of an execution with preagreement on 0 (\resp on 1), the adversary sends to this party messages corresponding to an execution with preagreement on 1 (\resp on 0). Without private-coin setup such as PKI, an isolated party cannot distinguish between honest messages in the real execution and fake messages from the simulation.

To carry out this attack, we need to show that there exist parties who receive messages from a ``small'' set of neighbors in both scenarios: when all parties start with input $0$ and when all start with input $1$ (otherwise, the adversary may not have a sufficient corruption budget for the attack). Before the protocols begins, the adversary decides on the set of parties to corrupt by first emulating in its head two executions, one with preagreement on $0$ and the other with preagreement on $1$, where the same linear-size set of parties act as isolated parties. We use a counting argument to show that there exist isolated parties who receive messages from a sublinear set of neighbors in both executions. The adversary targets one of these parties to attack and corrupts all ``simulated isolated parties'' except for the targeted one, along with the pair of neighbor-sets who communicate with the targeted party in the simulation.
We refer the reader to Section~\ref{sec:ba_lb_crs} for a formal description and analysis of the attack.

\paragraph{On the different PKI models.}
As discussed above, SRDS implies a single-round boost of almost-everywhere to full agreement, which in turn (by \cref{thm:intro:lb_crs}) requires some form of private-coin setup. Given this, one of our goals is to minimize the trust assumptions in the setup phase. Our SNARK-based construction offers the minimal setup requirement---a bare PKI---where every party locally generates its own signature keys and publishes the verification key on a bulletin board. The adversary can adaptively corrupt parties and change their keys as a function of all the public setup information (including the honest parties' verification keys and the CRS, in case it exists). This is the prevalent PKI model that has appeared in, \eg \cite{Canetti04,KK06,KL11,CSV16}.

Our OWF-based construction, on the other hand, assumes an honestly generated PKI, where the adversary cannot alter the corrupted parties' keys. Such a setup assumption is normally captured by a trusted party who samples the keys for all the parties, and provides each party with its secret key as well as all public keys; see, \eg \cite{LLR06,ACDNPRS19,CPS19,CHMOS19,CPS20}.
We note that our trusted-PKI setup is \emph{weaker} than a full-blown trusted party in two aspects: First, the distribution from which the trusted party samples the values is a \emph{product distribution}, \ie parties' keys are independent; this is weaker than a general correlated randomness setup (as used, \eg in \cite{ADDNR19} for threshold-signatures setup). Second, we consider \emph{public-coin} sampling in the sense that the sampling coins are revealed to the corresponding party (\ie intermediate key-generation values are not kept hidden). In fact, one can consider the model where every party honestly generates and publishes its public key, and corrupted parties can deviate from the protocol only in the online phase.

\paragraph{Necessity of OWF for single-round boost in PKI model.}
\cref{thm:intro:lb_crs} states the necessity of private-coin setup for single-round protocols (from almost-everywhere agreement to full agreement) where every party sends $o(n)$ messages. In the PKI model, where the public/private keys of each party are independently generated, we further prove that \emph{cryptographic assumptions} are necessary. Intuitively, if one-way functions (OWF) do not exist, an adversary can invert the PKI algorithm with noticeable probability to find a pre-image for each public key. In this case, the adversary can carry out the attack for the CRS model, discussed above.

\begin{theorem}[OWF needed for single-shot boost in PKI model, informal]\label{thm:intro:lb_pki}
If OWF do not exist, there is no single-round protocol from almost-everywhere to everywhere agreement in the trusted PKI model where every party sends sublinear many messages.
\end{theorem}

We note that this lower bound does not extend to more complex private-coin setups, where the parties receive correlated secret strings that are jointly sampled from some distribution, \eg setup for information-theoretic signatures. Indeed, given such a setup it is possible to boost almost-everywhere to everywhere agreement in a single round with information-theoretic security and where every party sends $\polylog$ many messages (albeit, each of size $\Omega(n)$)~\cite{BCDH18}.
The reason that the proof approach of \cref{thm:intro:lb_pki} does not apply in this case is that when the private keys of two honest parties are  correlated, it is unclear how an (even computationally unbounded) adversary that only receives partial information about this correlation can consistently invert the setup information and impersonate honest parties. We leave the feasibility of single-round boost protocols from almost-everywhere to everywhere in the correlated-randomness model, in which every party sends sublinear many \emph{bits} (as opposed to messages), as an interesting open question.

\paragraph{Connection to succinct arguments.}
Our SRDS construction from CRH and SNARKs works with minimal setup requirements, but relies on relatively undesirable cryptographic assumptions (in particular, SNARKs are a non-falsifiable assumption~\cite{GW11}). On the other hand, our construction from one-way functions uses light computational assumptions, but (as with many other works in this area, \eg \cite{ACDNPRS19,CPS19,CHMOS19,CPS20}) requires a stronger assumption of trusted PKI. A clear goal is to obtain SRDS from better computational assumptions within a better setup model, ultimately reducing to bare PKI, or even more fine-grained intermediate models such as \emph{registered PKI}\footnote{In the registered PKI, every party can arbitrarily choose its public key (just like in bare PKI), but in order to publish it, the party must prove knowledge of a corresponding secret key.} (see \cite{Boldyreva03,LOSSW13} and a discussion in \cite{BN06}). A natural approach toward doing so is to build upon one of the closest existing relatives within this setting: \emph{multi-signatures}.

Recall that multi-signatures \emph{almost} provide the required properties of SRDS in this setting, in that they support succinct aggregation of signatures, with the sole issue that multi-signature verification requires knowledge of the set of parties who contributed to it---information that requires $\Theta(n)$ bits to describe. Multi-signatures have been constructed from (standard) falsifiable assumptions in the registered-PKI model, \eg \cite{LOSSW13}.
A natural approach toward constructing SRDS within this model is thus to simply augment a multi-signature scheme with some method of succinctly convincing the verifier that a given multi-signature is composed of signatures from sufficiently many parties. We demonstrate challenges toward such an approach, by showing that in some cases this \emph{necessitates} a form of succinct non-interactive arguments.

More specifically, we observe that asserting approval of a multi-signature by sufficiently many parties is inherently equivalent to asserting existence of a large subset of parties $S\subseteq[n]$, such that their corresponding verification keys $\sset{\vk_i}_{i\in S}$ satisfy a given function-target relation $f_{\sigma,m}(\sset{\vk_i}_{i\in S})=1$. (Here $m$ is the message, $\sigma$ is the multi-signature and $f_{\sigma,m}$ is a function that is derived from the multi-signature verification function.) Such a task can be viewed as a class of ``Subset-$f$'' problems on the verification keys $\vk_1,\ldots,\vk_n$, capturing as special cases the standard Subset-Sum and Subset-Product problems with functions $f_{\Sigma}(\sset{x_i}_{i\in S})=\sum_{i\in S} x_i$ and $f_{\Pi}(\sset{x_i}_{i\in S})=\prod_{i\in S} x_i$, respectively.

Considering even a generous setup model of trusted PKI, where parties' verification keys $\vk_1,\ldots,\vk_n$ are generated independently and honestly, the Subset-$f$ problem begins taking the form of problems where we do not know (or even possibly believe) that the witness $S\subseteq[n]$ can be compressed to $o(n)$ bits. As we show, an SRDS of this form implies a type of average-case non-interactive argument for asserting membership in Subset-$f$, with \emph{succinct} proof size: namely, a form of succinct non-interactive argument (SNARG) as in \cite{BCCT13}, with \emph{average-case} soundness guarantees. Although this average-case notion does not directly fall within the negative results of \citet{GW11}, it appears to be a powerful notion, which may be interpreted as a barrier toward this approach to SRDS construction without SNARG-like tools.

Motivated by this, we explore hardness of the Subset-$f$ problem for more general classes of functions $f$. We show that over rings with appropriate structure (namely, Hadamard product), NP-hardness results for Subset-Sum and Subset-Product can be extended to include (worst-case) Subset-$\phi$ for all elementary symmetric polynomials $\phi$.

We make explicit the above connection for the multi-signature scheme of \citet{LOSSW13} (LOSSW) in relation to (average-case) Subset-Product, and extend to multi-signature schemes of appropriate structure in relation to the Subset-$\phi$ problem for elementary symmetric polynomials $\phi$. The reduction leverages homomorphism, where a combined signature for a set of parties on message $m$ in the multi-signature scheme corresponds to a valid single-party signature with respect to a specific joint function of the parties' verification keys $\vk_i$; for LOSSW, their product $\vk^* = \prod_i \vk_i$.

\begin{theorem}[SRDS from multi-signatures requires average-case SNARGs, informal]
Any SRDS based on the LOSSW~\cite{LOSSW13} multi-signature scheme in a natural way (as we define) implies the existence of succinct non-interactive arguments for average-case Subset-Product.
This extends to a more general class of multi-signature schemes and Subset-$\phi$ for elementary symmetric polynomials.
\end{theorem}

At a high level, the reduction interprets a (random) Subset-Product instance with target $(x_1,\dots,x_n,t)$ as a set of uniform \emph{verification keys} $(\vk_1,\dots,\vk_n,\vk_{n+1}=t^{-1})$ for the multi-signature scheme.
Given a satisfying witness $S \subseteq [n]$ with $\prod_{i \in S} x_i = t$ of appropriate size, this translates to knowing a large subset of verification keys for which generating an SRDS on their behalf can be achieved efficiently with respect to the \emph{degenerate} verification key $\vk^* = \prod_{i \in S} \vk_i \cdot t^{-1} = 1$.
On the other hand, for uniformly sampled keys \emph{without} such an embedded trapdoor subset, forging such an SRDS will be hard.

We refer the reader to \cref{sec:succinct_arguments} for formal definitions of these notions (including SRDS ``based on'' a multi-signature scheme and average-case SNARGs), as well as further discussion and details of our claims and proofs.

\subsection{Additional Related Work}\label{sec:relatedWork}

\paragraph{Distributed signatures.}
Distributed signatures come in many flavors. We compare SRDS to existing notions from the literature.

\emph{Threshold signatures} \cite{DF89,Rabin98,GJKR01,Sho00,Boldyreva03,HAP18} can guarantee that a sufficiently large number of parties signed the message, while keeping the signature-length (including all information needed to verify) independent of $n$. However, threshold signatures require the keys to be generated by a trusted party in a correlated way (\eg as a Shamir sharing of the signing key), and the signature-reconstruction protocol of existing schemes does not offer succinct aggregation in ``small'' batches. SRDS imply threshold signatures by having the setup algorithm produce the PKI for the parties, and using the aggregation algorithm to reconstruct a signature.

We note that \citet{LJY16} constructed \emph{fully distributed} threshold signatures that do not require any setup assumptions. However, this scheme is not applicable in our setting, since it requires an interactive key-generation protocol to generate the public and secret keys, and this protocol in turn uses a broadcast channel. In fact, as indicated by our lower bound, some form of private-coin setup is inherently needed for constructing SRDS.

\emph{Multi-signatures} \cite{IN83,MOR01,Boldyreva03,BN06,LOSSW13,BDN18} guarantee that a subset of parties signed the message.
Unlike threshold signatures, correlated trusted setup is not needed and a bare PKI suffices; in addition, some of the constructions enable succinct aggregation in ``small'' batches.
\emph{Aggregate signatures} \cite{BGLS03,LMRS04,BGOY07,LOSSW13,LLY13,HW18} are a similar primitive that allows aggregating signatures on different messages. The main distinction of SRDS is succinctness that enables verification \emph{without} knowing the signing parties. This property is crucial for our BA protocol construction.

\emph{Group signatures}~\cite{CH91} and \emph{ring signatures}~\cite{RST01} allow any individual party to sign a message on behalf of a set while hiding their identity.  This is different than our setting where we need to prove that a majority of the parties signed the message.

\paragraph{Large-scale MPC.}
The focus of this work is communication complexity of Byzantine agreement protocols; however, \cref{cor:intro:mpc} demonstrates applications \wrt general secure multiparty computation (MPC).
Large-scale MPC was initially studied by \citet{DI06} and successors (\eg \cite{DN07,DIKNS08,DIK10}), in the sense that the amortized per-party work grows only as $\tilde O(|C|/n + \poly(n))$, where $C$ is the circuit to be computed.
\citet{DKMSZ17} applied the almost-everywhere agreement protocol~\cite{KSSV06} to achieve MPC with amortized per-party communication of $\tilde{O}(\ssize{C}/n+\sqrt{n})$.
Using cryptographic assumptions (threshold FHE), \citet{ZMS14} reduced the amortized cost to $\tilde{O}(\ssize{C}/n)$. Under comparable assumptions, our results achieve amortized cost of $\tilde{O}(\inputlen+\outputlen)$ (where $\inputlen$ and $\outputlen$ stand for the function's input/output length).

The \emph{bottleneck} complexity of MPC was studied in \cite{BJPY18}, as the maximum communication complexity required by any party within the protocol execution. It was shown that for some $n$-party functions $f: \zo^n \to \zo$, some parties must communicate $\Omega(n)$ bits to compute $f$, even if security is not required. This result rules out generic MPC with \emph{balanced}, sublinear communication per party, and motivates our MPC results of amortized sublinear communication per party. Note that in \cite{BCP15} load-balanced MPC was achieved, however, amortized over large programs (and in a model that allows each party to have a single use of a broadcast channel).

\paragraph{Communication-efficient BA.}
Known protocols that break the $\Omega(n^2)$ communication barrier from~\cite{DR85} (for deterministic protocols) follow one of two paradigms. The first is starting with the almost-everywhere agreement protocol of \cite{KSSV06} and boost it to full agreement; this approach includes \cite{KS09,KS11,KLST11,BGH13,BGT13}, as well as our results. The second is based on the sortition approach from Algorand~\cite{CM19}, where only a ``small'' set of parties are allowed to talk in every round, and includes \cite{CM19,ACDNPRS19}. The latter approach inherently leads to unbalanced protocols, since parties that are eligible to talk send messages to all other parties.

We note that while \cite{KS09,KLST11,BGH13,BGT13} and our results hold in the static-corruptions setting, some protocols are resilient to \emph{adaptive} corruptions.
Assuming secure data erasures (\ie where honest parties can erase some parts of their internal states) $\tilde{O}(\sqrt{n})$-balanced BA \cite{KS11} and $\tilde{O}(1)$-amortized BA~\cite{CM19} can be achieved against adaptive corruptions. In the erasure-free setting, \cite{ACDNPRS19} achieved $\tilde{O}(1)$-amortized BA against adaptive corruptions.
One of the interesting open questions we pose in \cref{sec:intro:questions} is whether $\tilde{O}(1)$-balanced BA can be achieved in the adaptive setting.

\subsection{Open Questions}\label{sec:intro:questions}

Our results leave open several interesting questions for followup work.

Our constructions of SRDS offer a trade-off between cryptographic assumptions and setup assumptions (indeed, our lower bound indicates that some form of private-coin setup is needed). Is it possible to get the best of both, \ie construct SRDS with bare PKI under standard, falsifiable assumptions? This in turn would imply $\tilde O(1)$-balanced BA from the corresponding computational assumption and setup. Alternatively, does SRDS in a weak setup model \emph{require} strong computational assumptions: For example, do SRDS with bare PKI \emph{imply} some kind of succinct non-interactive arguments (SNARGs)?

Taking a step back: Is it possible to achieve $\tilde O(1)$-balanced BA \emph{unconditionally}? While our SRDS-based approach inherently makes use of computational assumptions (and our lower bound implies this necessity for a \emph{one-shot} boost from almost-everywhere to everywhere agreement in the PKI model), this leaves open the possibility of removing cryptography via an alternative approach.

Can one further extend the lower bound in this work, identifying a minimal required \emph{round complexity} for generically converting from almost-everywhere to everywhere agreement within various setup models?

In the $\tilde{O}(1)$-amortized BA setting, known constructions consider stronger security models. Namely, the protocol in \citet{BGH13} is secure against static corruptions (similarly to our protocols); however, no trusted setup assumptions are required. The protocol of \citet{ACDNPRS19} guarantees security against adaptive corruptions; however, it requires a trusted PKI assumption. In contrast, the protocol of \citet{KS11} does not require setup assumptions and is resilient to adaptive corruptions, but it provides suboptimal total communication $\tilde{O}(n\sqrt{n})$. It is interesting to explore if $\tilde{O}(1)$-balanced BA can be achieved without setup or in the adaptive setting.

Regarding the communication model, the vast majority of sub-quadratic BA protocol are defined in the synchronous model. In the asynchronous setting unbalanced sub-quadratic BA in the trusted PKI model was recently proposed \cite{CKS20,BKLL20}. We note that \emph{balanced} sub-quadratic BA is not known even in the partially synchronous model. An interesting question is to expand our techniques beyond the synchronous realm.

Finally, all known BA protocols with $o(n^2)$ total communication follow either the approach of \citet{KSSV06} or of \citet{CM19}, which are based on electing a polylog-size committee. As such, these protocols only support a non-optimal constant fraction of corruptions. Is it possible to achieve $o(n^2)$ total communication while tolerating the optimal number of corruptions $t<n/2$?

\subsection*{Paper Organization}
In \cref{sec:Preliminaries}, we provide basic definitions. SRDS are defined in \cref{sec:SRDS}. Our BA protocol and the lower bounds appear in \cref{sec:app:BA}. \cref{sec:SRDS:construction} presents two constructions of SRDS, and in \cref{sec:succinct_arguments}, we explore the connection of SRDS based on multi-signatures to succinct non-interactive arguments.
Some of the definitions and proofs are deferred to the appendix.

%% file: tbl_ae_to_e.tex
\renewcommand{\arraystretch}{1.2}

\begin{table}[!htb]
\footnotesize{
\centering
\begin{tabular}{llllllll}
    \toprule
    \textbf{protocol} &
    \textbf{rounds} &
    \Centerstack[l]{\textbf{max com.}\\\textbf{per party}} &
    \textbf{setup} &
    \Centerstack[l]{\textbf{cryptographic}\\\textbf{assumptions}}  &
    \Centerstack[l]{\textbf{message}\\\textbf{filtering}}  &
    \textbf{corrupt.}&
    \textbf{remark}\\

    \midrule

    HKK'08~\cite{HKK08} &  & $\Omega(\sqrt[3]{n})$ & crs & & static & static & lower bound\\
    KS'09~\cite{KS09} & $O(1)$ & $\tilde{O}(n \cdot\sqrt{n})$ & - & - & dynamic & static & \\
    KS'11~\cite{KS11} & $\polylog(n)$ & $\tilde{O}(\sqrt{n})$ & - & - & dynamic & adaptive & \\
    KLST'11~\cite{KLST11} & $\polylog(n)$ & $\tilde{O}(\sqrt{n})$ & - & - & dynamic & static & \\
    BGH'13~\cite{BGH13} & $O(1)$ & $\tilde{O}(n)$ & - & - & dynamic & static & \\
    BGT'13~\cite{BGT13} & $1$ & $\tilde{O}(n)$ & pki & owf & dynamic & static & \\
    CM'19~\cite{CM19}$^{\dag}$ & exp $O(1)$ & $\tilde{O}(n)$ & trusted-pki & RO\plus unique-sig & dynamic & adaptive & \\
    ACD$^+$'19~\cite{ACDNPRS19}$^{\dag}$ & exp $O(1)$ & $\tilde{O}(n)$ & trusted-pki & bilinear maps & dynamic & adaptive & \\
    CKS'20~\cite{CKS20}$^{\dag}$ & exp $O(1)$ & $\tilde{O}(n)$ & trusted-pki & vrf & dynamic & adaptive & asynchronous\\
    BKLL'20~\cite{BKLL20}$^{\dag}$ & exp $O(1)$ & $\tilde{O}(n)$ & trusted-pki & fhe\plus nizk & dynamic & adaptive & asynchronous\\

    \midrule

     & $1$ & $\Omega(n)$ & crs & & dynamic & static & lower bound\\
    \textbf{This work} & $1$ & $\tilde{O}(1)$ & pki\plus crs & snarks$^*$\plus crh & dynamic & static & \\
     & $1$ & $\tilde{O}(1)$ & trusted pki & owf & dynamic & static & \\
    \bottomrule
\end{tabular}
\caption \footnotesize {Comparison of protocols boosting from almost-everywhere to full agreement, tolerating $(1/3-\epsilon)\cdot n$ corruptions. The $\tilde{O}$ notation hides polynomial terms in the security parameter $\secParam$ and in $\log{n}$. \emph{crs} stand for a common random string, \emph{pki} stands for bare pki, and \emph{trusted pki} stands for honestly generated pki. By \emph{snarks$^*$} we refer to SNARKs with linear extraction, \ie where the size of the extractor is linear in the size of the prover.
\emph{RO} stands for random oracle and \emph{unique-sig} for unique signatures. \emph{vrf} stand for verifiable pseudorandom functions, \emph{fhe} for fully homomorphic encryption, and \emph{nizk} for non-interactive zero-knowledge proofs.
Static corruptions are done before the protocol begins but can be a function of the trusted setup; adaptive corruptions can occur during the course of the protocol.
$^{(\dag)}$ The protocols from \cite{CM19,ACDNPRS19,CKS20,BKLL20} reach agreement from scratch (hence also from almost-everywhere agreement) with amortized $\tilde{O}(1)$ communication per party; the expected round complexity is constant and termination is guaranteed in $\polylog(n)$ rounds. \emph{Static message filtering} requires honest parties to decide on the parties they will listen to and process their messages before the beginning of each round, whereas \emph{dynamic message filtering} allows this decision to be done during the round and depending on the content of the incoming messages.
}
\label{tbl:intro:ae_to_e}
}
\end{table}

%% file: Preliminaries.tex
\section{Preliminaries}\label{sec:Preliminaries}

In this section, we present the security model and the definition of Byzantine agreement. Additional definitions of proof-carrying data systems, of Merkle hash proof systems, and of multi-signatures, can be found in \cref{sec:Preliminaries_cont}.

\paragraph{Protocols.}
All protocols considered in this paper are PPT (probabilistic polynomial time): the running time of every party is polynomial in the (common) security parameter, given as a unary string. For simplicity, we consider Boolean-input Boolean-output protocols, where apart from the common security parameter, each party has a single input bit, and each of the honest parties outputs a single bit. We note that our protocols can be used for agreement on longer strings, with an additional dependency of the communication complexity on the input-string length.

We consider protocols in the PKI model, and we distinguish between two flavors of PKI: a \emph{trusted PKI} and a \emph{bare PKI}. In both settings, a trusted party samples a secret key $\sk_i$ and a public key $\vk_i$, for every $i\in[n]$, from some distribution. The adversary is allowed to corrupt parties adaptively based on $(\vk_1,\ldots,\vk_n)$ and learn the secret key associated with every corrupted party. In the bare-PKI model, the adversary can replace the public key of every corrupted party by an arbitrary string $\vk_i'$ of its choice, even ``after'' looking at the public keys of honest parties.

The communication model is \emph{synchronous}, meaning that protocols proceed in rounds. In each round, every party can send a message to every other party over a private channel.\footnote{We note that using standard techniques, our constructions can be defined over \emph{authenticated} channels (that do not provide privacy) by additionally assuming the existence of key-agreement protocols.} It is guaranteed that every message sent in a round will arrive at its destinations by the end of that round. The adversary is \emph{rushing} in the sense that it can use the messages received by corrupted parties from honest parties in a given round to determine the corrupted parties' messages for that round.

Note that an adversary can always blow up the communication complexity of a protocol by flooding honest parties with many bogus messages. It is therefore standard to count only messages that are actually processed by honest parties. We follow the model of \cite{BCDH18}, who formalized this intuition. Namely, message receival consists of two phases: a \emph{filtering} phase, where incoming messages are inspected according to specific filtering rules defined by the protocol, and some messages may be discarded, followed by a \emph{processing} phase, where the party computes its next-message function based on the remaining non-filtered messages. In practice, the filtering procedure should be ``lightweight,'' and consist of operations like counting messages or verifying validity of a signature.

\input{BA_def}

%% file: BA_def.tex
\paragraph{Byzantine Agreement.}
Informally, in an $n$-party, $t$-resilient Byzantine agreement protocol, the honest parties must agree on one of their input bits, even when $t$ parties collude and actively try to prevent it. We provide two definitions for BA: the first is the standard, property-based definition and the second is based on the real/ideal paradigm.

We start with the property-based definition. This definition captures the core properties required for consensus; namely, \emph{agreement}, \emph{validity}, and \emph{termination}. This is a weaker definition than the simulation-based one, and as such is most suitable for proving lower bounds.

\begin{definition}[BA, property-based]\label{def:ba:properties}
Let $\pi$ be an $n$-party protocol in which every party $\Party_i$ has an input bit $x_i\in\zo$ and outputs a bit $y_i\in\zo$ at the end of the protocol. The protocol $\pi$ is an $n$-party, $t$-resilient BA protocol (according to the property-based definition) if the following properties are satisfied with all but negligible probability when up to $t$ parties maliciously attack the protocol:
\begin{itemize}
    \item\textbf{Agreement.}
    For every pair of honest parties $\Party_i$ and $\Party_j$ it holds that $y_i=y_j$.
    \item\textbf{Validity.}
    If there exists a bit $x$ such that for every honest party $\Party_i$ it holds that $x_i=x$, then for every honest party $\Party_i$ it holds that $y_i=x$.
    \item\textbf{Termination.}
    Every honest party eventually outputs a bit.
\end{itemize}
\end{definition}

We proceed with the simulation-based definition, which requires the protocol to realize an ideal BA functionality.
Roughly speaking, an ideal functionality represents a trusted third party that receives inputs from all the parties and provides them with the correct output. A protocol in the real world (where parties communicate between themselves and no trusted party exists) securely realizes an ideal functionality if every attack in the real protocol can be simulated in the ideal-world computation; since by definition no attack can happen in the ideal world, it is concluded that no attack can happen in the real world as well.
We refer the reader to \cite{Canetti00,Canetti01,Goldreich04} for further details on the real/ideal paradigm. We follow the standard ideal functionality for Byzantine agreement, see, \eg \cite{Nielsen02,Cohen16,CGHZ16,CCGZ19,CCGZ21}, where the functionality collects the inputs from all parties and counts them; in case of a strong majority of the inputs equals some bit, then this bit is set as the output, and otherwise the adversary gets to choose the output.
This definition implies the property-based one and is stronger as it guarantees security under composition; we use this definition for our protocol constructions.

\input{func_ba}

\begin{definition}[BA, simulation-based]\label{def:ba:simulation}
An $n$-party, $t$-resilient Byzantine agreement protocol (according to the simulation-based definition) is a protocol $\pi$ that realizes the BA ideal functionality (defined in \cref{fig:fba}) tolerating a malicious adversary statically corrupting up to $t$ parties.
\end{definition}

\paragraph{Balanced BA.}
In this work, we design a balanced Byzantine agreement protocol, with $\tilde{O}(1)$ per-party communication; i.e., for all adversarial strategies, the communication complexity incurred by each honest party is $\polylog(n)\cdot\poly(\secParam)$ with overwhelming probability. One can also consider a slightly relaxed variant, where the per-party communication is $\tilde{O}(1)$ \emph{in expectation}. The relaxed notion is satisfied by committee-based BA protocols such as \cite{CM19,ACDNPRS19}, where each party has a similar probability of being elected in the committee and hence the parties incur a similar per-party communication when given sufficiently many invocations of the protocol. In this work, however, we focus on the former stronger notion that is formalized as follows.

\begin{definition}[Balanced BA]\label{def:balanced-BA}
Let $\alpha(n,\secParam)$ be a function. An $n$-party, $t$-resilient Byzantine agreement protocol has $\alpha(n,\secParam)$-balanced communication complexity, if for every PPT adversary corrupting up to $t$ parties, the communication complexity incurred by each honest party is $\alpha(n,\secParam)$, except for negligible probability.
\end{definition}

%% file: func_ba.tex
\begin{nfbox}{The Byzantine agreement functionality}{fig:fba}
\begin{center}
	\textbf{The functionality} $\fba$
\end{center}
The functionality $\fba$ proceeds as follows, running with parties $\Party_1, \ldots, \Party_n$ and an adversary $\Sim$, statically corrupting a subset of parties indexed by a set $\IS\subseteq[n]$ of size $\ssize{\IS}\leq t$.
\begin{enumerate}
	\item
    For each $i\in[n]\setminus\IS$, party $\Party_i$ sends a bit $x_i\in\zo$ as its input; the functionality sends $(\Party_i,x_i)_{i\in[n]\setminus\IS}$ to the adversary.
    \item
    The adversary sends input bits $(x_i)_{i\in\IS}$ for the corrupted parties and a ``tie-breaker bit'' $\tilde{x}$.
	\item
    Once all parties provided their inputs, if there exists a bit $b$ such that $\ssize{\sset{i\mid x_i=b}}\geq n-t$, then set $y=b$; otherwise, set $y=\tilde{x}$. Send $y$ to every party.
\end{enumerate}
\end{nfbox}

%% file: SRDS.tex
\section{\SRDS}\label{sec:SRDS}

In this section, we introduce a new notion of a distributed signature scheme for $n$ parties, which can be used to obtain low-communication BA. As discussed earlier, every party has signing/verification keys based on some form of PKI, and the parties may receive additional setup information consisting of public parameters for the underlying signature scheme and potentially a common random string (CRS).
We allow the adversary to adaptively corrupt a subset of the parties before the protocol begins, based on the setup information and all $n$ verification keys. We consider two PKI models: a \emph{bare PKI}, where the adversary can choose the corrupted parties' keys, and a \emph{trusted PKI}, where the keys are honestly generated and cannot be changed. We do not permit adaptive corruptions once the parties start signing messages.

Below, we define the new signature scheme and the security requirements.
Later, in \cref{sec:SRDS:construction}, we present two constructions offering a tradeoff between cryptographic and setup assumptions: the first assuming one-way functions in the trusted-PKI model and the second assuming CRH and SNARKs with linear extraction in the CRS and bare-PKI model.
In \cref{sec:succinct_arguments}, we show that a natural approach towards constructing SRDS in an untrusted-PKI model has strong connections to succinct average-case argument systems for certain NP-Complete problems.

\input{SRDS_def}

%% file: SRDS_def.tex
\ifdefined\IsTrackChanges\else
\vspace{-0.3cm}
\fi
\paragraph{The Definition.}\label{sec:SRDS:def}

We start by presenting the syntax of the definition, and later, define the required properties from the scheme: succinctness, robustness, and unforgeability.

\begin{definition}[SRDS syntax]\label{def:SRDS}
A \textsf{\srds scheme} with message space $\MS$ and signature space $\XS$ for a set of parties $\PS=\sset{\Party_1,\ldots,\Party_n}$, is defined by a quintuple of PPT algorithms $(\TSCR$, $\TSGen$, $\TSSignShare$, $\TSAggr$, $\TSVerify)$ as follows:
\begin{itemize}[leftmargin=*]
	\item
    $\TSCR(1^\secParam, 1^n)\to\pp$: On input the security parameter $\secParam$ and the number of parties $n$, the setup algorithm outputs public parameters $\pp$.
	\item
    $\TSGen(\pp)\to(\vk,\sk)$: On input the public parameters $\pp$, the key-generation algorithm outputs a verification key $\vk$ and a signing key $\sk$.
	\item
    $\TSSignShare(\pp, i, \sk, m)\to\sigma$: On input the public parameters $\pp$, the signer's identity $i$, a signing key $\sk$, and a message $m\in \MS$, the signing algorithm outputs a signature $\sigma\in\XS\cup\sset{\bot}$.
    \item
    $\TSAggr(\pp, \sset{\vk_1,\ldots,\vk_n}, m, \sset{\sigma_1,\ldots,\sigma_q})\to \sigma$: On input the public parameters $\pp$, the set of all verification keys $\sset{\vk_i}_{i\in [n]}$, a message $m\in\MS$, and a set of signatures $\sset{\sigma_i}_{i\in [q]}$ for some $q=\poly(n)$, the aggregation algorithm outputs a signature $\sigma\in\XS\cup\sset{\bot}$.
    \item
    $\TSVerify(\pp,\sset{\vk_1,\ldots,\vk_n}, m, \sigma)\to b$: On input the public parameters $\pp$, the set of all verification keys $\sset{\vk_i}_{i\in [n]}$, a message $m\in \MS$, a signature $\sigma\in\XS$, the verification algorithm outputs a bit $b\in\zo$, representing accept or reject.
\end{itemize}
\end{definition}

We assume without loss of generality that each signature encodes the index $i$ of the corresponding verification key $\vk_i$, and each aggregated signature encodes information about the maxima and minima of the indices associated with verification keys corresponding to the base signatures that are aggregated within them.\footnote{In \cref{sec:SRDS:construction} we will show how this property can be achieved by each of our constructions.} Given a base signature/aggregated signature, let $\indexMax(\sigma)$ denote the function that extracts the maxima associated with $\sigma$ and $\indexMin(\sigma)$ denote the function that extracts the maxima associated with $\sigma$ (in case of a base signature, both $\indexMax(\sigma)$ and $\indexMin(\sigma)$ will return the same value).

\paragraph{Remark (Notation $n$).}
Here, we use $n$ to denote the number of parties in the SRDS scheme. Looking ahead, the effective number of (virtual) parties in the SRDS used in our BA protocol in \cref{sec:app:BA} will be larger than the actual (real) participants of the protocol.

\medskip
We proceed to define three properties of an SRDS scheme: \emph{succinctness}, \emph{robustness}, and \emph{unforgeability}.
We define these properties with respect to any $t<n/3$ corruptions. Although the definitions can be stated for $t < n/2$, we opted for the former for clarity and concreteness, as both our BA protocol (\cref{sec:app:BA}) and our SRDS constructions (\cref{sec:SRDS:construction}) support $n/3$ corruptions.

\paragraph{Succinctness.}
We require that the size of each signature is $\tilde{O}(1)$. This holds both for signatures in the support of $\TSSignShare$ and of $\TSAggr$. In order for parties to jointly perform the signature aggregation process with low communication, we also require that the aggregate algorithm can be decomposed into two algorithms $\TSAggr_1$ and $\TSAggr_2$.
Depending on the set of input signatures $\sset{\sigma_i}_{i\in [q]}$ and the verification keys, the first algorithm $\TSAggr_1$ \emph{deterministically} outputs a subset of the signatures $\setsig$.
The second (possibly randomized) algorithm $\TSAggr_2$ then aggregates these signatures without relying on the verification keys. In particular, the input to the (possibly randomized) step $\TSAggr_2$ is short.

Looking ahead at the BA protocol in \cref{sec:ba_from_srds}, subsets of the parties will collectively run the aggregation algorithm.
Although the inputs to the aggregation algorithm need not be kept private, it could be the case that the randomness used should remain secret. For this reason, the computation of $\TSAggr_2$ in the BA construction will be carried out using an MPC protocol; to keep the overall communication of every party $\tilde{O}(1)$, we require the circuit size representing $\TSAggr_2$ to be $\tilde{O}(1)$. The goal of $\TSAggr_1$ is to deterministically filter out invalid inputs (using the verification keys), such that $\TSAggr_2$ only depends on the verified signatures and \emph{not} on the $n$ verification keys (otherwise the circuit size will be too large).

\begin{definition}[succinctness]
An $n$-party SRDS scheme is \textsf{succinct} if it satisfies the following:

\begin{enumerate}[leftmargin=*]
	\item \textbf{Size of signatures:}
    There exists $\alpha(n,\secParam)\in\poly(\log{n},\secParam)$ such that $\XS\subseteq\zo^{\alpha(n,\secParam)}$.
	\item \textbf{Decomposability:} The $\TSAggr$ algorithm can be decomposed into 2 algorithms $\TSAggr_1$ and $\TSAggr_2$, such that the following hold:
	\begin{itemize}[leftmargin=*]
		\item $\TSAggr_1(\pp,\sset{\vk_1,\ldots,\vk_n},m,\sset{\sigma_1,\ldots,\sigma_q})\to\setsig$, where $\setsig$ is of size $\poly(\log n, \secParam)$ and $\TSAggr_1$ is deterministic.
		\item $\TSAggr_2(\pp, m,\setsig)\to\sigma$, \ie aggregate the signatures in $\setsig$ into a new signature $\sigma$. 	\end{itemize}
\end{enumerate}
\end{definition}

\paragraph{Robustness.}
Informally, a scheme is robust if no adversary can prevent sufficiently many honest parties from generating an accepting signature on a message.
We define robustness as a game between a challenger and an adversary $\Adv$. The game is formally defined in \cref{fig:exp_robust} and comprises of three phases.
In the \emph{setup and corruption} phase, the challenger generates the public parameters $\pp$ and a pair of signature keys for every party. Given $\pp$ and all verification keys $\vk_1,\ldots,\vk_n$, the adversary can adaptively corrupt a subset of (up to) $t$ parties and learn their secret keys.
In the case of a bare PKI (but \emph{not} of trusted PKI), the adversary can replace the verification key of any corrupted party by another key of its choice. Unless specified otherwise, we consider the bare PKI to be the default setup model.

In the \emph{robustness challenge} phase, the adversary chooses a tree $\tree$ describing the order in which the signatures of all the parties are to be aggregated. The nodes on level $0$  correspond to set of all parties who generate signatures (\ie all \emph{virtual} parties in the BA protocol). We slightly abuse notation and refer to level-$1$ nodes as leaf nodes, as they correspond to the actual leaves in the communication tree of~\cite{KSSV06}. For our application in the BA protocol in \cref{sec:ba_from_srds}, we require this tree to be an ``$(n,\IS)$\textsf{-almost-everywhere-communication tree}'' (see \cref{def:rob-com-tree}), where $n$ is the number of parties and $\IS$ is the set of corrupt parties.\footnote{This tree is a combinatorial object that was first defined by \citet{KSSV06}. They also proposed an interactive protocol that allows the parties to collectively build such a tree on the fly. This tree and that protocol will be an integral part of our BA protocol in \cref{sec:ba_from_srds}.}
Furthermore, we assume that level-$0$ nodes are indexed and ordered by the parties in such a way that when the tree topology is expressed flat as a planar graph (no crossovers), then the IDs of level-$0$ nodes are in increasing order. Looking ahead, we will show that this property of the tree is sufficient for our BA protocol in \cref{sec:ba_from_srds}. The adversary also chooses messages $m\in\MS$ and $\sset{m_i}_{i\in \goodP}$, where $\goodP$ is the subset of honest parties that are assigned to leaf nodes that do not have a \emph{good path} (\ie where more than a third of the parties assigned to at least one of the nodes on the path are corrupt) to the root.

Given signatures of parties in $\goodP$ on the respective $m_i$'s and of the remaining honest parties on $m$, the adversary computes signatures of all corrupt parties. The challenger and adversary then interactively aggregate all these signatures in the order specified by the tree $\tree$. In particular, partially aggregated signatures corresponding to  intermediate nodes in the tree that consist of a majority of honest parties are computed by the challenger, while partially aggregated signatures corresponding to the remaining nodes are chosen by the adversary.

Finally, in the \emph{output} phase, the challenger runs the verification algorithm on the message $m$ and the final aggregated signature obtained in the root of the tree, and \Adv wins if the verification fails. We say that an SRDS scheme is robust if no adversary can win this game except with negligible probability.

We start by formally describing the properties of an $(n,\IS)$\textsf{-almost-everywhere-communication tree}, which is a slight variant of the tree described in \citet{KSSV06}.

\begin{definition}[$(n,\IS)$\textsf{-almost-everywhere-communication tree}]\label{def:rob-com-tree}
Let $\IS\subseteq[n]$ be a subset of size $t$ for  $t<n/3$.
A directed rooted tree $\tree=(\vertex,\edge)$ is an \textsf{$(n,\IS)$-almost-everywhere-communication tree} if the following properties are satisfied:

\begin{enumerate}[leftmargin=*]
	\item
	The height of $\tree$ is $\ell^*\in O(\log n/ \log \log n)$.
	Each node $v$ from level $\ell> 1$ has $\log n$ children in level $\ell-1$.
	\item Each node on level $\ell>1$ is assigned a set of $\log^3 n$ parties.
	\item
	A node is \emph{good} if less than a third of the parties assigned to it are in $\IS$.
	Then, it holds that the root is good.
	\item
	All but a $3/\log n$ fraction of the leaves have a \emph{good path} (consisting of good nodes) to the root.
	\item The nodes on level 0 correspond to the $n$ parties.
	\item Each party (on level 0) is assigned to exactly one leaf node (on level 1).
	\item There are $n/\log^5 n$ leaf nodes and each leaf node is assigned a set of $\log^5 n$ parties.
\end{enumerate}
\end{definition}

\begin{definition}[robustness]\label{def:TSig:robust}
Let $t<n/3$.
An SRDS scheme $\Pi$ is $t$-\textsf{robust} with a bare PKI (\resp with a trusted PKI) if for $\mode=\bbpki$ (\resp $\mode=\trpki$) and for any (stateful) PPT adversary $\Adv$ it holds that:
\[
\pr{ \ExptTSrobust_{\mode,\Pi,\Adv}(\secParam, n, t)=0 } \leq \negl(\secParam,n).
\]
The experiment $\ExptTSrobust_{\mode,\Pi,\Adv}$ is defined in \cref{fig:exp_robust}.
\end{definition}

We note that \emph{robustness} is a strictly stronger notion than \emph{completeness}. In a complete scheme, correctness is guaranteed if all the parties are honest. In a robust scheme, even if a subset of parties are corrupted, as long as there are sufficiently many honest parties, correctness is still guaranteed. Hence, any signature scheme satisfying robustness, immediately satisfies completeness.

\input{Expt_TS_robust}

\paragraph{Unforgeability.}
Informally, a scheme is unforgeable if no adversary can use signatures of a large majority of the honest parties on a message $m$ and of a few honest parties on messages of its choice to forge an aggregated SRDS signature on a message other than $m$.

In a similar way to robustness, we consider an unforgeability game between a challenger and an adversary. The \emph{setup and corruption} phase is identical to that in the robustness game.
In the \emph{forgery challenge} phase, the adversary chooses a set $\setP\subseteq[n]\setminus\IS$ such that $\ssize{\setP\cup\IS}<n/3$, and messages $m$ and $\sset{m_i}_{i\in \setP}$. Given signatures of all honest parties outside of $\setP$ on the message $m$ and a signature of each honest party $\Party_i$ in $\setP$ on the message $m_i$, the adversary outputs a signature $\sigma$.
In the \emph{output} phase, the challenger checks whether $\sigma$ is a valid signature on a message different than $m$; if so, the adversary wins. An SRDS scheme is unforgeable if no adversary can win the game except for negligible probability.

\begin{definition}[unforgeability]
Let $t<n/3$.
An SRDS scheme $\Pi$ is $t$-\textsf{unforgeable} with a bare PKI (\resp with a trusted PKI) if for $\mode=\bbpki$ (\resp $\mode=\trpki$) and for every (stateful) PPT adversary $\Adv$ it holds that
\[
\pr{\ExptTSforge_{\mode,\Pi,\Adv}(\secParam, n, t)=1 }\leq \negl(\secParam,n).
\]
The experiment $\ExptTSforge_{\mode,\Pi,\Adv}$ is defined in \cref{fig:expforge}.
\end{definition}

\input{Expt_TS_forge}

We note that as described, the security definition is only for one-time SRDS signatures. Although this is sufficient for our applications in \cref{sec:app:BA}, it is possible to extend the definition and provide the adversary an oracle access to signatures of honest parties on messages of its choice. However, in that case, the adversary must choose the set $\setP$ before getting oracle access.

\paragraph{Security.}
We say that an SRDS scheme is secure in the respective PKI model, if it satisfies all the above properties.
\begin{definition}[secure SRDS]
Let $t<n/3$. An SRDS scheme $\Pi$ is $t$-\textsf{secure} with a bare PKI (\resp with a trusted PKI) if it is \textsf{succinct}, $t$-\textsf{unforgeable} and $t$-\textsf{robust} with a bare PKI (\resp with a trusted PKI).
\end{definition}

%% file: Expt_TS_robust.tex
\begin{nfbox}{Robustness experiment for SRDS}{fig:exp_robust}
\begin{center}
\textbf{Experiment} $\ExptTSrobust_{\mode,\Pi,\Adv}(\secParam,n,t)$
\end{center}
The experiment $\ExptTSrobust$ is a game between a challenger and the adversary $\Adv$. The game is parametrized by an SRDS scheme $\Pi$ and proceeds as follows:
\begin{enumerate}
	\item[A.] \textbf{Setup and corruption.}
    In the first phase, the challenger generates the public parameters and the signature keys for the parties. Given the public information, $\Adv$ can adaptively corrupt parties, learn their secret information, and potentially change their public keys.
	\begin{enumerate}[label={(\arabic*)}]
		\item Compute $\pp\gets\TSCR(1^\secParam, 1^n)$.
		\item For every $i\in[n]$, compute $(\vk_i,\sk_i)\gets \TSGen(\pp)$.
		\item Invoke $\Adv$ on $(1^\secParam, 1^n,\pp,\sset{\vk_1,\ldots,\vk_n})$ and set $\IS=\emptyset$.
        \item\label{step:corruption} As long as $\ssize{\IS}\leq t$
        and $\Adv$ requests to corrupt a party $\Party_i$:
		\begin{enumerate}[label={(\alph*)}]
			\item Send $\sk_i$ to $\Adv$ and receive back $\vk_i'$.
			\item\label{step:replacevk} If $\mode=\bbpki$, set $\vk_i=\vk_i'$.
			\item Set $\IS=\IS\cup \sset{i}$.
		\end{enumerate}
	\end{enumerate}
	\item[B.] \textbf{Robustness challenge.} In this phase, \Adv tries to break the robustness of the scheme.
	\begin{enumerate}[label={(\arabic*)}]
        \item \Adv chooses an $(n,\IS)$-\textsf{almost-everywhere-communication tree} $\tree=(\vertex,\edge)$ (as per \cref{def:rob-com-tree}), in which level-$0$ nodes are indexed and ordered by the parties in such a way that when the tree topology is expressed flat as a planar graph (no crossovers), then the IDs of level-$0$ nodes are in increasing order. Let $\goodP$ be the set of honest parties assigned to the leaf nodes that do not have a good path to the root.
        \item \Adv also chooses a message $m\in \MS$ and a message $m_i\in\MS$ for each $i\in \goodP$.
        \item
        For every $i\in [n]\setminus( \IS\cup \goodP)$, let $\sigma_i\gets\TSSignShare(\pp,i, \sk_i, m)$ and for every $i\in \goodP$, let $\sigma_i\gets\TSSignShare(\pp,i, \sk_i, m_i)$.

        \item
        Send $\sset{\sigma_i}_{i\in [n]\setminus \IS}$ to \Adv and receive back $\sset{\sigma_i}_{i\in \IS}$.
        \item
        For each $\ell=\sset{2,\ldots,\height(\tree)}$ and every node $v$ on level $\ell$:
        \begin{itemize}
        	\item If $v$ is a good node, compute $\sigma_v\gets \TSAggr(\pp, \sset{\vk_1,\ldots,\vk_n}, m, \sset{\sigma_u}_{u\in \child(v)})$, where $\child(v)\subseteq\vertex$ refers to the set of children of the node $v\in\vertex$, and send $\sigma_v$ to \Adv.
        	\item Else, if $v$ is a bad node, receive $\sigma_v$ from $\Adv$.
        \end{itemize}
   	\end{enumerate}

	\item[C.] \textbf{Output Phase.}
    Output $\TSVerify(\pp,\sset{\vk_1,\ldots,\vk_n},m,\sigma_\rroot)$, where $\rroot$ is the root node in $\tree$.
\end{enumerate}
\end{nfbox}

%% file: Expt_TS_forge.tex
\begin{nfbox}{Forgery experiment for SRDS}{fig:expforge}
\begin{center}
\textbf{Experiment} $\ExptTSforge_{\mode,\Pi,\Adv}(\secParam,n, t)$
\end{center}
The experiment $\ExptTSforge$ is a game between a challenger and the adversary $\Adv$. The game is parametrized by an SRDS scheme $\Pi$ and consists of the following phases:
\begin{enumerate}
	\item[A.] \textbf{Setup and Corruption.} As in the robustness experiment in \cref{fig:exp_robust}.
	\item [B.] \textbf{Forgery Challenge.} In this phase, the adversary tries to forge a signature.
	\begin{enumerate}
		\item
        \Adv chooses a subset $\setP\subseteq [n]\setminus\IS$ such that $\ssize{\setP\cup \IS}< n/3$.
        It also chooses messages $m$ and $\sset{m_i}_{i\in \setP}$ from $\MS$.
		\item For every $i\in \setP$, compute $\sigma_i\gets\TSSignShare(\pp, i, \sk_i, m_i)$.
		\item  For every $i\notin (\setP\cup\IS)$, compute $\sigma_i\gets\TSSignShare(\pp, i, \sk_i, m)$.
		\item Send $\sset{\sigma_i}_{i\in [n]\setminus \IS}$ to \Adv and get back $\sigma'\in\XS$ and $m'\in\MS$.
	\end{enumerate}

	\item[C.] \textbf{Output Phase.}
    Output $1$ if and only if $\TSVerify(\pp,\sset{\vk_1,\ldots,\vk_n}, m',\sigma')=1$ and $m'\neq m$.
\end{enumerate}
\end{nfbox}

%% file: balanced_ba.tex
\section{Balanced Communication-Efficient Byzantine Agreement}\label{sec:app:BA}

In this section, we consider Byzantine agreement protocols with $\tilde{O}(1)$ communication per party.
In \cref{sec:ba_from_srds}, we show how to use \srds (SRDS) to boost almost-everywhere agreement to full agreement in a balanced way via a single communication round.
In \cref{sec:ba_lb}, we show that a similar task cannot be achieved under weaker setup assumptions.

\subsection{Balanced Byzantine Agreement from SRDS}\label{sec:ba_from_srds}
We start by showing how to combine \srds (SRDS) with the protocol of \cite{BGT13} to obtain BA with balanced $\tilde{O}(1)$ communication. We prove the following theorem.

\begin{theorem}[\cref{thm:intro:ba}, restated]\label{thm:ba}
Let $\beta<1/3$ and assume existence of a $\beta n$-secure SRDS scheme in the bare-PKI model (\resp trusted PKI model). Then, there exists a $\beta n$-resilient BA protocol (according to \cref{def:ba:simulation}) in a hybrid model for generating the SRDS setup and the relevant PKI, such that:
\begin{itemize}
    \item
    The round complexity and communication locality are $\polylog(n)$; every party communicates $\polylog(n)\cdot\poly(\secParam)$ bits.
    \item
    The adversary can adaptively corrupt the parties based on the public setup and the PKI before the onset of the protocol. For bare PKI, the adversary can additionally replace the corrupted parties' public keys.
\end{itemize}
\end{theorem}

By instantiating \cref{thm:ba} with our SRDS constructions from \cref{sec:SRDS:construction}, we get the following corollaries.

\begin{corollary}
Let $\beta<1/3$. Assuming OWF, there exists a $\beta n$-resilient BA protocol in the trusted-PKI model with balanced $\tilde{O}(1)$ communication per party.
\end{corollary}

\begin{corollary}
Let $\beta<1/3$. Assuming CRH and SNARKs with linear extraction, there exists a $\beta n$-resilient BA protocol in the bare PKI and CRS model with balanced $\tilde{O}(1)$ communication per party.
\end{corollary}

\paragraph{High-level overview.}
The protocol is defined in a hybrid model that abstracts the communication tree of \cite{KSSV06}. The parties can communicate in a way that mimics almost-everywhere agreement, and the adversary is allowed to isolate a $o(1)$ fraction of the parties. Each party is assigned to $z=O(\log^4n)$ leaf nodes and $\zs=O(\log^5 n)$ parties are assigned to each leaf node in the communication tree. Since each party will send a signature to every leaf node he is assigned to, it is essential to ensure the same fraction of signatures is generated by corrupted parties as their fraction in the party-set.
For this reason, we allocate $z$ ``virtual identities'' to every party.
The SRDS is used for $n\cdot z$ virtual identities and each party samples separate SRDS keys for each of his virtual identities. These virtual IDs are assigned to the parties in such a way that the virtual IDs associated with the \kth leaf node belong in the range $[(k-1)\cdot \zs+1,k\cdot \zs]$. This ensures that when the tree topology is expressed flat as a planar graph (no crossovers), then the virtual IDs of the leaf nodes are in increasing order. Looking ahead, this property is necessary for robustness, when using our SRDS construction from Section~\ref{sec:SRDS:snarks}.

The protocol starts by invoking $\faecomm$ (defined below) to obtain an almost-everywhere-communication tree where each party is assigned to $z$ leaves. The supreme committee members (parties assigned to the root-node) run Byzantine agreement on their inputs to agree on the output $y$ and run a coin-tossing protocol to agree on a random seed $s$. The supreme committee then makes use of the communication-tree to distribute these values to all non-isolated parties. The parties then collectively generate an SRDS signature to certify the pair $(y,s)$.

To compute this signature, each party locally signs the received pair of values; this is done using a different virtual identity for every leaf node corresponding to the party. Each signature is sent to all parties assigned to the corresponding leaf node. For each node in the tree, the assigned parties aggregate the received signatures, while making sure that the maxima and minima of virtual IDs associated with the signatures that they aggregate indeed lie within the range associated with the leaf nodes that have a path to the current node, and propagate them to the node's parent in a recursive way until reaching the root, where the final aggregated signature is computed.

Next, the supreme-committee again uses the communication-tree to distribute this aggregated signature to all non-isolated parties. Each non-isolated party evaluates a PRF on the seed $s$ and its identity to determine a set of parties, to which he sends the pair $(y,s)$ along with the signature. Isolated parties can now verify the signature and be convinced about the correct output $y$.
Here correctness crucially relies on the fact that the adversary could not have forged an aggregated SRDS signature on any other value.

In \cref{sec:app:ba:func}, we define the ideal functionalities to be used in the BA protocol, and in \cref{sec:app:ba:prot}, we describe the protocol and prove its security. Finally, in \cref{sec:app:ba:mpc}, we present applications of our protocol to broadcast and MPC.

\subsubsection{Functionalities used in the Protocol\SUBSUBSEC}\label{sec:app:ba:func}
We start by describing the functionalities used in our construction.

\paragraph{Almost-everywhere communication.}
The functionality $\faecomm$ is a reactive functionality that abstracts the properties obtained by the protocol from \cite{KSSV06}.
In the first invocation, the adversary specifies a special communication tree that allows all honest parties to communicate, except for a $o(1)$ fraction of isolated parties $\DS\subseteq[n]$. In all subsequent calls, the ``supreme committee,'' \ie the parties associated with the root of the tree, can send messages to all of the parties but~$\DS$.
We use a slightly modified version of the $(n,\IS)$\textsf{-almost-everywhere-communication tree} defined in \cref{sec:SRDS}. Specifically, in \cref{def:rob-com-tree}, each party was assigned to a single leaf node of the tree. Here, each party in the BA protocol will be assigned to multiple leaf nodes (but will participate in the SRDS aggregation as multiple ``virtual'' parties, one for each appearance).

\begin{definition}[$(n,\IS)$\textsf{-almost-everywhere-communication tree with repeated parties}]\label{def:com-tree}
Let $\IS\subseteq[n]$ be a subset of size $\beta n$ for a constant $\beta<1/3$.
A directed rooted tree $\tree=(\vertex,\edge)$ is an $(n,\IS)$-\textsf{almost-everywhere-communication tree with repeated parties} if it satisfies the first four properties of an $(n,\IS)$-\textsf{almost-everywhere-communication tree} (\cref{def:rob-com-tree}) and additionally, the following properties are satisfied:

\begin{enumerate}[leftmargin=*]
	\item Each leaf node of the tree is assigned a set of $\log^5 n$ parties.
	\item Each party is assigned to $O(\log^4 n)$ nodes at each level.
\end{enumerate}
\end{definition}

The original protocol of \cite{KSSV06} has an inverse-polynomial error in $n$; the reason is that the committees are chosen to be $O(\log{n})$. Boyle et al.\ \cite{BGT13} adjusted the protocol to have committees of poly-logarithmic size, thus obtaining poly-logarithmic locality with a negligible error in $n$. Note that the security parameter $\secParam$ is not used in this protocol, so the locality is independent of $\secParam$.

As observed in \cite{BGT13}, the fact that $1-o(1)$ fraction of the leaves are on good paths to the root implies that for a $1-o(1)$ fraction of the parties, a  majority of the leaf nodes that they are assigned to are good. The protocol of \citet{KSSV06} securely realizes $\faecomm$ in the authenticated-channels model tolerating a computationally unbounded, malicious adversary statically corrupting $\beta n$ parties, for a constant $\beta<1/3$.
Every invocation requires $\polylog(n)$ rounds, and every party sends and processes $\polylog(n)$ bits. Throughout all invocations, every party sends to, and processes messages received from $\polylog(n)$ other parties.

\input{func_ae_comm_tree}

\paragraph{Byzantine agreement.}
We consider the standard Byzantine agreement functionality $\fba$ as defined in \cref{sec:Preliminaries} (to be used within small committees in the larger protocol).
Every party sends its input to the trusted party who forwards the input value to the adversary. If more than $n-t$ inputs equal the same value $y\in\zo$, then deliver $y$ as the output for every party. Otherwise, let the adversary choose the value $y\in\zo$ to be delivered.

The $n$-party BA protocol of \citet{GM93} realizes $\fba$ over authenticated channels tolerating a computationally unbounded, malicious adversary statically corrupting $t<n/3$ parties using $t+1$ rounds and $\poly(n)$ communication complexity.
An immediate corollary is that for $n'=\polylog(n)$, the $n'$-party BA functionality $\fba$ can be instantiated using $\polylog(n)$ rounds and $\polylog(n)$ communication complexity.

\paragraph{Coin tossing.}
The coin-tossing functionality $\fct$ samples a uniformly distributed $s\in\zo^\secParam$ and delivers $s$ to all the parties. The protocol of \citet{CGMA85} realizes $\fct$ over a broadcast channel assuming an honest majority (by having each party verifiably secret share (VSS) a random value, and later reconstruct all values and XOR them). By instantiating the broadcast channel using the protocol of \cite{GM93}, $n'=\polylog(n)$ parties can agree on a random $\secParam$-bit string in $\polylog(n)$ rounds and $\polylog(n)\cdot\poly(\secParam)$ communication.

\paragraph{Signature aggregation.}
The signature-aggregation functionality $\faggrsig$ (formally described in \cref{fig:faggrsig}) is an $n'$-party functionality, where every party $\Party_i$ provides a message $m_i$ and a set of signatures.
The functionality first determines the set of signatures received from a majority of the parties and aggregates only those signatures to obtain a new signature $\sigma$, which is delivered as the output for every party.

Note that the inputs to the aggregation procedure are not private, so if the aggregation algorithm $\TSAggr_2$ is deterministic (for example, in the OWF-based SRDS construction in \cref{sec:SRDS:owf}) the parties simply need to agree on the common set of input signatures $S_\sig$ and locally run $\TSAggr_2$ to obtain the same aggregated signature. To agree on $S_\sig$, each party broadcasts its input signatures and filters-out invalid signatures by running the deterministic algorithm $\TSAggr_1$. However, if the algorithm $\TSAggr_2$ is randomized, it may be the case that security relies on keeping the random coins hidden from the parties. For this reason, after the parties agree on $S_\sig$, we use an MPC protocol to compute the aggregated signature and realize $\faggrsig$.

\input{func_aggr_sig}

Assuming the existence of one-way functions, the protocol of \citet{DI05} can be used to realize the $n'$-party functionality $\faggrsig$, for $n'=\polylog(n)$, over secure channels, tolerating a malicious adversary corrupting a minority of the parties.
In addition, if the size of set $S_\sig$ is $\tilde{O}(1)$, the protocol requires $\polylog(n)\cdot\poly(\secParam)$ communication. In our construction, this functionality is used by the parties assigned to a node (in the almost-everywhere communication-tree obtained from $\faecomm$) for aggregating signatures received from parties assigned to their children. From \cref{def:com-tree}, we know that each node only has $\log(n)$ child nodes and each node is assigned $\polylog(n)$ parties. Therefore, $\faggrsig$ is only used for aggregating at most $\polylog(n)$ signatures.  Note that in \cite{DI05} a broadcast channel is also required and the resulting protocol is constant round. For $n'=\polylog(n)$ the broadcast can be realized by a deterministic protocol, \eg from \cite{GM93}, and the resulting protocol has $\polylog(n)$ rounds and $\polylog(n)\cdot\poly(\secParam)$ communication.

\subsubsection{The Byzantine Agreement Protocol\SUBSUBSEC}\label{sec:app:ba:prot}

Having defined the ideal functionalities, we are now ready to present our BA protocol in \cref{fig:protba}.
To reduce the security of $\protba$ to that of the SRDS scheme, we will show that by \emph{robustness} every honest party will receive an accepting signature on $(y,s)$, and by unforgeability, no party will receive an accepting signature on a different value. Before proceeding to the proof, we discuss a subtlety in the reduction.

\ifdefined\IsTrackChanges

\input{prot_ba}
\newpage
\fi

Recall that robustness of an SRDS scheme ensures that an adversary who after the \emph{setup and corruption} phase is allowed to choose a message $m$, and the order of aggregation (using a directed rooted tree $\tree$), cannot prevent the honest parties from successfully signing $m$.
Note that if in $\protba$, an adversary can prevent the honest parties from signing $(y,s)$, then we can derive the corresponding tree and partially aggregated signatures of the corrupted parties to break the robustness of the SRDS scheme. Note that here, in the robustness game, we will assume that the total number of parties are $n\cdot z$ (\ie each virtual party in the Byzantine agreement protocol is a real party in the SRDS game) and  $(n,\IS)$-\textsf{almost-everywhere-communication tree with repeated parties} $\tree$ used in the Byzantine agreement protocol is transformed into an $(n\cdot z,\{(i,j)\}_{i\in\IS,j\in[z]})$-\textsf{almost-everywhere-communication tree} by augmenting it with a level 0 comprising of $n\cdot z$ nodes (representing the $n\cdot z$ parties in the SRDS game), and adding an edge between each of these nodes and the leaf node that it (\ie the party that they represent) is assigned to.

\begin{lemma}\label{lem:ba}
Let $\beta<1/3$ and assume the existence of PRF and $\beta n$-secure SRDS in the bare-PKI model (\resp trusted-PKI model). Then, protocol $\protba$ is a $\beta n$-resilient BA protocol in the $(\faecomm,\fba,\fct,\faggrsig)$-hybrid model such that:

\begin{itemize}[leftmargin= *]
    \item
    The round complexity and the locality of the protocol are $\polylog(n)$; the number of bits communicated by each party is $\polylog(n)\cdot\poly(\secParam)$.
    \item
    The adversary can adaptively corrupt the parties based on the public setup of the SRDS, \ie $\pp$ and $\sset{\vk_{1,1},\ldots,\vk_{n,z}}$ before the onset of the protocol. For bare PKI, the adversary can additionally replace the corrupted parties' public keys.
\end{itemize}
\end{lemma}

\vspace{-0.1 cm}
\noindent The proof of \cref{lem:ba} can be found in \cref{sec:ba_from_srds_cont}.

\input{BA_applications}

\subsection{Lower Bound on Balanced Byzantine Agreement}\label{sec:ba_lb}
In the previous section, we showed how to extend almost-everywhere agreement to full agreement in one round. The minimal setup assumptions used were a bare PKI and CRS. In \cref{sec:ba_lb_crs}, we show the some form of private-coin setup is necessary for this task.\footnote{We note that, our lower bound easily extends to the random oracle model, for the sake of simplicity we prove it merely with a CRS setup.}
In \cref{sec:ba_lb_pki}, we show that in the PKI model, where the public/private keys of each party are independently sampled, cryptographic assumptions are further needed.

\subsubsection{Lower Bound on Balanced Byzantine Agreement in CRS Model\SUBSUBSEC}\label{sec:ba_lb_crs}

We denote by $\faecomm^\ast$ a weakened version of the functionality $\faecomm$ (from \cref{fig:faecomm}) that enables communication between almost all of the parties, except for an isolated set $\DS$ that is randomly chosen by the functionality, rather than by the adversary.
We note that this notion is non-standard and is not achieved by existing protocols for almost-everywhere agreement.
The purpose of this adjustment is to provide a stronger lower bound, as the adversary's capabilities are more restricted. In fact, we only require that with some inverse-polynomial probability, there exists a single isolated party that is chosen by the functionality.

\def\ThmLBCRS
{
Let $\pi$ be a $\beta n$-resilient Byzantine agreement protocol in the $(\fcrs,\faecomm^\ast)$-hybrid model, for $\beta<1$.
Assume that $\pi$ has two parts: the first consists of a polynomial number of rounds where communication is via $\faecomm^\ast$, and the second consists of a single round over \ptp channels.
Then, there exists a party that sends $\Theta(n)$ messages in the last round.
}

\begin{theorem}[\cref{thm:intro:lb_crs}, restated]\label{thm:ba_lb_crs}
\ThmLBCRS
\end{theorem}

\input{lb_proof_crs}

\subsubsection{Lower Bound on Balanced Byzantine Agreement in PKI Model\SUBSUBSEC}\label{sec:ba_lb_pki}
We proceed to prove the second lower bound, showing that in the trusted PKI model, where each party receives an independently sampled pair of public/private keys, one-way functions are necessary for extending almost-everywhere agreement to full agreement in a single communication round. Note that a lower bound in the trusted PKI model readily implies a lower bound in weaker PKI models.

\def\ThmLBPKI
{
Let $\pi$ be a $\beta n$-resilient Byzantine agreement protocol in the trusted PKI and $\faecomm^\ast$-hybrid model, for $\beta<1$.
Assume that $\pi$ has two parts: the first consists of a polynomial number of rounds where communication is via $\faecomm^\ast$, and the second consists of a single round over \ptp channels.
Then, if one-way functions do not exist, there exists a party that sends $\Theta(n)$ messages in the last round.
}
\begin{theorem}[\cref{thm:intro:lb_pki}, restated]\label{thm:ba_lb_pki}
\ThmLBPKI
\end{theorem}

\input{lb_proof_pki}

%% file: func_ae_comm_tree.tex
\begin{nfbox}{The almost-everywhere communication functionality}{fig:faecomm}
\begin{center}
	\textbf{The functionality} $\faecomm$
\end{center}
The $n$-party reactive functionality $\faecomm$ proceeds as follows:
\begin{itemize}[leftmargin=*]
	\item \textbf{First invocation:}
    Upon receiving an $\init$ message from each party, the functionality asks the adversary for a communication tree $\commtree=(\vertexCommtree,\edgeCommtree)$ and does the following:
    \begin{enumerate}[leftmargin=*]
	   \item
        Verify that $\commtree$ is an $n$-party almost-everywhere-communication tree \wrt the set of corrupted parties $\IS$ (otherwise, output $\bot$ to all parties).
	    \item
        Let $\DS$ be the set of isolated parties in $\commtree$ and let $\CS$ be the set of parties assigned to the root.
        \item
        The functionality sends to each $\Party_i$ for $i\in[n]$ its local view in the tree, consisting of:
        \begin{itemize}[leftmargin=*]
        	\item All the nodes that $\Party_i$ is assigned to (and the parties assigned to them).
        	\item All the parent and children nodes (and the parties assigned to them) of the nodes that $\Party_i$ is assigned to.
        \end{itemize}
    \end{enumerate}
    \item \textbf{Subsequent invocations:}
    Every party $\Party_i$ with $i\in\CS$ provides a message $m_i$. If more than $2/3$ of the parties in $\CS$ provided the same message $m$, send $m$ to the adversary and receive back $\sset{\hat{m}_j}_{j\in \DS}$. For every $i\notin\DS$ deliver $m$ to $\Party_i$ and for every $j\in\DS$ deliver $\hat{m}_j$ to $\Party_j$.
\end{itemize}
\end{nfbox}

%% file: func_aggr_sig.tex
\begin{nfbox}{The signature-aggregation functionality}{fig:faggrsig}
\begin{center}
	\textbf{The functionality} $\faggrsig(\PS)$
\end{center}
The $n'$-party functionality $\faggrsig$, running with parties $\PS=\sset{\Party_1,\ldots,\Party_{n'}}$ and the adversary,
is parametrized by the public parameters $\pp$ and proceeds as follows.
\begin{enumerate}[leftmargin=*]
\item
    Every party $\Party_i$ sends $(m_i,S_{\sig_i})$ as input, where $S_{\sig_i}$ is a set of signatures.
	\item
    If at least $2/3$ of the parties provided the same message $m$ and the same set $S_\sig$, then compute
    \[
    \sigma\gets\TSAggr_2\big(\pp, m,S_\sig\big).
    \]
    Else, let the adversary choose $\sigma$.
	\item
    Finally, deliver $\sigma$ to every party $\Party_i$.
\end{enumerate}
\end{nfbox}

%% file: prot_ba.tex
\begin{nfbox}{Byzantine agreement with balanced $\polylog$ communication}{fig:protba}
\begin{center}
    \textbf{Protocol} $\protba$
\end{center}
\begin{itemize}
\vspace{-.7cm}
    \item\textbf{Common Input:}
    An SRDS scheme and a PRF family $\FS=\sset{F_s}_{s\in\zo^\secParam}$ mapping elements of $[n]$ to subsets of $[n]$ of size $\polylog(n)$.
\vspace{-.15cm}
    \item\textbf{Private Input:}
    Every party $\Party_i$, for $i\in[n]$, has input $x_i\in\zo$.
\vspace{-.15cm}
    \item\textbf{Setup:}
    Let $z=O(\log^4n)$, $\zs=O(\log^5 n)$ and let $\pp\gets\TSCR(1^\secParam,1^{n\cdot z})$.
   Every party $\Party_i$ locally computes $(\vk_{i,j},\sk_{i,j})\gets\TSGen(\pp)$ for every $j\in[z]$.
    The public output consists of $\pp$ and the set of public keys $\vk=\sset{\vk_{i,j}}_{i\in[n],j\in[z]}$.  We assume that there exists a mapping $\mapping:[n]\times[z]\to[n\cdot z]$ that maps the each $(i,j)$ above to a virtual ID $\is\in[n\cdot z]$, such that virtual IDs of the parties assigned and corresponding to the \kth leaf node belong in the range $[(k-1)\cdot\zs+1,k\cdot\zs]$ (This ensures that when the tree topology is expressed flat as a
planar graph (no crossovers), then the virtual IDs of the leaf nodes are in increasing order.). 
\vspace{-.15cm}
    \item\textbf{Hybrid Model:}
    The protocol is defined in the $(\faecomm,\fba,\fct,\faggrsig)$-hybrid model.
\vspace{-.15cm}
    \item\textbf{The Protocol:}
\end{itemize}
\begin{enumerate}
\vspace{-.6cm}
    \item\label{step:faecomm_first}
    Every party invokes $\faecomm$ and receives back its local view in the communication tree $\commtree=(\vertexCommtree,\edgeCommtree)$.
    Let $\CS$ denote the supreme committee, \ie the parties assigned to the root node.
\vspace{-.15cm}
    \item
    Every party $\Party_i$ in the supreme committee (\ie with $i\in\CS$) proceeds as follows:
    \begin{enumerate}
        \item\label{step:fba}
        Invoke $\fba$ on his input value $x_i$ and receive back $y\in\zo$.
        \item\label{step:fct}
        Invoke $\fct$ and receive back $s\in\zo^\secParam$.
    \end{enumerate}
\vspace{-.15cm}
    \item\label{step:faecomm_second}
    The parties in the supreme committee $\CS$ send $(y,s)$ to $\faecomm$.
    For every $i\in[n]$ denote the output of party $\Party_i$ as $(y_i,s_i)$.
\vspace{-.15cm}
    \item\label{step:sign}
    Every party $\Party_i$ signs the received message $(y_i,s_i)$ for each virtual identity $j\in[z]$ as $\sigma_{i,j}\gets\TSSignShare(\pp,\mapping(i,j),\sk_{i,j},(y_i,s_i))$.
    Let $L_i=\sset{\leafnode_{i_1},\ldots,\leafnode_{i_z}}\subseteq\vertexCommtree$ be the subset of leaves assigned to $\Party_i$.
    For each $j\in[z]$, party $\Party_i$ sends $\sigma_{i,j}$ to all the parties assigned to the leaf node~$\leafnode_{i_j}$.
\vspace{-.15cm}
	\item\label{step:recurse}
    Denote by $\party(v)$ the set of parties assigned to a node $v\in \vertex$. Similarly, denote by $\child(v)$ and $\parent(v)$ the set of children nodes and parent node of $v\in \vertex$, resp. Let $\Range(v)$ denote the range of virtual IDs of the parties assigned to the leaf nodes that have a path to node $v\in\vertex$.
    For each level $\ell=1,\ldots,\ell^*$ and for each node $v$ on level $\ell$, the protocol proceeds as follows:
    \begin{enumerate}
  	    \item\label{step:recurse:firstlevel}
        For each $i\in \party(v)$, let $S^{i,\ell,1}_\sig$ be the set of signatures received by $\Party_i$ in the previous round (for $\ell=1$, \ie for leaf nodes, from each $\Party_j$ with $v\in L_j$; for $\ell>1$, from every party $\Party_j$ assigned to a child node of $v$).
        \vspace{-.15cm}
        \item\label{step:recurse:ba}
        Every $\Party_i$ with $i\in\party(v)$ broadcasts\footnote{\tiny To ensure that the corrupt parties do not broadcast very long messages, we assume that the parties broadcast each element in $S^{i,\ell,1}_\sig$ one-by-one and each party is only allowed to initiate polylogarithmic number of broadcasts.} $S^{i,\ell,1}_\sig$  to all the parties in $\party(v)$. Let $S^{i,\ell,2}_\sig$ be the union of all sets received from the parties in $\party(v)$.
        \vspace{-.15cm}
  	    \item\label{step:recurse:secondlevel}
        Every $\Party_i$ with $i\in\party(v)$ computes $\TSAggr_1(\pp,\sset{\vk_{1,1},\ldots,\vk_{n,z}},(y_i,s_i),S^{i,\ell,2}_\sig)\to S^{i,\ell,3}_\sig$.
       	If $\ell=1$, for each  $\signsig$ in $S^{i,\ell,3}_\sig$ it checks if $\indexMin(\signsig)=\indexMax(\signsig)$ and if $\indexMin(\signsig)\in\Range(v)$ and if $\ell>1$, it checks if $\exists v'\in\child(v)$ such that the range  $[\indexMin(\signsig),\indexMax(\signsig)]$ falls within the range $\Range(v')$. If this check fails for any $\signsig$, it updates $S^{i,\ell,3}_\sig=S^{i,\ell,3}_\sig\setminus\sset{\signsig}$. It invokes $\faggrsig$ on input $((y_i,s_i),S^{i,\ell,3}_\sig)$ to obtain the aggregated signature $\sigma_v$.
       	\vspace{-.15cm}
  	    \item
        If $\ell<\ell^*$, for each $i\in\party(v)$, party $\Party_i$ sends $\sigma_v$ to all parties in $\parent(v)$.
    \end{enumerate}
\vspace{-.15cm}
    \item\label{step:faecomm_third}
    Let $\sigma_\rroot$ be the signature obtained by the supreme committee.
    The parties in the supreme committee send $(y,s,\sigma_\rroot)$ to $\faecomm$.
    Let the output of party $\Party_i$ for $i\in[n]$ be $(y_i',s_i',\sigma'_i)$

\vspace{-.15cm}
    \item\label{step:prf}
    Each party $\Party_i$ (for $i\in[n]$) computes $\CS_i=F_{s'_i}(i)$, and sends $(y'_i,s'_i,\sigma'_i)$ to every party in $\CS_i$.
\vspace{-.15cm}
    \item\label{step:output}
    A party $\Party_j$ that receives a valid message $(y,s,\sigma)$ from a party $\Party_i$, satisfying $j\in F_s(i)$ and $\TSVerify(\pp,\sset{\vk_{1,1},\ldots,\vk_{n,z}}, (y,s), \sigma)=1$, outputs $y$ and halts.

    \vspace{-.5cm}
\end{enumerate}
\vspace{-.15cm}
\end{nfbox}

%% file: BA_applications.tex
\subsubsection{Applications\SUBSUBSEC}\label{sec:app:ba:mpc}

We point out a few applications of our BA protocol.

\paragraph{Broadcast with balanced polylog communication.}
Consider a single a run of the protocol (on dummy inputs). The communication graph forms a tree with stronger properties than \cref{def:com-tree}, achieving \emph{everywhere} agreement of all parties on the supreme committee, such that every party sends only $\tilde{O}(1)$ throughout the protocol constructing it. Having established the communication tree, it is possible to run a simple broadcast protocol in the PKI model.
The sender signs his input bit and sends it up to the supreme committee, which in turn sends the signed bit to all other parties. If fact, since the communication tree is reusable, after multiple executions (with different senders) the communication will grow in a proportional way only to the number of bits that have been broadcasted. In particular, note that the SRDS PKI is only needed for a single run of the protocol (to establish the communication tree) and is not needed afterwards.

\ifdefined\IsTrackChanges\else
\input{prot_ba}
\newpage
\fi

\begin{corollary}
Let $\beta<1/3$ be a constant.
Assuming $\beta n$-secure SRDS schemes, there exists an $n$-party binary broadcast protocol tolerating a malicious adversary that can statically corrupt $\beta n$ of the parties, such that the communication locality of $\ell$ executions is $\polylog(n)$, and the round complexity and the number of bits each party communicates is $\ell\cdot\polylog(n)\cdot\poly(\secParam)$.
\end{corollary}

\paragraph{MPC with amortized polylog communication overhead.}
Following the MPC protocol from~\cite{BGT13}, the supreme committee can run among themselves a protocol establishing an encryption key of a public-key encryption scheme where the decryption key is secret shared among the committee members, and broadcast the public key. Every party encrypts its input and sends it up the tree to the supreme committee that run an MPC protocol for decrypting all ciphertexts and compute the function. Using FHE-based MPC that minimize the communication (\eg \cite{AJLTVW12}), we obtain the following corollary.
\begin{corollary}
Let $\beta<1/3$ be a constant.
Assuming $\beta n$-secure SRDS and FHE schemes, every $n$-party functionality $f:(\zo^\inputlen)^n\to \zo^\outputlen$ can be securely computed tolerating a malicious adversary that can statically corrupt $\beta n$ parties, such that communication locality and round complexity are $\polylog(n)$, and amortized communication complexity is $(\inputlen+\outputlen)\cdot\polylog(n)\cdot\poly(\secParam)$.
\end{corollary} 

%% file: lb_proof_crs.tex
\begin{proof}
By classical results~\cite{PSL80,FLM86}, BA protocols cannot tolerate one-third of corrupted parties, even in the CRS model; therefore, we can assume that $\beta<1/3$.
Let $\pi$ be a protocol in the $(\fcrs,\faecomm^\ast)$-hybrid model that invokes $\faecomm^\ast$ for polynomially many rounds followed by a single \ptp round, and assume that the number of messages sent by every party in the last round is $o(n)$.
We will construct an adversarial strategy that violates the \emph{validity} of $\pi$ with noticeable probability.

\paragraph{Choosing the corrupted set.}
Given the common reference string $\crs$, the adversary starts by deciding on the set of corrupted parties. The adversary chooses a random subset $\JS\subseteq[n]$ of size $\beta n/2$ and simulates two executions of $\pi$ inside its head.
\begin{itemize}
    \item
    In the first execution, all parties have input bit $0$ where every party $\Party_j$ with $j\in\JS$ is corrupted and does not send any message throughout the protocol. For every $j\in\JS$, denote the set of parties that sends messages to $\Party_j$ in the last \ptp round by $\CS_j^0$ and record the messages as $\sset{\hat{m}^0_{i \to j}}_{i\in\CS_j^0}$.
    \item
    In the second execution, all parties have input bit $1$ where every party $\Party_j$ with $j\in\JS$ is corrupted and does not send any message throughout the protocol. For every $j\in\JS$, denote the set of parties that sends messages to $\Party_j$ in the last \ptp round by $\CS_j^1$ and record the messages as $\sset{\hat{m}^1_{i\to j}}_{i\in\CS_j^1}$.
\end{itemize}

In each of the virtual executions described above, from the joint view of all parties $\Party_i$ with $i\notin\JS$, every party $\Party_j$ with $j\in\JS$ could be an isolated honest party, so they must join forces and send messages to every such $\Party_j$.
Note that it could be that some parties in $\JS$ receive a linear number of messages, \eg if every party $\Party_i$ with $i\notin\JS$ sends a message to the same party $\Party_j$ for some $j\in\JS$.
However, as each party sends only $o(n)$ messages in this step, the number of such parties cannot be too large; in particular, there must be a party who receives $o(n)$ messages in \emph{both} of the above executions.

\begin{claim}\label{claim:lb_corruptions}
There exists $j\in\JS$ such that $\ssize{\CS_j^0\cup\CS_j^1}\in o(n)$.
\end{claim}
\begin{proof}
Consider the first virtual execution, where all honest parties start with input $0$.
Denote by $\JS'=\sset{j\in\JS\mid\ssize{\CS^0_j}\in\Theta(n)}$ the set of parties that receive a linear number of messages from $\sset{\Party_i}_{i\notin\JS}$ (\ie receive $\delta(n)$ messages for some $\delta\in\Theta(n)$). If $\ssize{\JS'}\in\Theta(n)$, \ie there are linear many parties that receive a linear number of messages, it must be that the number of messages sent from $\sset{\Party_i}_{i\notin\JS}$  to $\sset{\Party_j}_{j\in\JS}$ is quadratic. This will contradict to the assumption that every party in $\sset{\Party_i}_{i\notin\JS}$ only sends a sublinear number of messages.
Therefore, $\ssize{\JS'}\in o(n)$, and it holds that $\ssize{\CS^0_j}\in o(n)$ for a majority of $j\in\JS$. By an analogue argument, also in the second virtual execution, where all honest parties start with input $1$, it holds that $\ssize{\CS^1_j}\in o(n)$ for a majority of $j\in\JS$. Hence, there exists $j\in\JS$ for which $\ssize{\CS^0_j\cup\CS^1_j}\in o(n)$.
\QED
\end{proof}

The adversary proceeds by choosing uniformly at random $\is\in\JS$. If it holds that $\ssize{\CS_\is^0\cup\CS_\is^1}\geq\beta n/2$, the adversary aborts the attack and halts. By \cref{claim:lb_corruptions}, the adversary does not abort with probability at least $1/n$.
Next, the adversary chooses a random subset $\IS\subseteq[n]\setminus\sset{\is}$ of size $\beta n$, such that $\JS\cup\CS_\is^0\cup\CS_\is^1\setminus\sset{\is}\subseteq\IS$.
Denote by $\E$ the event where the adversary does not abort and that party $\Party_\is$ is isolated by $\faecomm^\ast$ \wrt the set of corrupted parties $\IS$ as defined above. By the definition of $\faecomm^\ast$ and by \cref{claim:lb_corruptions}, this event happens with inverse-polynomial probability. The attack defined below will be analyzed conditioned on the event $\E$.

\paragraph{The attack.}
We proceed by defining a series of hybrid experiments to contradict the \emph{validity} of $\pi$.
For the first claim, we define the adversarial strategy $\Adv_1$, where the corrupted parties are $\Party_i$ with $i\in\JS$. The parties in $\JS\setminus\sset{\is}$ do not send messages throughout the protocol, whereas party $\Party_\is$ does not send any message during the first part of the protocol, but in the last round sends messages as an honest party with input $0$ that was isolated in the first part.
\begin{claim}\label{claim:lb_one_corrupt}
Consider an execution of $\pi$ with $\Adv_1$, where all parties start with input bit $1$. Then, all honest parties output $1$ with all but negligible probability.
\end{claim}
\begin{proof}
The claim follows immediately by the \emph{validity} property of $\pi$.
\QED
\end{proof}

For the second claim, we define the adversarial strategy $\Adv_2$, where the set of corrupted parties is~$\IS$.
The parties in $\JS\setminus\sset{\is}$ do not send messages throughout the protocol, and the parties in $\IS\setminus\JS$ play honestly on input $1$, except that in the last round, the set of parties in $\CS_\is^0$ additionally sends the messages $\sset{\hat{m}_{i\to\is}^0}_{i\in\CS_\is^0}$ to $\Party_\is$.
\begin{claim}\label{claim:lb_second_adv}
Consider an execution of $\pi$ with $\Adv_2$, where party $\Party_\is$ starts with input bit $0$ and all other parties with input bit $1$. Then, conditioned on $\E$, all honest parties (including $\Party_\is$) output $1$ with all but negligible probability.
\end{claim}
\begin{proof}
Conditioned on $\E$, the view of all honest parties other than $\Party_\is$, is identically distributed as in \cref{claim:lb_one_corrupt}. It follows that every honest party but $\Party_\is$ will output $1$ except for negligible probability. By \emph{agreement}, $\Party_\is$ will also output $1$ except for negligible probability.
\QED
\end{proof}

Next, consider the adversarial strategy $\Adv_3$, where the set of corrupted parties is~$\IS$.
The parties in $\JS\setminus\sset{\is}$ do not send messages throughout the protocol, and the parties in $\IS\setminus\JS$ play honestly on input $0$, except that in the last round, the set of parties in $\CS_\is^1$ additionally sends the messages $\sset{\hat{m}^1_{i\to\is}}_{i\in\CS_\is^1}$ to $\Party_\is$.
\begin{claim}\label{claim:lb_third_adv}
Consider an execution of $\pi$ with $\Adv_3$ where all parties starts with input bit $0$. Then, conditioned on the event $\E$, all honest parties output $1$ with noticeable probability.
\end{claim}
\begin{proof}
We will show that, conditioned on $\E$, the view of $\Party_\is$ in this scenario will be distributed as in previous scenario with noticeable probability; hence, by \cref{claim:lb_second_adv}, party $\Party_\is$ will output $1$ with the same probability. By \emph{agreement} so will all other honest parties.

To analyze the view of $\Party_\is$ in the first scenario (where all parties outside of $\JS$ start with input $1$), let $\BS_1$ be the set of honest parties that send messages to $\Party_\is$ in the last round. Denote by $\sset{m^1_i}_{i\in\BS_1}$ the messages sent by these parties to $\Party_\is$.
The view of $\Party_\is$ consists of his input bit $0$, his random coins, the $\crs$, the messages $\sset{\hat{m}^0_{i\to\is}}_{i\in\CS_\is^0}$, and messages $\sset{m^1_i}_{i\in\BS_1}$.

To analyze the view of $\Party_\is$ in the second scenario (where all parties outside of $\JS$ start with input $0$), let $\BS_0$ be the set of honest parties that send messages to $\Party_\is$ in the last round. Denote by $\sset{m^0_i}_{i\in\BS_0}$ the messages sent by these parties to $\Party_\is$.
The view of $\Party_\is$ consists of his input bit $0$, his random coins, the $\crs$, the messages $\sset{\hat{m}^1_{i\to\is}}_{i\in\CS_\is^1}$, and messages $\sset{m^0_i}_{i\in\BS_0}$.

Recall that by \cref{claim:lb_corruptions}, when running two \emph{independent} executions of $\pi$ in which the parties in $\JS$ do not talk till the last round, the first where every $\Party_j$ with $j\in[n]\setminus\JS$ starts with $0$ and the second when every such $\Party_j$ starts with $1$, there exists $\js\in\JS$ such that $\Party_\js$ receives $o(n)$ messages in both executions with probability at least $1/n$.
Since the executions in the first and second scenarios are independent of each other and also of the two virtual executions run in the head of the adversary, it holds that there exists a party $\Party_\js$ with $\js\in\JS$ that receives $o(n)$ messages in each of the four executions with probability at least $1/n^2$.
Since $\is$ is chosen uniformly at random in $\JS$, it holds that the sizes of $\CS_\is^0$, $\CS_\is^1$, $\BS_0$, and $\BS_1$ are all is $o(n)$ with probability at least $1/n^3$. In this case, it holds that the pair of sets $\sset{\hat{m}^1_{i\to\is}}_{i\in\CS_\is^1}$ and $\sset{m^0_i}_{i\in\BS_0}$ is identically distributed as $\sset{\hat{m}^0_{i\to\is}}_{i\in\CS_\is^0}$ and $\sset{m^1_i}_{i\in\BS_1}$, and the view of $\Party_\is$ is identically distributed in both the first and second scenarios.
\QED
\end{proof}

Since by assumption, the event $\E$ occurs with inverse-polynomial probability, the attack succeeds with inverse-polynomial probability. This concludes the proof of \cref{thm:ba_lb_crs}.
\QED
\end{proof}

%% file: lb_proof_pki.tex
At a high level, the proof of the theorem considers an adversary that receives the public keys $(\vk_1,\ldots,\vk_n)$ of the PKI setup, where each $\vk_i$ is sampled with a secret $\sk_i$ \emph{independently} of other keys. Under the assumption that one-way functions do not exist, with noticeable probability the adversary can find a corresponding secret key $\widetilde{\sk}_i$ (\ie a pre-image) for every $\vk_i$, and then carry out the attack from \cref{sec:ba_lb_crs}. This intuition, however, is not sufficient for proving the theorem, since the distribution of randomly generated keys $\sset{(\vk_i,\sk_i)}_{i\in[n]}$ may be different than the distribution of the inverted keys $\sset{(\vk_i,\widetilde{\sk}_i)}_{i\in[n]}$. In this case, the simulated messages generated by the adversary when emulating the executions in its head may be different than those generated in the real protocol, and so honest parties can tell them apart.

To overcome this subtlety, recall that \citet{IL89} showed that the existence of \emph{distributional one-way functions} (functions for which it is hard to sample a uniform pre-image) implies the existence of one-way functions. Stated differently, if one-way functions do not exist, then for any polynomial $p(\cdot)$ and any polynomial-time computable function $f$, there exists a PPT algorithm $\Inv$ such that, for infinitely many $n$, the following distributions are $1/p(n)$-statistically close:
\begin{itemize}
    \item
    $\sset{(x,f(x))\mid x\gets \zn}$.
    \item
    $\sset{(\Inv(f(y)),y)\mid x\gets \zn, y=f(x)}$.
\end{itemize}
In this case, we say that $\Inv$ inverts $f$ with $1/p(n)$-statistical closeness. In case the distributions are identically distributed we call the inverter \emph{perfect} and denote it by $\PInv$.

\begin{proof}[Proof of \cref{thm:ba_lb_pki}]
Without loss of generality, in the following we consider $n=\secParam$.
The trusted PKI setup can be modeled by a trusted party that for every $i\in[n]$ samples uniformly random $r_i\in\zo^n$, computes a polynomial-time function $(\vk_1,\ldots,\vk_n)=\fpki(r_1,\ldots,r_n)$, where for every $i\in[n]$, $\vk_i=\fpki^i(r_i)$ for some function $\fpki^{i}$. The trusted party outputs to each party $\Party_i$ the random coins $r_i$ along with $(\vk_1,\ldots,\vk_n)$.
Denote by $1/p(n)$ the success probability of the attack in the proof of \cref{thm:ba_lb_crs} and let $\Inv$ be the inverter algorithm for $\fpki$ that is guaranteed to exist by~\cite{IL89} with $1/2p(n)$-statistical closeness under the assumption that one-way functions do not exist.

Let $\pi$ be a protocol in the trusted PKI and $\faecomm^\ast$-hybrid model that invokes $\faecomm^\ast$ for polynomially many rounds followed by a single \ptp round, and assume that the number of messages sent by every party in the last round is $o(n)$.
Following the lines of the proof of \cref{thm:ba_lb_crs}, we will construct an adversarial strategy that violates the \emph{validity} of $\pi$ with non-negligible probability.

\paragraph{Choosing the corrupted set.}
Initially, the adversary receives the public keys $(\vk_1,\ldots,\vk_n)$ from the trusted party modeling the trusted PKI, and computes $(\tilde{r}_1,\ldots, \tilde{r}_n)\gets\Inv(\vk_1,\ldots,\vk_n)$. Next, the adversary chooses a random subset $\JS\subseteq[n]$ of size $\beta n/2$ and simulates two executions of $\pi$ inside its head.
\begin{itemize}
    \item
    In the first execution, every party $\Party_i$ has input bit $0$ and receives $\tilde{r}_i$ and $(\vk_1,\ldots,\vk_n)$ from the trusted PKI. Every party $\Party_j$ with $j\in\JS$ is corrupted and does not send any message throughout the protocol. For every $j\in\JS$, denote the set of parties that sends messages to $\Party_j$ in the last \ptp round by $\CS_j^0$ and record the messages as $\sset{\hat{m}^0_{i\to j}}_{i\in\CS_j^0}$.
    \item
    In the second execution, every party $\Party_i$ has input bit $1$ and receives $\tilde{r}_i$ and $(\vk_1,\ldots,\vk_n)$ from the trusted PKI. Every party $\Party_j$ with $j\in\JS$ is corrupted and does not send any message throughout the protocol. For every $j\in\JS$, denote the set of parties that sends messages to $\Party_j$ in the last \ptp round by $\CS_j^1$ and record the messages as $\sset{\hat{m}^1_{i\to j}}_{i\in\CS_j^1}$.
\end{itemize}

\begin{claim}\label{claim:lb_pki_corruptions}
There exists $j\in\JS$ such that $\ssize{\CS_j^0\cup\CS_j^1}\in o(n)$, except for probability $1/2p(n)$.
\end{claim}
\begin{proof}
Consider a perfect inverter $\PInv$ for $\fpki$. In that case for every $j\in\JS$, the simulated messages by the adversary $\sset{\hat{m}^0_{i\to j}}_{i\in\CS_j^0}$ (\resp $\sset{\hat{m}^1_{i\to j}}_{i\in\CS_j^1}$) are identically distributed as the messages that $\Party_j$ receives in the last round in an honest execution where all parties in $[n]\setminus\JS$ have input $0$ (\resp $1$) and parties in $\JS\setminus\sset{j}$ are corrupted and do not send messages. Therefore, by an identical argument to \cref{claim:lb_corruptions}, there exists $j\in\JS$ such that $\ssize{\CS_j^0\cup\CS_j^1}\in o(n)$.

The claim follows since $\Inv$ is an inverter with $1/2p(n)$-statistical closeness.
\QED
\end{proof}

The adversary proceeds by choosing uniformly at random $\is\in\JS$, and as before, if $\ssize{\CS_\is^0\cup\CS_\is^1}\geq\beta n/2$, the adversary aborts the attack and halts. By \cref{claim:lb_pki_corruptions} the adversary does not abort with probability at least $1/n - 1/2p(n)$ (recall that by the proof of \cref{thm:ba_lb_crs}, $1/p(n)\leq 1/n$; hence, $1/n - 1/2p(n)>0$).
Next, the adversary chooses a random subset $\IS\subseteq[n]\setminus\sset{\is}$ of size $\beta n$, such that $\JS\cup\CS_\is^0\cup\CS_\is^1\setminus\sset{\is}\subseteq\IS$.
Denote by $\E$ the event where the adversary does not abort and that party $\Party_\is$ is isolated by $\faecomm^\ast$ \wrt the set of corrupted parties $\IS$ as defined above. By the definition of $\faecomm^\ast$ and by \cref{claim:lb_pki_corruptions}, this event happens with inverse-polynomial probability.

The rest of the proof proceeds exactly as in the proof of \cref{thm:ba_lb_crs}, with the only difference that the statistical distance of the PKI private keys in the protocol and those simulated by the adversary is bounded by $1/2p(n)$. Since the attack in the proof of \cref{thm:ba_lb_crs} succeeds with probability $1/p(n)$, it holds that $1/p(n) - 1/2p(n)$ is noticeable.
\QED
\end{proof}

%% file: SRDS_Constructions.tex
\section{Constructions of SRDS}\label{sec:SRDS:construction}

In \cref{sec:SRDS:owf}, we present an SRDS scheme with trusted PKI based on OWF, and in \cref{sec:SRDS:snarks}, an SRDS scheme with bare PKI based on proof-carrying data and CRH.

\input{SRDS-owf}

\input{SRDS-snarks}

%% file: SRDS-owf.tex
\subsection{SRDS from One-Way Functions\SUBSUBSEC}\label{sec:SRDS:owf}

\def\ThmSRDSOWF
{
Let $\beta<1/3$ be a constant.
Assuming the existence of one-way functions, there exists a $\beta n$-secure SRDS scheme in the trusted PKI model.
}
\begin{theorem}[\cref{thm:intro:srds_owf}, restated]\label{thm:srds_owf}
\ThmSRDSOWF
\end{theorem}

The main building block in our construction is an augmented version of digital signatures with the ability to obliviously sample a verification key without knowing the signing key. Note that by assuming secure erasures, or a trusted party that does not reveal the key-generation coins, our construction can be based on any digital signatures scheme.

\begin{definition} [signatures with oblivious key generation]\label{def:Okeygen}
A digital signature scheme $(\DSGen$, $\DSSign$, $\DSVerify)$ has \textsf{oblivious key generation} if there exists an algorithm $\DSOGen$ that on input the security parameter $1^\secParam$ outputs a key $\ovk$, such that the following hold:
\begin{itemize}[leftmargin=*]
	\item \textbf{Indistinguishability.}
     The distribution of $\vk$, where $(\vk,\sk)\gets\DSGen(1^\secParam)$, should be computationally indistinguishable from $\ovk$, where $\ovk\gets\DSOGen(1^\secParam)$.
	\item \textbf{Obliviousness.}
    A PPT adversary $\Adv$ can win the following game with negligible probability:
    \begin{enumerate}[leftmargin=*]
        \item
        Challenger computes $\ovk=\DSOGen(1^\secParam;r)$ and sends $(\ovk,r)$ to $\Adv$.
        \item
        $\Adv$ responds with a pair $(m,\sigma)$, and wins if $\DSVerify(\ovk,m,\sigma)=1$.
    \end{enumerate}
\end{itemize}
\end{definition}

\def\ClaimOBLKEYGEN
{
Assuming the existence of one-way functions, there exists a one-time digital signature scheme with oblivious key generation.
}
\begin{claim}\label{claim:obl_keygen}
\ClaimOBLKEYGEN
\end{claim}

\begin{proof}[Proof Sketch.]
Recall the one-time signatures of \citet{Lamport79} for $\ell$-bit messages. Given a one-way function $f$, the signing key consists of $2\ell$ random $\secParam$-bits strings $x_1^0,x_1^1,\ldots,x_\ell^0,x_\ell^1$ and the verification key is $y_1^0,y_1^1,\ldots,y_\ell^0,y_\ell^1$, where $y_i^b=f(x_i^b)$. A signature on a message $m=(m_1,\ldots,m_\ell)$ is $\sigma=(x_1^{m_1},\ldots,x_\ell^{m_\ell})$. To verify a signature $\sigma=(\sigma_1,\ldots,\sigma_\ell)$, check for each $i\in[\ell]$ if $f(\sigma_i)=f(x_i^{m_i})$.

By instantiating the one-way function with a length-doubling pseudorandom generator $G$, we can define the oblivious key-generation algorithm by sampling $2\ell$ random $2\secParam$-bit strings. Indistinguishability follows from the pseudorandomness of $G$, and obliviousness from its one-wayness.
\QED
\end{proof}

Note that via standard transformations (\eg \cite[Sec.\ 6.4]{Goldreich04}) the one-time signature construction above can be extended to multi-message signatures with oblivious key generation.

\paragraph{Overview of the construction.}
Our construction makes use of a digital signature scheme with oblivious key generation (\cref{def:Okeygen}). For each party toss a biased coin that outputs $\heads$ with probability $\ell/n$, for some $\ell=\omega(\log(n))$. If the output is $\heads$, sample standard signature keys $(\vk_i,\sk_i)\gets\DSGen(1^\secParam)$; otherwise, obliviously sample $\vk_i\gets\DSOGen(1^\secParam)$.
A signature on a message $m$ can be computed only by parties with a valid signing key.
The aggregation algorithm concatenates these valid signatures.\footnote{Since this aggregation process is deterministic, decomposing the algorithm is redundant -- we represent it by two algorithms for completeness, to make the syntax compatible with the BA protocol in \cref{sec:ba_from_srds}.}
Verification of a signature requires counting how many valid signatures were signed on the message. Since each signature in this construction encodes the index associated with the corresponding verification key and each aggregate/partially aggregate signature is essentially a concatenation of the base signatures, it is easy to see that in this construction, given a signature/aggregate signature, the maxima and minima associated with it can be easily determined.

The construction of the SRDS scheme is formally described in \cref{fig:srds_from_owf} and the proof of \cref{thm:srds_owf} can be found in \cref{sec:SRDS:owf_cont}.

\input{srds_from_owf}

%% file: srds_from_owf.tex
\begin{nfbox}{\Srds from one-way functions}{fig:srds_from_owf}
\begin{center}
\textbf{SRDS from OWF}
\end{center}
Let $(\DSGen,\DSOGen,\DSSign, \DSVerify )$ be a signature scheme with oblivious key generation, let $\alpha(n,\secParam)\in\poly(\log{n},\secParam)$, and let $\ell=\log^c(n)$ for some constant $c>1$.
\begin{itemize}[leftmargin=*]
	\item $\TSCR(1^\secParam,1^n):$ Output $\pp=1^\secParam$.
	
	\item $\TSGen(\pp):$
    Toss a biased coin that outputs $\heads$ with probability $\ell/n$. If the outcome is $\heads$, compute $(\vk,\sk)\gets \DSGen(1^\secParam)$; else compute $\vk\gets\DSOGen(1^\secParam)$ and set $\sk=\bot$. Output $(\vk,\sk)$.
	
	\item $\TSSignShare(\pp,i,\sk, m):$
    If $\sk\neq \bot$, compute $\signsig\leftarrow \DSSign(\sk,m)$ and set $\sigma=\sset{(i,m,\signsig)}$; otherwise, set $\sigma=\bot$. Output $\sigma$.
	
	\item $\TSAggr_1(\pp, \sset{\vk_1,\ldots,\vk_n}, m,  \sset{\sigma_1,\ldots,\sigma_q})$:
    Initialize $\setsig=\emptyset$. For every $i\in [q]$:
	\begin{itemize}[leftmargin=*]
        \item
        Parse $\sigma_i$ as a set of signature-tuples  $\sset{(i_j, m_{i_j},\signsig_{i_j})}_{i_j\in S_i}$ (for some set $S_i$).
		\item
        For every $(i_j, m_{i_j},\signsig_{i_j})\in\sigma_i$, check whether $m=m_{i_j}$ and $\DSVerify(\vk_{i_j},m,\signsig_{i_j})=1$ and $\nexists (i_j,\cdot,\cdot,\cdot)\in\setsig$.
        If so, set $\setsig=\setsig\cup\sset{(i_j, m,\signsig_{i_j})}$.
	\end{itemize}
    Denote by $\|\sigma\|$ the bit length of $\sigma$.
    If $\sum_{\sigma\in\setsig}\|\sigma\|\leq\alpha(n,\secParam)$, output $\setsig$; else, output $\bot$.

	\item $\TSAggr_2(\pp, m,\setsig):$
    If $\sum_{\sigma\in\setsig}\|\sigma\|\leq\alpha(n,\secParam)$, output $\setsig$; else, output $\bot$.
	
	\item $\TSVerify(\pp,\sset{\vk_1,\ldots,\vk_n},m,\sigma)$:
    Initialize $\verifyset=\emptyset$.
	\begin{itemize}[leftmargin=*]
        \item
        Parse $\sigma$ as a set of signature-tuples  $\sset{(i, m_{i},\signsig_{i})}_{i\in S}$ (for some set $S$).
        \item
        For each $(i, m_i,\signsig_i)\in\sigma$, if $m=m_i$ and $\DSVerify(\vk_i,m_i,\signsig_i)=1$ and $\nexists(i,\cdot,\cdot)\in\verifyset$,
        set $\verifyset=\verifyset\cup\sset{(i, m_i,\signsig_i)}$.
        \item
        If $\sum_{\sigma\in\verifyset}\|\sigma\|\leq\alpha(n,\secParam)$ and $\ssize{\verifyset}> \ell'/3$, where $\ell'=\ell/2$, output 1; else, output 0.
	\end{itemize}
\end{itemize}

\end{nfbox}

%% file: SRDS-snarks.tex
\subsection{SRDS from CRH and  SNARKs\SUBSUBSEC}\label{sec:SRDS:snarks}
The construction in \cref{sec:SRDS:owf} was in the trusted PKI model. In this section, we show how to construct SRDS in the bare PKI, albeit under stronger cryptographic assumptions.
Namely, we consider CRH and SNARKs with \textsf{linear extraction}, where the size of the extractor is linear in the size of the prover (\ie $\ssize{\mathbb{E}_{\PS^*}}\leq c\cdot \ssize{\PS^*}$ for some constant $c$). An extractability assumption of this kind has been considered in \cite{Val08,DFH12,GS14,BJPY18}.

\def\ThmSRDSSNARKs
{
Let $t<n/3$. Assuming the existence of CRH, digital signatures, and SNARKs with linear extraction, there exists a $t$-secure SRDS scheme in the CRS model with a bare PKI.
}
\begin{theorem}[\cref{thm:intro:srds_snarks}, restated]\label{thm:SNARKSconstruction}
\ThmSRDSSNARKs
\end{theorem}

The construction of the SRDS scheme is formally described in \cref{fig:srds_from_snarks} and the proof of \cref{thm:SNARKSconstruction} can be found in \cref{sec:SRDS:snark_cont}.

\paragraph{Overview of the construction.}
As discussed in the Introduction, a PCD system allows for propagation of information up the tree in a succinct and publicly verifiable way.
Having the parties  locally sign the message and keep track of the number of verified signatures aggregated so far via the PCD system, seems to capture most of our requirements for SRDS. However, in order to prevent an adversary from aggregating fake signatures or multiple copies of the same signature, we need to devise a mechanism of verifying the base signatures in the compliance predicate.

One approach is to hard-wire all verification keys into the compliance predicate and verify each base-level signature.
However, this will blow-up the size of the predicate to $O(n)$ and, as a result, the PCD-prover algorithm will run in time $O(n)$.
In this case, the scheme will no longer be succinct, as the algorithm $\TSAggr_2$ internally runs the PCD prover. Indeed, recall that in the BA protocol (in \cref{sec:app:ba:prot}) $\TSAggr_2$ is executed via an MPC protocol; hence, its complexity must be $\tilde{O}(1)$.

To get around this barrier, we use a \emph{Merkle tree} to hash all the verification keys; a Merkle tree enables a long string (here, the list of \emph{all} verification keys $(\vk_1,\dots,\vk_n)$) to be hashed to a short value in a committing way, such that one can prove inclusion of the key $\vk_i$ in the input string by providing an ``opening'' to $\vk_i$ in low complexity (here, logarithmic in $n$) (see \cref{sec:merkle-hash} for details).
Each incoming and outgoing PCD transcript will now contain this hash value $\hash_\vk$. The base-level transcript will also consist of:
\begin{enumerate}
	\item The signature $\gamma_i$ and corresponding verification key $k_i=\vk_i$.
	\item A Merkle proof $p_i$ certifying that $k_i$ is the \ith verification key in the computation of $\hash_\vk$.
\end{enumerate}

\noindent
The compliance predicate, in this case, will verify:
\begin{enumerate}
	\item The signature $\gamma_i$ \wrt the $k_i$.
	\item That $k_i$ is properly hashed in the Merkle tree.
\end{enumerate}
To prevent an adversary from using a different $\hash_\vk$ value, we add an additional check in the compliance predicate that the value of $\hash_\vk$ is consistent in all the incoming and outgoing transcripts. Finally, to prevent an adversary from potentially aggregating multiple copies of the same base signature, we encode a maxima $\indexMax$ and a minima $\indexMin$ of the indices of the keys used to sign the base signatures in each transcript of the PCD proof.

We proceed to give a more detailed overview of our construction.
Each base signature (and aggregate signature) corresponds to a ``truncated'' PCD transcript and a corresponding proof. For base signatures, this proof is set to $\bot$ and in the remaining aggregated signatures, this proof corresponds to a PCD proof. Each truncated transcript $z'=(m,c,\indexMax,\indexMin,\gamma)$ consists of a message $m$ over which the signature is computed, a counter $c$ to keep a count of the number of distinct keys used to sign this signature, a maxima $\indexMax$ and minima $\indexMin$ of the indices of the keys that signed the message, and a value $\gamma$ that in the base case is a signature on $m$ corresponding to $\vk_{\indexMax}$ and in all other cases is set to $\bot$.

Each party starts by locally signing the message using its signing key $\sk_i$ and preparing $z_i'$.
The algorithm $\TSAggr_1$ collects base signatures and/or partially aggregated signatures, checks for their validity and prepares their corresponding PCD transcripts.
For base signatures (where $z'=(m,1,i,i,\gamma)$ and $\pi=\bot$), $\TSAggr_1$ checks that $\gamma$ and $m$ verify \wrt $\vk_i$; if so, it prepares a Merkle proof $p$ for $\vk_i$ and the PCD transcript is set to $z=z||(\hash_\vk,\vk_i,p)$.
For partially aggregated signatures (where $z'=(m,c,\indexMax,\indexMin,\bot)$ and $\pi\neq\bot$), it completes the transcript by setting $z=z'||(\hash_\vk,\bot,\bot)$ and runs the PCD verification algorithm on $(z,\pi)$. The algorithm $\TSAggr_2$ computes the outgoing transcript that is compliant with the valid incoming PCD transcripts and computes a PCD proof certifying this, \ie that it is based on $c$ \emph{distinct and valid} individual signatures. Finally, to verify $(z',\pi)$, set $z=z'||(\hash_\vk,\bot,\bot)$, verify the PCD $(z,\pi)$, and count the total number of keys used for signing this signature.

\ifdefined\IsTrackChanges
\input{srds_from_snarks}

\fi

Since each signature/aggregate signature in this construction essentially consists of a transcript and a proof and the transcript encodes information about maxima and minima associated with the signature, it is easy to see that in this construction, given a signature/aggregate signature, the maxima and minima associated with it can be easily determined.

%% file: srds_from_snarks.tex
\begin{nfbox}{\Srds from CRH and SNARKs}{fig:srds_from_snarks}
\begin{center}
\textbf{SRDS from CRH and SNARKs}
\end{center}
\vspace{-1em}
Let $(\DSGen,\DSSign,\DSVerify)$ be a digital signature scheme, let $(\PCDGen$, $\PCDProve$, $\PCDVerify)$ be a publicly verifiable proof-carrying data (PCD) system for \emph{logarithmic}-depth \emph{polynomial}-size compliance predicates $\compliance$, and let $(\msetup,\mhash,\mproof,\mverify)$ be the Merkle hash proof system corresponding to a hash function $\hash$.
Let $\alpha(n,\secParam)\in\poly(\log{n},\secParam)$.

\vspace{-.5em}
\begin{itemize}[leftmargin=*]
	\item $\TSCR(1^\secParam)$:
    Sample $\seed\gets\msetup(1^\secParam)$ and PCD keys corresponding to a compliance predicate $\compliance$ (defined below), as
	$
    (\PCDsigma,\PCDtau)\gets\PCDGen(1^\secParam,\compliance).
    $

    \textbf{The predicate $\compliance$:}
    Given an input vector $\vec{z}_\inputvar$ of length $\ell$, such that for $j\in[\ell]$ the \jth entry of $\vec{z}_\inputvar$ is of the form $\vec{z}_\inputvar[j]=(m_{\inputvar,j},c_{\inputvar,j},\indexMax_{\inputvar,j},\indexMin_{\inputvar,j},\gamma_{\inputvar,j},\hash_{\vk,j},k_{\inputvar,j},p_{\inputvar,j})$, and output data of the form $z_\outputvar=(m_\outputvar,c_\outputvar,\indexMax_\outputvar,\indexMin_\outputvar,\gamma_{\outputvar},\hash_{\vk,\outputvar},k_\outputvar,p_\outputvar)$, the predicate $\compliance(\vec{z}_\inputvar,z_\outputvar)$ equals 1 iff:
	\begin{enumerate}
		\item For every $j\in[\ell]$, it holds that $\hash_{\vk,j}=\hash_{\vk,\outputvar}$.
		\item For every $j\in[\ell]$, if it is a base level (\ie if $\indexMax_{\inputvar,j}=\indexMin_{\inputvar,j}$ and $\gamma_{\inputvar,j}\neq \bot$), then $\DSVerify(k_{\inputvar,j},m_{\inputvar,j},\gamma_{\inputvar,j})=1$ and $\mverify(\seed,(\indexMax_{\inputvar,j}||k_{\inputvar,j}),\hash_{\vk,\outputvar},p_{\inputvar,j})=1$.
        \item It holds that $\indexMin_{\inputvar,\ell}\leq\indexMax_{\inputvar,\ell}$ and
        for every $j\in[\ell-1]$ that $\indexMin_{\inputvar,j}\leq\indexMax_{\inputvar,j}<\indexMin_{\inputvar,j+1}$, \ie $\indexMax$ of an input is greater than or equal to its $\indexMin$ and less than the $\indexMin$ of the next input.
        \item $\indexMin$ of the output transcript is equal to the $\indexMin$ of the first input, \ie
        $\indexMin_\outputvar=\indexMin_{\inputvar,1}$.
        \item $\indexMax$ of the output transcript is equal to the $\indexMax$ of the last input, \ie
        $\indexMax_\outputvar=\indexMax_{\inputvar,\ell}$.
        \item $c_\outputvar$ stores a count of the number of signatures aggregated so far, \ie
        $c_\outputvar=\sum_{j\in[\ell]}c_{\inputvar,j}$.
    \end{enumerate}
    The output is $\pp=(1^\secParam,\PCDsigma,\PCDtau,\seed)$.
	\item $\TSGen(\pp)$:
    Parse $\pp=(1^\secParam,\PCDsigma,\PCDtau,\seed)$, compute $(\vk,\sk)\gets \DSGen(1^\secParam)$, output $(\vk,\sk)$.
	\item $\TSSignShare(\pp,i,\sk_i, m_i)$:
	Compute $\gamma_i\gets\DSSign(\sk_i,m_i)$, set $z'=(m_i,1, i,i,\gamma_i)$, and output $\sigma=(z',\bot)$.
	\item $\TSAggr_1(\pp, \sset{\vk_1,\ldots,\vk_n}, m,\sset{\sigma_1,\ldots,\sigma_q})$:
    Parse $\pp=(1^\secParam,\PCDsigma,\PCDtau,\seed)$. Compute $\hash_\vk=\mhash(\seed,(1||\vk_1),\ldots,(n||\vk_n))$ and set $\setsig= \sset{\hash_\vk}$.
    For each $i\in [q]$ do the following:
	\begin{itemize}
		\item Parse $\sigma_i=( z_i',\pi_i)$ and $z_i'=(m_i,c_i,\indexMax_i,\indexMin_i,\gamma_i)$
		\item For
        base level (where $\indexMax_i=\indexMin_i$, $m=m_i$, $\pi_i=\bot$, $\gamma_i\neq\bot$ and $\DSVerify(\vk_{\indexMax_i},m_i,\gamma_i)=1$), compute $p_i=\mproof(\seed,(1||\vk_1),\ldots,(n||\vk_n),(\indexMax_i||\vk_{\indexMax_{i}}))$, prepare the transcript $z_i=z_i'||(\hash_\vk,\vk_{\indexMax_{i}},p_i)$ and set $\setsig=\setsig\cup\sset{(z_i,\pi_i)}$.
        \item Else, set $z_i=z_i'||(\hash_\vk,\bot,\bot)$ and check whether $\PCDVerify(\PCDtau, z_i, \pi_i)=1$ and $m=m_i$. If so, set $\setsig=\setsig\cup\sset{(z_i,\pi_i)}$.
	\end{itemize}
    If $\|\setsig\|\leq \alpha(n,\secParam)$,\footnote{$\|\setsig\|$ stands for the bit length of $\setsig$.} output $\setsig$ ; else, output $\bot$.
	\item $\TSAggr_2(\pp, m,\setsig)$:
	Parse $\pp=(1^\secParam,\PCDsigma,\PCDtau,\seed)$ and set $c_\outputvar=0$.
	Parse $\setsig=\sset{\hash_\vk,\ldots}$. For each $\sigma_i\in \setsig\setminus \sset{\hash_\vk}$, parse $\sigma_i=(z_i',\pi_i)$ and $z_i'=(m_i,c_i,\indexMax_i,\indexMin_i ,\gamma_i)$ and set $c_\outputvar=c_\outputvar+c_i$.
	Let $(z_{\inputvar,1},\pi_{\inputvar,1})$ be the first element in $\setsig$ where $z_{\inputvar,1}=(\cdot,\cdot,\cdot,\indexMin_{\inputvar,1},\cdot)$. Set $\indexMin_\outputvar=\indexMin_{\inputvar,1}$. Similarly, denote $u=\ssize{\setsig}$ and let $(z_{\inputvar,u},\pi_{\inputvar,u})$ be the last element in $\setsig$ where $z_{\inputvar,u}=(\cdot,\cdot,\indexMax_{\inputvar,u},\cdot,\cdot)$. Set $\indexMax_\outputvar=\indexMax_{\inputvar,u}$, set $z'_\outputvar=(m,c_\outputvar,\indexMax_\outputvar,\indexMin_\outputvar,\bot)$, and set $z_\outputvar=(z'_\outputvar,\hash_\vk,\bot,\bot)$. Compute $\pi_\outputvar\gets \PCDProve(\PCDsigma, \setsig, \linp=\bot, z_\outputvar)$ and output $\sigma=(z'_\outputvar,\pi_\outputvar)$.
	\item $\TSVerify(\pp, \sset{\vk_i}_{i\in[n]},m,\sigma)$:
	Parse $\pp=(1^\secParam,\PCDsigma,\PCDtau ,\seed)$, $\sigma=(z',\pi)$, $z'=(m',c,\indexMax,\indexMin,\gamma_i)$.
    Compute $\hash_\vk=\mhash((1||\vk_1),\ldots,(n||\vk_n);\seed)$ and $z=z'||(\hash_\vk,\bot,\bot)$. If $m'=m$, $\PCDVerify(\PCDtau, z, \pi)=1$, $c\geq n/3$, and $\|\sigma\|\leq \alpha(n,\secParam)$, output 1; else, output 0.
\end{itemize}
\end{nfbox}

%% file: relation-snargs.tex
\section{Connection with Succinct Arguments}\label{sec:succinct_arguments}

In \cref{sec:SRDS:construction}, we showed how to construct SRDS with a strong setup assumption (trusted PKI) from OWF, and with relatively weak setup assumptions (bare PKI) at the expense of strong, non-falsifiable, cryptographic assumptions (SNARKs with linear extractors).
A natural approach towards constructing SRDS that balances the cryptographic and setup assumptions, is to augment a multi-signature scheme with some method of convincing the verifier that sufficiently many parties contributed to the signing process. Indeed, multi-signatures are known to exist under standard falsifiable assumptions in the \emph{registered PKI} model~\cite{LOSSW13}. In this model each party locally generates its own keys (as with bare PKI) but to publish its verification key, the party must prove knowledge of the corresponding secret key, see~\cite{Boldyreva03,LOSSW13} and a discussion in~\cite{BN06}.

In this section, we discuss challenges toward such an approach, by showing that in some cases this \emph{necessitates} some form of succinct non-interactive arguments.
We begin in \cref{sec:snargs} by formalizing the notion of SNARGs for average-case instances of a language, and formalizing the notion of SRDS ``based on'' multi-signatures. Next, in \cref{sec:snark1}, we show that any SRDS based on LOSSW multi-signatures imply SNARGs for average-case instances of the Subset-Product problem. Finally, in \cref{sec:snark-2}, we explore hardness of various Subset-$f$ problems and their connection to SRDS based on more general multi-signature schemes.

\ifdefined\IsTrackChanges\else
\input{srds_from_snarks}
\fi

\subsection{Average-Case SNARGs and SRDS Based on Multi-signatures}\label{sec:snargs}

\paragraph{Average-case SNARGs.}
We consider a notion of SNARGs for \emph{average-case} instances of an NP language $\LS$. This constitutes a weaker primitive than standard SNARGs (as per \cite{BCCGLRT17}), which requires soundness against worst-case instances, and may be viewed as a variant of the notion for \emph{cryptographically hard languages} considered in \cite{BISW18}. An average-case SNARG for a language $\LS$ is parametrized by an efficiently sampleable distribution $\yesD$ over the instance-witness pairs in $\LS$, and an efficiently sampleable distribution $\noD$ over instances outside of $\LS$. In a similar way to regular SNARGs, average-case SNARGs consist of setup, prover, and verification algorithms. Intuitively, given any instance-witness pair $(x,w)$ in $\yesD$, the prover algorithm should output a verifying succinct proof with overwhelming probability. At the same time, it should be hard for an adversary to compute a verifying proof for a \emph{random} instance $x$ from $\noD$.

\begin{definition}[average-case SNARG for $(\yesD,\noD)$]
Let $\LS$ be an NP language associated with a relation $R_{\LS}$, and let $\yesD$ and $\noD$ be efficiently sampleable distributions over $(x,w)\in R_{\LS}$ and $x\notin\LS$, respectively. A succinct non-interactive argument system $\Pi$ for average-case $\LS$, parametrized by the distributions $(\yesD,\noD)$, is defined by PPT algorithms $(\ProofSetup,\ProofProve,\ProofVerify)$ as follows:
\begin{itemize}
	\item $\ProofSetup(1^\secParam, 1^n)\to\crs$. On input the security parameter $\secParam$ and the instance size $n$, the setup algorithm outputs a common reference string $\crs$.
	\item $\ProofProve(\crs, x,w)\to\pi$. On input the $\crs$ and an instance-witness pair $(x,w)\in R_{\LS}$, the prover algorithm outputs a proof $\pi$.
	\item  $\ProofVerify(\crs, x,\pi)\to b$. On input the $\crs$, an instance $x$, and a proof $\pi$, the verification algorithm outputs a bit $b\in\zo$.
\end{itemize}
\end{definition}

\noindent
We require the argument system to satisfy the following properties:
\begin{enumerate}
	\item \textbf{Succinctness:} $|\pi| = \poly(\log n,\secParam)$ for all $(x,w)\gets\yesD(1^n)$.
	\item \textbf{Completeness:} For any instance-witness pair $(x,w)$ in the support of $\yesD$, it holds that
	\begin{align*}
	\pr{\ProofVerify(\crs,x,\pi)=1~\mid~ \crs\gets\ProofSetup(1^\secParam,1^n), \pi\gets\ProofProve(\crs,x,w)}\geq 1-\negl(n,\secParam).
	\end{align*}

	\item \textbf{Average-case soundness:} For any non-uniform PPT prover $\PS^\ast$, it holds that
	\begin{align*}
	\pr{\ProofVerify(\crs,x,\pi)=1~\mid~ \crs\gets\ProofSetup(1^\secParam,1^n), x \gets\noD(1^n), \pi \gets \PS^*(\crs,x)}\leq\negl(n,\secParam).
	\end{align*}
\end{enumerate}

\paragraph{SRDS based on multi-signatures.}
We consider implications of SRDS constructions based on an underlying multi-signature scheme (see \cref{sec:ms}) in the following sense. While rigorously specifying the notion is rather involved, at a high level, such a  scheme is one that satisfies three natural properties:
\begin{enumerate}
	\item \textbf{Structure:} The aggregate SRDS signature is a pair $(\sigmams,\pi)$, where $\sigmams$ is a multi-signature and $\pi$ is some (small) auxiliary information (of size $\tilde{O}(1)$).
	\item \textbf{Completeness:} Given a valid multi-signature $\sigmams$ on a message $m$ corresponding to a sufficiently large subset of keys $\sset{\vk_i}_{i\in S}$, together with knowledge of the subset $S$, it is easy to compute a valid SRDS signature certifying $m$.
	\item \textbf{Soundness:} Given a set of honestly generated verification keys, it is difficult to output a verifying SRDS signature $(\sigmams,\pi)$ on a message $m$ such that the multi-signature $\sigmams$ does not verify on $m$ against any sufficiently large subset of keys.
\end{enumerate}
In order to prove that sufficiently many parties agree on a message $m$, it suffices to certify that there exists an $s$-size subset of parties (where $s$ is sufficiently large) who agree on the same message~$m$. Therefore, moving forward for SRDS based on multi-signatures, we only focus on proving that exactly $s$ parties agree on a particular message.
We now formalize SRDS based on multi-signatures.

\begin{definition}[SRDS based on multi-signatures]\label{def:SRDS_multisig}
An SRDS scheme $\Pi=(\TSCR(1^\secParam,1^n),$ $\TSGen,$ $\TSSignShare,$ $\TSAggr,$ $\TSVerify)$ with bare PKI, is \textsf{based on a multi-signature scheme} $(\MSSetup,\MSGen,\MSSign,\MSVerify, \MSCombine,\MSMverify)$ if there exists $\subsetsize(n)\in\Theta(n)$, for which the following hold:
\end{definition}
\begin{itemize}[leftmargin=*]
	\item \textbf{Structure.} The SRDS has the following structure:
	\begin{itemize}
		\item $\TSCR(1^\secParam,1^n)$: Outputs public parameters of the form $\ppsrds=(\ppms, \pp_2)$, where $\ppms\gets\MSSetup(1^\secParam)$ and $\pp_2$ are (potentially) additional public parameters.
		\item $\TSGen(\ppsrds)$: Parses $\ppsrds=(\ppms,\pp_2)$ and outputs $(\sk,\vk)\gets\MSGen(\ppms)$.
		\item $\TSAggr$: Any SRDS $\sigma$ output by $\TSAggr$ is of the form $(\sigmams,\pi)\in \XSms \times \zo^{\poly(\log n, \secParam)}$, where  $\sigmams\in\XSms$ is a multi-signature (in the support of $\MSCombine$).
	\end{itemize}
		
	\item \textbf{Completeness.}
	There exists a PPT algorithm $\progP$, such that with overwhelming probability (in $(n,\secParam)$) over honestly sampled $\ppsrds=(\ppms,\pp_2) \gets \TSCR(1^\secParam,1^n)$ and independently sampled verification keys $(\vk_i,\cdot)\gets\TSGen(\ppsrds)$ for $i\in[n]$, the following holds:
	
	Let $m\in\MS$, let $S\subseteq[n]$ of size $\subsetsize(n)$, and let $\sigmams\in\XSms$ be a multi-signature satisfying $\MSMverify(\ppms,\vk_1,\ldots,\vk_n, S,m,\sigmams)=1$. Then, with overwhelming probability (in $(n,\secParam)$) over the auxiliary information $\pi\gets\progP(\ppsrds,\vk_1,\ldots,\vk_n,S,m,\sigmams)$, it holds that $\TSVerify(\ppsrds,\vk_1,\ldots,\vk_n,m,(\sigmams,\pi))=1$.

\item \textbf{Soundness.}
Every non-uniform polynomial-time adversary $\Adv$ wins the following experiment with at most negligible probability (in $(n,\secParam)$):
\begin{enumerate}
	\item The challenger samples $\ppsrds=(\ppms,\pp_2) \gets  \TSCR(1^\secParam,1^n)$ and for every $i\in[n]$, sets $(\vk_i,\sk_i)\gets\TSGen(\ppsrds)$.
	\item The challenger gives \Adv the values $(\ppsrds,\vk_1,\ldots,\vk_n)$ and get back $(m, (\sigmams,\pi))$.
	\item The adversary wins the game if and only if $\TSVerify(\ppsrds,\vk_1,\ldots,\vk_n,m,(\sigmams,\pi))=1$ and, in addition, there does not exist a subset $S\subseteq[n]$ of size $\subsetsize(n)$, such that $\MSMverify(\ppms,\vk_1,\ldots,\vk_n, S,m,\sigmams)=1$.
\end{enumerate}
(Observe that the output of this experiment is not necessarily efficiently computable.)
\end{itemize}

Note that for our purposes it will suffice to consider soundness against an adversary $\Adv$ who does not have access to a subset of keys $\sset{\sk_i}_{i\in S}$ or to a signing oracle. This is a weaker requirement than a comparable soundness guarantee when given corrupted secret keys, which means a barrier against such primitive is stronger.

\subsection{Multi-signatures of \texorpdfstring{\citet{LOSSW13}}{Lg} and Subset-Product}\label{sec:snark1}
We proceed to show that any SRDS based on the multi-signature scheme of \citet{LOSSW13} (``LOSSW'') as defined above implies SNARGs for a natural average-case version of Subset-Product. Intuitively, an LOSSW multi-signature $\sigma_{\ms}$ for a set of parties $S\subseteq[n]$ is equivalent to a single signature under the \emph{product} of the verification keys $\prod_{i\in S}\vk_i$. In turn, existence of a large set of approving parties $S$ for $\sigma_{\ms}$ is equivalent to existence of a large set of verification keys $\{\vk_i\}_{i\in S}$ for which $\prod_{i\in S}\vk_i$ takes a particular desired value determined by $\sigma_{\ms}$.

\medskip
The multi-signature scheme of \citet{LOSSW13}, is based on the \emph{bilinear computational Diffie-Hellman} (BCDH) assumption, parametrized by a bilinear map $e:\GG\times\GG\to\GG_T$, where $\GG$ and $\GG_T$ are multiplicative cyclic groups of order $p$.
In \cref{lossw}, we formally describe the LOSSW multi-signature scheme; their scheme roughly works as follows:

\begin{construction}[LOSSW Multi-signatures \cite{LOSSW13}] ~
\begin{itemize}[leftmargin=*]
\item $\MSSetup(1^\secParam)$: The setup algorithm outputs public parameters $\ppms=(\GG, \GG_T, p, g, e)$.
\item $\MSGen(1^\secParam)$: The key-generation algorithm outputs a signing key $\sk\in\ZZ_p$, and the corresponding verification key $\vk\in\GG_T$, computed as $e(g,g)^{\sk}$.
\item $\MSSign(\ppms,\sk,m)$: There is a deterministic function $\fmsg$ that takes the public parameters $\ppms$ and the message $m\in\MS$ as input and outputs an element in $\GG$ (see \cref{lossw} for full specification). Given this $\fmsg$, a signature on $m$ with secret key $\sk$ is generated by sampling $r\gets\ZZ_p$ and computing $\sigma=(\sig_1,\sig_2)$ as follows:
    \[
    \sig_1=g^\sk\cdot \left(\fmsg(\ppms,m)\right)^r~\text{and}~\sig_2=g^r.
    \]
\item $\MSCombine(\ppms, \vk_1,\ldots,\vk_n,\sset{\sigma_i}_{i\in S},m)$: Given a set of individual signatures $\sset{\sigma_i}_{i\in S}$ on a message $m\in\MS$, the combine function parses each $\sigma_i$ as $(\sig_1^{(i)},\sig_2^{(i)})$ and computes:
    \[
    \sig_1=\prod_{i\in S} \sig_1^{(i)}\text{ and }\sig_2=\prod_{i\in S}\sig_2^{(i)}.
    \]
    The output is the combined multi-signature $\sigmams=(\sig_1,\sig_2)$.
	
\textbf{Remark.} Recall that $\sig_1^{(i)}$ and $\sig_2^{(i)}$ for each $i\in S$, is of the form $$\sig_1^{(i)}=g^{\sk_i}\cdot \left(\fmsg(\ppms,m)\right)^{r_i}~\text{and}~\sig_2^{(i)}=g^{r_i}$$
	for some $r_i\in\ZZ_p$. Therefore,
	\[
	\sig_1=g^{\sk^*}\cdot \left(\fmsg(\ppms,m)\right)^{r^*}\text{ and }\sig_2=g^{r^*},
	\]
	where $\sk^*=\sum_{i\in S}{\sk_i}$ and $r^*=\sum_{i\in S}r_i$. Note that the multi-signature $\sigmams=(\sig_1,\sig_2)$ can now be viewed as an individual signature on $m$ corresponding to secret key $\sk^*$ and randomness~$r^*$.
	\item $\MSMverify(\ppms,\vk_1,\ldots,\vk_n,S,m,\sigmams)$: Given a message $m$, a multi-signature $\sigmams$ and the corresponding set $S$ of verification keys, the verification algorithm outputs 1 if and only if
	\[
	e(\sig_1,g)\cdot e(\sig_2,\fmsg(\ppms,m))^{-1}=\prod_{i\in S} \vk_{i}.
	\]
	Note that the same algorithm can be used to verify individual signatures (with $S=\sset{i}$).
\end{itemize}
\end{construction}

\paragraph{Average-case subset-product problem.}
We proceed to show a connection between any SRDS based on LOSSW multi-signatures, and the following average-case version of the Subset-Product problem.

\begin{definition}[average-case subset-product problem]\label{def:subset-product}
Let $\subsetsize=\subsetsize(n)$ be an integer and let $\GG$ be a multiplicative group.
Given an instance $x=(a_1,\ldots,a_n,t)\in \GG^{n+1}$, the \textsf{$(\subsetsize,\GG)$-Subset-Product problem} is the problem of deciding if there exists a subset $S\subseteq[n]$ of size $\ssize{S}=\subsetsize$, such that $\prod_{i\in S}a_i=t$. All such instances are said to be in the $(s,\GG)$-Subset-Product language $\LS_{\times}$.

We consider the average-case version of this problem characterized by the following two distributions:
\begin{enumerate}
	\item $\yesD(1^n)\to(x,w)$:
    For $i\in[n]$, sample $a_i\in \GG$ uniformly at random. Sample a set $S\subseteq[n]$ of size $\subsetsize$ uniformly at random. Set $t=\prod_{i\in S}a_i$. Output $x=(a_1,\ldots,a_n,t)$ and $w=S$.
	\item $\noD(1^n)\to x$:
    For $i\in[n]$, sample $a_i\in \GG$ uniformly at random. Sample a target $t\in \GG$ uniformly at random. Output $x=(a_1,\ldots,a_n,t)$.
\end{enumerate}
\end{definition}
Note that for appropriate parameter regimes, $\noD$ yields instances $(x\notin\LS_{\times})$ with high probability.
For example, consider $s=n/2$ and $\GG=\ZZ_M^*$ for $M=2^{4n}$: the probability that there exists a subset $S$ such that $\prod_{i\in S} a_i$ is equal to a randomly chosen value $t$ is approximately $2^{-2n}$.
\paragraph{Remark.}
Subset-Product is a well-studied problem, with known NP-hardness results in the worst case, and conjectured hardness in the average-case version considered above.
\begin{itemize}[leftmargin=*]
	\item \textbf{Hardness of worst-case subset product:}
	For $\GG=\ZZ^*_M$ and $s\in\Theta(n)$, the hardness of \emph{worst-case} $(s,\GG)$-Subset-Product depends on the density $n/\log M$ of the instance. For $M=2^{\Theta(n)}$ (\ie $n/\log M\in\Theta(1)$), there exists $s\in\Theta(n)$ for which the $(\subsetsize,\ZZ^*_M)$-Subset-product is NP-complete~\cite{Karp72,GJ79,IN96,LPS10}.\footnote{The Subset-Sum problem for arbitrary subset-sizes in this parameter regime was amongst the initial 21 problems that were shown to be NP-complete by \citet{Karp72}. There exists a generic reduction to reduce any instance of 3-SAT to an instance of the Subset-Sum problem for arbitrary subset-sizes. A slight modification to this reduction shows that there exists some $s\in\Theta(n)$ for which $(s,\ZZ_M)$-Subset-Sum problem is also NP-complete. There also exists a generic reduction from any instance of the $(s,\ZZ_M)$-Subset-Sum problem, where $\ZZ_M$ is an additive group of order $M$, to an instance of the $(s,\GG_M)$-Subset-Product problem, where $\GG_M$ is a cyclic multiplicative group of order $M$ (with efficient exponentiation).}
	\item \textbf{Hardness of average-case subset product:}
	 The average-case version of Subset-Product is thought to be computationally hard when $n/\log M$ is constant or even $O(1/\log n)$ \cite{IN96}, with the best known algorithms requiring at least $2^{\Omega(n)}$ time. Hardness of distinguishing between distributions $\yesD$ and $\noD$ as above is used (in an indirect way) as a computational hardness assumption in an assortment of cryptographic systems \cite{IN96,AD97,Reg03,Reg05,Pei09,LPS10}.\footnote{This follows from the constructions in \cite{IN96,AD97,Reg03,Reg05,Pei09,LPS10} based on the Subset-Sum problem.}
\end{itemize}

\paragraph{LOSSW-based SRDS implies SNARGs for average-case subset-product.}
We now show that an SRDS scheme based on the LOSSW multi-signature scheme implies the existence of SNARGs for average-case Subset-Product over the target group $\GG_T$ for the underlying bilinear map $e:\GG \times \GG \to\GG_T$. Intuitively, the construction takes the following form.

Given an instance $x=(a_1,\ldots,a_n,t)\in\GG_T^{n+1}$ (coming from either $\yesD$ or $\noD$), we will interpret $a_1,\ldots,a_n$ and $a_{n+1}=t^{-1}$ as $n+1$ \emph{verification keys} $\sset{\vk_i}_{i\in[n+1]}$ for the SRDS scheme. The succinct proof certifying that $x$ is in the language will be an SRDS signature on a message $m\in\MS$ with respect to the set of parties $S\cup\sset{n+1}$ for which $\prod_{i\in S} a_i=t$. The scheme is succinct by construction; the required completeness and average-case soundness properties will hold as follows:
\begin{itemize}[leftmargin=*]
	\item \emph{Completeness:} If $x$ was generated as $(x,w)\gets\yesD$ for $w=S\subseteq[n]$, then by definition $\prod_{i\in S}a_i=t$ and consequently $t^{-1}\cdot\prod_{i\in S}a_i=1$. Knowledge of the corresponding set of verification keys thus enables the prover to generate a valid LOSSW multi-signature under these keys, using the \emph{trivial} $\sk^*=0$ for $\vk^*=\prod_{i\in S\cup\sset{n+1}}\vk_i=1$. By the completeness of the LOSSW-based SRDS (\cref{def:SRDS_multisig}), the prover can then translate this multi-signature to a valid SRDS.
	\item \emph{Average-case Soundness:} On the other hand, if $x$ was generated as $x\gets\noD$, then since $t$ and consequently $t^{-1}$ is uniform conditioned on $a_1,\ldots,a_n$, the resulting verification keys $\sset{\vk_i}_{i\in[n+1]}$ are \emph{jointly uniform}. Thus (for appropriate parameters $n$ and $\ssize{\GG_T}$), soundness of the argument system holds from the soundness property of the LOSSW-based SRDS (see \cref{def:SRDS_multisig}).
\end{itemize}

\begin{lemma}\label{lem:lossw-snargs}
Assume there exists an SRDS scheme based on the LOSSW multi-signature scheme, where $\TSSetup(1^\secParam,1^n)$ generates $\ppms=(\GG, \GG_T, p, g,e)$, as per \cref{def:SRDS_multisig}, with $n/\log|\GG_T|<1$. Let $0<\alpha<1$ be a constant and let $s(n)=\alpha \cdot n$.
Then, there exist SNARGs for average-case $(\subsetsize(n),\GG_T)$-Subset-Product (as defined in \cref{def:subset-product}).
\end{lemma}

\begin{proof}
We construct average-case SNARGs for $(\subsetsize,\GG_T)$-Subset-Product using an SRDS scheme based on the LOSSW multi-signature scheme as per \cref{def:SRDS_multisig}. Recall that the LOSSW multi-signature scheme is parametrized by a bilinear map $e:\GG\times\GG\to\GG_T$, where $\GG$ and $\GG_T$ are multiplicative cyclic groups of order $p$. Let $\MS$ be the message domain of the LOSSW multi-signature scheme.
\begin{enumerate}
    \item
    $\ProofSetup(1^\secParam,1^n):$ Run the setup of the SRDS scheme $\ppsrds=(\ppms,\pp_2)\gets\TSSetup(1^\secParam,1^n)$ and output $\crs=\ppsrds$.
 	\item $\ProofProve(\crs,x,w):$
    Given an average-case $\yes$ instance-witness pair, $(x,w)\gets\yesD(1^n)$ of the form $x=(a_1,\ldots,a_n,t)$ and $w=S$, proceed as follows:
    \begin{itemize}[leftmargin=*]
    	\item Parse $\crs=(\ppms,\pp_2)$ and interpret the set $\sset{a_1,\ldots,a_n,t^{-1}}$ as a set of $(n+1)$ verification keys $\sset{\vk_1,\ldots,\vk_{n+1}}$. Note that $\prod_{i\in S'}\vk_i=1$ for $S'=S\cup\sset{n+1}$.
    	\item Since in the LOSSW multi-signature scheme, (aggregate) verification key $\vk^*=\prod_{i\in S'}\vk_i$ corresponds to a valid signing key $\sk^*=\sum_{i\in S'}\sk_i$, where $\vk^*=e(g,g)^{\sk^*}$, it holds that if $\vk^*=1$, then $\sk^*=0$. Choose an arbitrary $m\in\MS$ and sample $r\gets\ZZ_p$. Compute an LOSSW signature $\sigmams=(\sig_1,\sig_2)$ on $m$ with respect to $\vk^*=1$, where
   	\[
   	\sig_1=g^{0}\cdot\left(\fmsg\left(\ppms,m\right) \right)^r \quad \text{ and } \quad \sig_2=g^r.
    \]
    \item  Use the algorithm $\progP$ (that is guaranteed to exist by \cref{def:SRDS_multisig}) to compute the auxiliary information $\pi\gets\progP(\ppsrds,\vk_1,\ldots,\vk_{n+1}, S', m,\sigmams)$ from the signature $\sigmams$ and the set $S'\subseteq[n+1]$.
    \item Finally output $(m,\sigmams,\pi)$.
    \end{itemize}

 	\item $\ProofVerify(\crs,x,(m,\sigmams,\pi)):$ Parse $x=(a_1,\ldots,a_n,t)$ and proceed as follows:
 	\begin{itemize}[leftmargin=*]
 		\item Parse $\crs=(\ppms,\pp_2)$ and interpret the set $\sset{a_1,\ldots,a_n,t^{-1}}$ as a set of $(n+1)$ verification keys $\sset{\vk_1,\ldots,\vk_{n+1}}$.
 		\item  Run the verification algorithm of the LOSSW multi-signature scheme with respect to combined verification key $\vk^*=1$, \ie parse $\sigmams=(\sig_1,\sig_2)$ and check if
 		\[
 		e(\sig_1,g)\cdot e(\sig_2,\fmsg(\ppms,m))^{-1}=1.
 		\]
 		In other words, compute
 		\[
 		b' = \Bigg\{\begin{array}{lr}
 		1, \text{ if } e(\sig_1,g)\cdot e(\sig_2,\fmsg(\ppms,m))^{-1}=1\\
 		0, \text{ otherwise } \\
 		\end{array}.
 		\]
 		\item Run the verification algorithm of the SRDS scheme on $(\sigmams,\pi)$ with respect to $m$:
        \[
        b\gets\TSVerify(\ppsrds,\vk_1,\ldots,\vk_{n+1},m,(\sigmams,\pi)).
        \]
 		\item Output $b\wedge b'$.
 	\end{itemize}
\end{enumerate}
We now argue succinctness, completeness, and average-case soundness for this construction:

\smallskip\noindent\textbf{Succinctness.}
Succinctness follows from the succinctness of the SRDS scheme.

\smallskip\noindent\textbf{Completeness.}
Given any average-case $\yes$ instance-witness pair $(x,w)\gets\yesD(1^n)$, with $x=(a_1,\ldots,a_n,t)$ and $w=S$, it holds that $\prod_{i\in S} a_{j} =t$ or equivalently $t^{-1}\cdot\prod_{i\in S} a_{j} =1$. Let $S'=S\cup\sset{n+1}$. Recall that in the LOSSW scheme
\[
\prod_{i\in S'} \vk_{i}
= \prod_{i\in S'} e(g,g)^{\sk_{i}}
= e(g,g)^{\sum_{i\in S'} \sk_{i}},
\]
where $\sk_i$ is the secret key associated with $\vk_i$. It follows that if $\prod_{i\in S'} \vk_{i} =1$ then $\sum_{i\in S'} \sk_{i} =0$.
Hence, $\sigmams=\left(\left(\fmsg\left(\ppms,m\right)\right)^r, g^r\right)$ is a valid multi-signature on $m$ with respect to $\sset{\sk_i}_{i\in S'}$.
Since the multi-signature verifies  $\MSMverify(\ppms,\vk_1,\ldots,\vk_{n+1}, S', m, \sigmams)=1$, completeness of SRDS based on a multi-signature scheme (see \cref{def:SRDS_multisig}) implies that the output of $\progP$, given this signature and $S'$, will be a valid SRDS signature.

\smallskip\noindent\textbf{Average-case soundness.}
Recall that each of the values $(a_1,\ldots,a_n,t)$ in $x\gets\noD(1^n)$ are sampled uniformly at random. Since  $t$ is a randomly sampled value, so is $t^{-1}$. Therefore, the verification keys $\sset{\vk_1,\ldots,\vk_n,\vk_{n+1}}$, where $\vk_i=a_i$ for $i\in[n]$ and $\vk_{n+1}=t^{-1}$, are uniformly distributed over $\GG_T^{n+1}$.
Since by assumption, $n/\log |\GG_T|<1$, it holds with overwhelming probability (bounded by ${2^{n+1}}/{|\GG_T|}$) that there does not exist a subset $S'\subseteq[n+1]$ of size $\subsetsize+1$, such that $\prod_{i\in S'}\vk_i=1$.

Given $(m,\sigmams,\pi)$, we check if: (1) $\sigmams$ is a valid multi-signature on $m$ with respect to $\vk^*=1$, and (2) if $(\sigmams,\pi)$ is a valid SRDS on $m$. Recall that given a multi-signature $\sigma_{\ms}=(\sig_1,\sig_2)$, a message $m$, the public parameters $\pp_{\ms}$, and a set of verification keys $\{\vk_i\}_{i\in S}$, the verification algorithm of the LOSSW multi-signature scheme checks if
\[
e(\sig_1,g)\cdot e(\sig_2,\fmsg(\ppms,m))^{-1}=\prod_{i\in S} \vk_{i}.
\]
In other words, given a valid multi-signature $\sigma_{\ms}$ on a message $m$, there exists a unique aggregate verification key $\prod_{i\in S} \vk_{i}$ for which $\sigma_{\ms}$ verifies. Therefore, if check (1) goes through, then $\vk^*=1$ is the only aggregate verification key for which $\sigma_{\ms}$ is a valid multi-signature on $m$. As argued earlier, with a high probability there does not exist a subset $S'\subseteq[n+1]$ such that $\Pi_{i\in S'}\vk_i=1$. Also, from the soundness of SRDS based on a multi-signature scheme (\cref{def:SRDS_multisig}), we know that if there does not exist a subset $S'\subseteq[n+1]$ of size $s+1$, such that $\sigmams$ is a valid multi-signature on $m$ with respect to $\sset{\vk_i}_{i\in S'}$, then the probability of an adversary computing a valid SRDS $(\sigmams,\pi)$ on a message $m$ is negligible. Soundness now follows from the soundness of SRDS based on a multi-signature scheme.
\QED
\end{proof}

\subsection{General Multi-Signatures and the Subset-\texorpdfstring{$f$}{Lg} Problem}\label{sec:snark-2}
Although the proof of \cref{lem:lossw-snargs} depends on the specific LOSSW multi-signature scheme, the overall approach only depends on certain properties of that scheme; in particular, there is no inherent reliance on the structure of \emph{multiplication} of keys and Subset-Product.
Motivated by this observation, in this section, we start by exploring hardness of \emph{Subset-$f$ problems} for a more general class of functions $f$, focusing on the class of \emph{elementary symmetric polynomials} $\phi_{\ell}$. We begin by demonstrating (worst-case) NP-hardness results for Subset-$\phi_{\ell}$. We then abstract out the properties used in \cref{lem:lossw-snargs} (deemed ``SNARG-compliance''), and show that existence of SRDS based on a SNARG-compliant multi-signature scheme implies existence of SNARGs for corresponding Subset-$\phi_{\ell}$ problems.

\paragraph{The subset-$f$ problem.}
We first define the following analogous variant of average-case Subset-Product problem for more general functions $f$. We restrict our attention to the natural setting of symmetric functions $f$; one can extend to arbitrary $f$, \eg given a canonical ordering of inputs.

\begin{definition}[average-case subset-$f$]\label{def:subset-f}
Let $\subsetsize=\subsetsize(n)$ be an integer, let $R$ be a ring, and let $f:R^\subsetsize\to R$ an efficiently computable symmetric function.
Given an instance $x=(a_1,\ldots,a_n,t)\in R^{n+1}$, the \textsf{$(\subsetsize,R)$-Subset-$f$ problem} is the problem of deciding if there exists a subset $S\subseteq[n]$ of size $\ssize{S}=\subsetsize$, such that $f((a_i)_{i\in S})=t$.
Such instances are said to be in the $(s,R)$-Subset-$f$ language $\LS_f$.

We consider the average-case version of this problem characterized by the following two distributions:
\begin{enumerate}
	\item $\yesD(1^n)\to(x,w)$: For each $i\in[n]$, sample $a_i\in R$ uniformly at random. Sample a set $S\subseteq[n]$ of size $\subsetsize$ uniformly at random. Set $t=f((a_i)_{i\in S})$.  Output $x=(a_1,\ldots,a_n,t)$, $w=S$.
	\item $\noD(1^n)\to x$:  For each $i\in[n]$, sample $a_i\in R$ uniformly at random. Sample a target $t\in R$ uniformly at random. Output $x=(a_1,\ldots,a_n,t)$.
\end{enumerate}
We also consider a variant of the $(s,R)$-Subset-$f$ problem,  where the instance does not include the size of the subset, \ie given an instance $x=(a_1,\ldots,a_n,t)\in R^{n+1}$, the $R$-{\sf Subset-$f$ problem} is the problem of deciding if there exists a subset $S\subseteq[n]$ of any size such that $f((a_i)_{i\in S})=t$.
\end{definition}

Note that Subset-$f$ is within NP for any function $f$ describable by a polynomial-size circuit. For appropriate parameter regimes, the hardness of Subset-$f$ problems depends on the function~$f$. In \cref{thm:sym-poly-nphard}, we show that for rings (of appropriate size) with Hadamard product, Subset-$f$ for all \emph{elementary symmetric polynomials} $f$ is NP-complete.

\paragraph{NP-hardness of subset-$\phi_{\ell}$.} Recall that Hadamard product (also known as entry-wise product) takes two vectors of the same dimension and produces another vector of matching dimension where the \ith element of the resulting vector is a product of the \ith elements of the two input vectors.
\begin{definition}[Hadamard product]
Let $\field$ be a field and let $\vec{a}=(a_1\ldots,a_n),\vec{b}=(b_1\ldots,b_n)\in \field^n$ be vectors of length $n$. The \textsf{Hadamard product} of $\vec{a}$ and $\vec{b}$ is the vector $\vec{a}\odot \vec{b}=(a_1 b_1,\ldots, a_n b_n)\in \field^n$.
\end{definition}

\noindent
We now define elementary symmetric polynomials.

\begin{definition}[elementary symmetric polynomials]
Let $n\in\NN$ and $\ell\in[n]$.
The elementary symmetric polynomial $\phi_\ell(x_1,\ldots,x_n)$ is defined as:
\[
\phi_\ell(x_1,\ldots,x_n)=\sum_{1\leq j_1<\ldots< j_\ell\leq n} x_{j_1}\cdot\ldots\cdot x_{j_\ell}.
\]
\end{definition}

In the following theorem, we show that for certain rings $R$ that admit Hadamard product, and any elementary symmetric polynomial $\phi_\ell$, the $(\subsetsize,R)$-Subset-$\phi_\ell$ problem is NP-complete. In particular, we show this for suitably sized rings of the form $R=\field^n$, where for $\ell=2$, the characteristic of the field must be at least 63 and for $\ell>2$, the characteristic of the field must be at least $\ell+2$.
\def\ThmSymPolyNPHard
{
There exists $\subsetsize(n)\in\Theta(n)$ such that, for any field $\field$ with $\characteristic(\field)\geq \indexMax(\ell+2,63)$, any ring $R=\field^n$ of size $\ssize{R}=2^{\Theta(n)}$ with Hadamard product, and any elementary symmetric polynomial $\phi_\ell$, the $(\subsetsize,R)$-Subset-$\phi_\ell$ problem is NP-complete.
}
\begin{theorem}\label{thm:sym-poly-nphard}
\ThmSymPolyNPHard
\end{theorem}

We next present a high-level overview of the proof; the full proof can be found in \cref{sec:sym-poly}. We start with a recap of the proof for NP-completeness of subset sum by \citet{Karp72}.
\paragraph{NP-completeness of subset sum.}
The  proof for NP-completeness of $\ZZ_M$-Subset-Sum~\cite{Karp72}\footnote{Recall that this is a variation of the $(s,\ZZ_M)$-Subset-sum problem, where the instance does not include the size of the subset, as defined in \cref{def:subset-f}.} shows a polynomial-time reduction from 3-SAT. At a high level, the reduction proceeds as follows: Given a 3-SAT instance with $N$ variables $\sset{x_i}_{i\in[N]}$ and $m$ clauses $\sset{C_j}_{j\in[m]}$, define a $\ZZ_M$-Subset-Sum instance with the following $2(N+m)$ numbers, each with $N+m$ digits, for $M\geq 10^{N+m}$:
\begin{enumerate}
	\item For each input $i\in[N]$, define two numbers $v_i$ and $v_i'$. The \ith  least significant digit of both these numbers is set to 1.  If $x_i\in C_j$, then the \iith{(N+j)} least significant digit of $v_i$ is set to 1, else if $\neg x_i\in C_j$, then the \iith{(N+j)} least significant digit of $v'_i$ is set to 1. The remaining digits in both these numbers are set to 0.
	\item For each clause $j\in[m]$, define two numbers $c_j^1$ and $c^2_j$. The \iith{(N+j)} least significant digit of both these numbers is set to 1 and all the remaining digits are set to 0.
	\item The target number $t$ is also an $(N+m)$-digit number in which the first $N$ digits are set to 1, while the remaining digits are set to 3.
\end{enumerate}

Intuitively, given a satisfying assignment for the 3-SAT instance, the corresponding witness for the $\ZZ_M$-Subset-sum instance includes the following: For each $i\in[N]$, it includes $v_i$ if $x_i=1$, and $v_i'$ if $x_i=0$. For each $j\in[m]$, it includes any one of $c^1_j$ or $c_j^2$ if there are two literals with value 1 in the \jth clause, and both $c^1_j$ or $c_j^2$ if there is only one literal with a value of 1 in the \jth clause.

\begin{proof}[Proof sketch of \cref{thm:sym-poly-nphard}]
We extend this reduction to show that $R$-Subset-$\phi_\ell$ for $\ell>1$ is also NP-complete, where $R$ is a ring of appropriate size with Hadamard product.
Each of the $a_i$ (for $i\in[n]$) elements and the target value $t$ in an instance of $R$-Subset-$\phi_{\ell}$ is an element in $R$ and thereby a vector of elements in $\field$. Unlike simple addition, since $\phi_{\ell}$ is a sum of products, if (any) \kth entry in the target value is a non-zero element in $\field$, the solution to a  $\yes$ instance of $R$-Subset-$\phi_{\ell}$ must consist of at least $\ell$ elements with  non-zero \kth entries.
Therefore, depending on $\ell$, we need to define additional elements in the reduction.
We give an overview of our reduction from any 3-SAT instance to $R$-Subset-$\phi_\ell$ for $\ell\geq 3$; the special case of $\ell=2$ requires a slight modification that is addressed in \cref{sec:sym-poly}. In a similar way to Subset-Sum, this reduction can also be adjusted to show that there exists $s\in\Theta(n)$, for which $(s,R)$-Subset-$\phi_\ell$ problem is also NP-complete, which is sketched in \cref{sec:sym-poly}.
\input{reduction}

Given a 3-SAT instance with $N$ variables $\sset{x_i}_{i\in[N]}$ and $m$ clauses $\sset{C_j}_{j\in[m]}$, define a $R$-Subset-$\phi_{\ell}$ instance with $\ell+2N+(\ell-1) m$ elements, where $R=\field^{1+N+m}$. As shown in \cref{fig:subset-f-reduction}, each of these elements is a vector of $1+N+m$ elements in the field $\field$ and are defined as follows:
\begin{itemize}[leftmargin=*]
	\item An element $\alpha_0\in R$, whose first entry is 1. All the remaining entries in $\alpha_0$ correspond to 0.
	\item For each $k\in[\ell-1]$, define $\alpha_k\in R$, whose first $N+1$ entries correspond to 1, and the remaining entries correspond to 0.
    \item For each $i\in[N]$, define two elements $v_i\in R$ and $v_i'\in R$. The \iith{(1+i)} entry of both these numbers is set to 1. If $x_i\in C_j$, then the \iith{(1+N+j)} entry of $v_i$ is set to 1, else if $\neg x_i\in C_j$, then the \iith{(1+N+j)} entry of $v'_i$ is set to 1. All the remaining entries correspond to 0.
	\item For each $j\in[m]$ and $k\in[\ell-1]$, define element $c_j^k\in R$. The \iith{(1+N+j)} entry in $c_j^k$ corresponds to 1 and the remaining entries correspond to 0.
	\item The target element $t$ is also a vector of $1+N+m$ elements in $\field$, with all its entries set to 1.
\end{itemize}

Now, given a satisfying assignment for the 3-SAT instance, the corresponding witness for the $R$-Subset-$\phi_\ell$ instance includes the following: It includes $\alpha_0$ and each $\alpha_k$ for $k\in[\ell-1]$. For each $i\in[N]$, it includes $v_i$ if $x_i=1$, and $v_i'$ if $x_i=0$. For each $j\in[m]$, it includes any $\ell-3$ of the elements $c_j$ if all three literals in the \jth clause have value 1, else if any two literals have value 1 then it includes any $\ell-2$ of the elements $c_j$ and if only one of the literals has value 1 then all the $\ell-1$ elements $c_j$ are included in the witness. This guarantees that the value 1 appears precisely $\ell$ times in the column of each satisfied clause, so that $\phi_{\ell}$ will evaluate to the target value 1 in these positions.

Similarly for soundness, a valid witness $S$ for the $R$-Subset-$\phi_\ell$ instance must include $a_0,\ldots,a_{\ell-1}$ in order to get $\ell$ times the value 1 in the first column. Apart from $a_1,\ldots,a_{\ell-1}$, the only other elements that have the value 1 in the next $N$ columns are $v_i$ and $v'_i$. For each $i\in[N]$, if both $v_i$ and $v'_i$ are included in the set $S$, a total of $\ell+1$ elements in $S$ will have value 1 in the \iith{(i+1)} column. The \iith{(i+1)} entry in the result obtained by applying $\phi_{\ell}$ over such a set is $\ell+1$. Since the characteristic of the field $\field$ is at least $\ell+2$, we know that $\ell+1\neq 1$.
Therefore, $S$ can either contain $v_i$ (implying $x_i=1$) or $v'_i$ (implying $\neg x_i=1$) for each $i\in[N]$, but not both.  For each of the last $m$ columns, $S$ can contain some or all of the elements $c_j$ (for each $j\in[m]$). But since this set of $c_j$ elements can only contribute at most $\ell-1$ times the value 1 in the \iith{(1+N+j)} column,
we need at least one of the $v$ or $v'$ elements to contribute a 1 value to that column, in order to get a non-zero \iith{(1+N+j)} entry in the result of $\phi_{\ell}$.
This guarantees at least one variable with a value of 1 in each clause. We give a full proof of completeness and soundness for this reduction \cref{sec:sym-poly}.
\QED
\end{proof}

\paragraph{SNARG-compliant multi-signatures and subset-$\phi_\ell$.}  We now identify the properties of the LOSSW multi-signature scheme used in \cref{lem:lossw-snargs} to provide the connection with average-case SNARGs. Roughly, these properties are:
\begin{itemize}[leftmargin=*]
	\item Verification keys are sampled independently and uniformly from the key-space of the multi-signature scheme. This property is important for arguing soundness in \cref{lem:lossw-snargs}.
	\item The verification algorithm with keys $\sset{\vk_i}_{i\in S}$, is equivalent to the verification algorithm with a single \emph{aggregate} key $\vkagg=\prod_{i\in S}\vk_i$. In other words, there exists a \emph{key-aggregation function} $\fagg$ (\eg $\fagg=\prod$ in the LOSSW multi-signature scheme), such that the verification algorithm can be decomposed into first applying $\fagg$ over the set of keys to obtain an aggregate key $\vkagg$ and then running some residual function $\MSVerifyAgg$ to perform the remaining verification with respect to $\vkagg$.
	\item Given a valid multi-signature $\sigma_{\ms}$ on a message $m$, there exists a \emph{unique and well-defined} aggregate key $\vk$ for which the residual function $\MSVerifyAgg$ (as defined in the previous bullet) outputs~1. Moreover, this aggregate key is easy to compute. For example, for LOSSW, this property is crucially used for arguing soundness in \cref{lem:lossw-snargs}.
	\item And finally, there exist degenerate keys $\skdeg$ and $\vkdeg$ (\eg $\skdeg=0\in\GG$ and $\vkdeg=1\in\GG_T$ in the LOSSW multi-signature scheme) that allow forging a multi-signature on any message. This property is used in the completeness argument in \cref{lem:lossw-snargs}.
\end{itemize}
We call multi-signature schemes that satisfy these properties as \emph{SNARG-compliant} multi-signature schemes. We formally define this notion in \cref{def:SRDScompliant} in \cref{sec:snarg-compliant}.
Finally, by using the properties of a SNARG-compliant multi-signature scheme, we are able to prove a generalized version of \cref{lem:lossw-snargs}. Namely, we show in \cref{lem:snarg-compliant} that an SRDS scheme based on a SNARG-compliant multi-signature scheme with key-aggregation function $\fagg=\phi_{\ell}$, implies SNARGs for average-case Subset-$\phi_{\ell}$.

%% file: reduction.tex
\begin{figure}
\centering
\begin{tikzpicture}
\node at (11.3,6) {$\alpha_0$};
\node at (12.05,6.5) {$1$};
\draw [draw=black] (11.8, 5.75) rectangle ++(.5,.5) node[pos=.5] {$1$};

\draw[draw=black] (12.5,5.75) rectangle ++(.5,.5) node[pos=.5] {$0$};
\draw[draw=black] (13,5.75) rectangle ++(.5,.5) node[pos=.5] {$0$};
\draw[draw=black] (13.5,5.75) rectangle ++(.5,.5) node[pos=.5] {$0$};
\draw[draw=black] (14,5.75) rectangle ++(.5,.5) node[pos=.5] {$0$};
\draw [->] (14.9,6.5) --  (17,6.5);
\node at (14.75,6.5) {$N$};
\draw [<-] (12.5,6.5) --  (14.6,6.5);
\draw[draw=black] (14.5,5.75) rectangle ++(.5,.5) node[pos=.5] {$0$};
\draw[draw=black] (15,5.75) rectangle ++(.5,.5) node[pos=.5] {$0$};
\draw[draw=black] (15.5,5.75) rectangle ++(.5,.5) node[pos=.5] {$0$};
\draw[draw=black] (16,5.75) rectangle ++(.5,.5) node[pos=.5] {$0$};
\draw[draw=black] (16.5,5.75) rectangle ++(.5,.5) node[pos=.5] {$0$};

\draw[draw=black] (17.2,5.75) rectangle ++(.5,.5) node[pos=.5] {$0$};
\draw[draw=black] (17.7,5.75) rectangle ++(.5,.5) node[pos=.5] {$0$};
\draw[draw=black] (18.2,5.75) rectangle ++(.5,.5) node[pos=.5] {$0$};
\draw[draw=black] (18.7,5.75) rectangle ++(.5,.5) node[pos=.5] {$0$};
\draw[draw=black] (19.2,5.75) rectangle ++(.5,.5) node[pos=.5] {$0$};
\draw [->] (20.2,6.5) --  (22.7,6.5);
\node at (19.95,6.5) {$m$};
\draw [<-] (17.2,6.5) --  (19.7,6.5);
\draw[draw=black] (19.7,5.75) rectangle ++(.5,.5) node[pos=.5] {$0$};
\draw[draw=black] (20.2,5.75) rectangle ++(.5,.5) node[pos=.5] {$0$};
\draw[draw=black] (20.7,5.75) rectangle ++(.5,.5) node[pos=.5] {$0$};
\draw[draw=black] (21.2,5.75) rectangle ++(.5,.5) node[pos=.5] {$0$};
\draw[draw=black] (21.7,5.75) rectangle ++(.5,.5) node[pos=.5] {$0$};
\draw[draw=black] (22.2,5.75) rectangle ++(.5,.5) node[pos=.5] {$0$};

\node at (10.55,5.25) {$\alpha_1,\ldots,\alpha_{\ell-1}$};
\draw [draw=black] (11.8, 5) rectangle ++(.5,.5) node[pos=.5] {$1$};

\draw[draw=black] (12.5,5) rectangle ++(.5,.5) node[pos=.5] {$1$};
\draw[draw=black] (13,5) rectangle ++(.5,.5) node[pos=.5] {$1$};
\draw[draw=black] (13.5,5) rectangle ++(.5,.5) node[pos=.5] {$1$};
\draw[draw=black] (14,5) rectangle ++(.5,.5) node[pos=.5] {$1$};
\draw[draw=black] (14.5,5) rectangle ++(.5,.5) node[pos=.5] {$1$};
\draw[draw=black] (15,5) rectangle ++(.5,.5) node[pos=.5] {$1$};
\draw[draw=black] (15.5,5) rectangle ++(.5,.5) node[pos=.5] {$1$};
\draw[draw=black] (16,5) rectangle ++(.5,.5) node[pos=.5] {$1$};
\draw[draw=black] (16.5,5) rectangle ++(.5,.5) node[pos=.5] {$1$};

\draw[draw=black] (17.2,5) rectangle ++(.5,.5) node[pos=.5] {$0$};
\draw[draw=black] (17.7,5) rectangle ++(.5,.5) node[pos=.5] {$0$};
\draw[draw=black] (18.2,5) rectangle ++(.5,.5) node[pos=.5] {$0$};
\draw[draw=black] (18.7,5) rectangle ++(.5,.5) node[pos=.5] {$0$};
\draw[draw=black] (19.2,5) rectangle ++(.5,.5) node[pos=.5] {$0$};
\draw[draw=black] (19.7,5) rectangle ++(.5,.5) node[pos=.5] {$0$};
\draw[draw=black] (20.2,5) rectangle ++(.5,.5) node[pos=.5] {$0$};
\draw[draw=black] (20.7,5) rectangle ++(.5,.5) node[pos=.5] {$0$};
\draw[draw=black] (21.2,5) rectangle ++(.5,.5) node[pos=.5] {$0$};
\draw[draw=black] (21.7,5) rectangle ++(.5,.5) node[pos=.5] {$0$};
\draw[draw=black] (22.2,5) rectangle ++(.5,.5) node[pos=.5] {$0$};

\node at (11.3,4.5) {$v_i$};
\node at (11.3,3.25) {$v'_i$};
\draw [-]
(10.75,4.6) -- (10.75,3.15) node [black,midway,xshift=-0.6cm]
{\footnotesize $\forall i\in [N]~~$  };
\node at (10.51,2.4) {$\forall j\in[m]$};
\node at (10.51,1.9) {$c_j^1,\ldots,c_j^{\ell-1}$};
\node at (11.3,0.75) {$t$};
\draw [draw=black] (11.8, 4.25) rectangle ++(.5,.5) node[pos=.5] {$0$};

\draw[draw=black] (12.5,4.25) rectangle ++(.5,.5) node[pos=.5] {$0$};
\draw[draw=black] (13,4.25) rectangle ++(.5,.5) node[pos=.5] {$0$};
\draw[draw=black] (13.5,4.25) rectangle ++(.5,.5) node[pos=.5] {$0$};
\draw[draw=black] (14,4.25) rectangle ++(.5,.5) node[pos=.5] {$1$};
\node at (14.25,3.95) {$i$};
\draw [<-] (14,3.95) --  (14.15,3.95);
\draw [->] (14.35,3.95) --  (14.5,3.95);
\draw[draw=black] (14.5,4.25) rectangle ++(.5,.5) node[pos=.5] {$0$};
\draw[draw=black] (15,4.25) rectangle ++(.5,.5) node[pos=.5] {$0$};
\draw[draw=black] (15.5,4.25) rectangle ++(.5,.5) node[pos=.5] {$0$};
\draw[draw=black] (16,4.25) rectangle ++(.5,.5) node[pos=.5] {$0$};
\draw[draw=black] (16.5,4.25) rectangle ++(.5,.5) node[pos=.5] {$0$};

\draw[draw=black] (17.2,4.25) rectangle ++(.5,.5) node[pos=.5] {$0$};
\draw[draw=black] (17.7,4.25) rectangle ++(.5,.5) node[pos=.5] {$1$};
\node at (17.95,3.95) {$p$};
\draw [<-] (17.7,3.95) --  (17.85,3.95);
\draw [->] (18.05,3.95) --  (18.2,3.95);
\draw[draw=black] (18.2,4.25) rectangle ++(.5,.5) node[pos=.5] {$0$};
\draw[draw=black] (18.7,4.25) rectangle ++(.5,.5) node[pos=.5] {$0$};
\draw[draw=black] (19.2,4.25) rectangle ++(.5,.5) node[pos=.5] {$0$};
\draw[draw=black] (19.7,4.25) rectangle ++(.5,.5) node[pos=.5] {$0$};
\draw[draw=black] (20.2,4.25) rectangle ++(.5,.5) node[pos=.5] {$0$};
\draw[draw=black] (20.7,4.25) rectangle ++(.5,.5) node[pos=.5] {$0$};
\draw[draw=black] (21.2,4.25) rectangle ++(.5,.5) node[pos=.5] {$1$};
\node at (21.45,3.95) {$q$};
\draw [<-] (21.2,3.95) --  (21.35,3.95);
\draw [->] (21.55,3.95) --  (21.7,3.95);
\draw[draw=black] (21.7,4.25) rectangle ++(.5,.5) node[pos=.5] {$0$};
\draw[draw=black] (22.2,4.25) rectangle ++(.5,.5) node[pos=.5] {$0$};
\node at (23.9, 4.5) {if $x_i\in C_p$};
\node at (23.9, 3.95) {and $x_i\in C_q$};

\draw [draw=black] (11.8, 3) rectangle ++(.5,.5) node[pos=.5] {$0$};

\draw[draw=black] (12.5,3) rectangle ++(.5,.5) node[pos=.5] {$0$};
\draw[draw=black] (13,3) rectangle ++(.5,.5) node[pos=.5] {$0$};
\draw[draw=black] (13.5,3) rectangle ++(.5,.5) node[pos=.5] {$0$};
\draw[draw=black] (14,3) rectangle ++(.5,.5) node[pos=.5] {$1$};
\node at (14.25,2.7) {$i$};
\draw [<-] (14,2.7) --  (14.15,2.7);
\draw [->] (14.35,2.7) --  (14.5,2.7);
\draw[draw=black] (14.5,3) rectangle ++(.5,.5) node[pos=.5] {$0$};
\draw[draw=black] (15,3) rectangle ++(.5,.5) node[pos=.5] {$0$};
\draw[draw=black] (15.5,3) rectangle ++(.5,.5) node[pos=.5] {$0$};
\draw[draw=black] (16,3) rectangle ++(.5,.5) node[pos=.5] {$0$};
\draw[draw=black] (16.5,3) rectangle ++(.5,.5) node[pos=.5] {$0$};

\draw[draw=black] (17.2,3) rectangle ++(.5,.5) node[pos=.5] {$1$};
\draw[draw=black] (17.7,3) rectangle ++(.5,.5) node[pos=.5] {$0$};
\draw[draw=black] (18.2,3) rectangle ++(.5,.5) node[pos=.5] {$0$};
\draw[draw=black] (18.7,3) rectangle ++(.5,.5) node[pos=.5] {$0$};
\draw[draw=black] (19.2,3) rectangle ++(.5,.5) node[pos=.5] {$0$};
\draw[draw=black] (19.7,3) rectangle ++(.5,.5) node[pos=.5] {$1$};
\draw[draw=black] (20.2,3) rectangle ++(.5,.5) node[pos=.5] {$0$};
\draw[draw=black] (20.7,3) rectangle ++(.5,.5) node[pos=.5] {$0$};
\draw[draw=black] (21.2,3) rectangle ++(.5,.5) node[pos=.5] {$0$};
\draw[draw=black] (21.7,3) rectangle ++(.5,.5) node[pos=.5] {$0$};
\draw[draw=black] (22.2,3) rectangle ++(.5,.5) node[pos=.5] {$0$};
\node at (17.45,2.7) {$r$};
\draw [<-] (17.2,2.7) --  (17.35,2.7);
\draw [->] (17.55,2.7) --  (17.7,2.7);
\node at (19.95,2.7) {$s$};
\draw [<-] (19.7,2.7) --  (19.85,2.7);
\draw [->] (20.05,2.7) --  (20.2,2.7);
\node at (24, 3.25) {if $\neg x_i\in C_r$};
\node at (24, 2.7) {and $\neg x_i\in C_s$};

\draw [draw=black] (11.8, 1.75) rectangle ++(.5,.5) node[pos=.5] {$0$};

\draw[draw=black] (12.5,1.75) rectangle ++(.5,.5) node[pos=.5] {$0$};
\draw[draw=black] (13,1.75) rectangle ++(.5,.5) node[pos=.5] {$0$};
\draw[draw=black] (13.5,1.75) rectangle ++(.5,.5) node[pos=.5] {$0$};
\draw[draw=black] (14,1.75) rectangle ++(.5,.5) node[pos=.5] {$0$};
\draw[draw=black] (14.5,1.75) rectangle ++(.5,.5) node[pos=.5] {$0$};
\draw[draw=black] (15,1.75) rectangle ++(.5,.5) node[pos=.5] {$0$};
\draw[draw=black] (15.5,1.75) rectangle ++(.5,.5) node[pos=.5] {$0$};
\draw[draw=black] (16,1.75) rectangle ++(.5,.5) node[pos=.5] {$0$};
\draw[draw=black] (16.5,1.75) rectangle ++(.5,.5) node[pos=.5] {$0$};

\draw[draw=black] (17.2,1.75) rectangle ++(.5,.5) node[pos=.5] {$0$};
\draw[draw=black] (17.7,1.75) rectangle ++(.5,.5) node[pos=.5] {$0$};
\draw[draw=black] (18.2,1.75) rectangle ++(.5,.5) node[pos=.5] {$1$};
\draw[draw=black] (18.7,1.75) rectangle ++(.5,.5) node[pos=.5] {$0$};
\draw[draw=black] (19.2,1.75) rectangle ++(.5,.5) node[pos=.5] {$0$};
\draw[draw=black] (19.7,1.75) rectangle ++(.5,.5) node[pos=.5] {$0$};
\draw[draw=black] (20.2,1.75) rectangle ++(.5,.5) node[pos=.5] {$0$};
\draw[draw=black] (20.7,1.75) rectangle ++(.5,.5) node[pos=.5] {$0$};
\draw[draw=black] (21.2,1.75) rectangle ++(.5,.5) node[pos=.5] {$0$};
\draw[draw=black] (21.7,1.75) rectangle ++(.5,.5) node[pos=.5] {$0$};
\draw[draw=black] (22.2,1.75) rectangle ++(.5,.5) node[pos=.5] {$0$};
\node at (18.45,1.45) {$j$};
\draw [<-] (18.2,1.45) --  (18.35,1.45);
\draw [->] (18.55,1.45) --  (18.7,1.45);


\draw [draw=black] (11.8, 0.5) rectangle ++(.5,.5) node[pos=.5] {$1$};

\draw[draw=black] (12.5,0.5) rectangle ++(.5,.5) node[pos=.5] {$1$};
\draw[draw=black] (13,0.5) rectangle ++(.5,.5) node[pos=.5] {$1$};
\draw[draw=black] (13.5,0.5) rectangle ++(.5,.5) node[pos=.5] {$1$};
\draw[draw=black] (14,0.5) rectangle ++(.5,.5) node[pos=.5] {$1$};
\draw[draw=black] (14.5,0.5) rectangle ++(.5,.5) node[pos=.5] {$1$};
\draw[draw=black] (15,0.5) rectangle ++(.5,.5) node[pos=.5] {$1$};
\draw[draw=black] (15.5,0.5) rectangle ++(.5,.5) node[pos=.5] {$1$};
\draw[draw=black] (16,0.5) rectangle ++(.5,.5) node[pos=.5] {$1$};
\draw[draw=black] (16.5,0.5) rectangle ++(.5,.5) node[pos=.5] {$1$};

\draw[draw=black] (17.2,0.5) rectangle ++(.5,.5) node[pos=.5] {$1$};
\draw[draw=black] (17.7,0.5) rectangle ++(.5,.5) node[pos=.5] {$1$};
\draw[draw=black] (18.2,0.5) rectangle ++(.5,.5) node[pos=.5] {$1$};
\draw[draw=black] (18.7,0.5) rectangle ++(.5,.5) node[pos=.5] {$1$};
\draw[draw=black] (19.2,0.5) rectangle ++(.5,.5) node[pos=.5] {$1$};
\draw[draw=black] (19.7,0.5) rectangle ++(.5,.5) node[pos=.5] {$1$};
\draw[draw=black] (20.2,0.5) rectangle ++(.5,.5) node[pos=.5] {$1$};
\draw[draw=black] (20.7,0.5) rectangle ++(.5,.5) node[pos=.5] {$1$};
\draw[draw=black] (21.2,0.5) rectangle ++(.5,.5) node[pos=.5] {$1$};
\draw[draw=black] (21.7,0.5) rectangle ++(.5,.5) node[pos=.5] {$1$};
\draw[draw=black] (22.2,0.5) rectangle ++(.5,.5) node[pos=.5] {$1$};
\end{tikzpicture}
\caption{Reducing an instance of 3-SAT with $N$ variables  $\sset{x_i}_{i\in[N]}$ and $m$ clauses $\sset{C_j}_{j\in[m]}$ to an instance of $R$-Subset-$\phi_{\ell}$ for $\ell\geq 3$ with $n=\ell+2N+(\ell-1)m$ elements in $R$, where $R=\FF^{1+N+m}$. Here, $0$ (resp., $1$) values inside the vectors refer to the $0$ (resp., $1$) element of $\FF$.}
\label{fig:subset-f-reduction}
\end{figure}

%% file: Preliminaries_cont.tex
\section{Preliminaries (Cont'd)}\label{sec:Preliminaries_cont}
In this section, we provide additional definitions: for SNARKs,  proof-carrying data and multi-signatures.

\newcommand{\snarkgen}{\mathsf{SNARK.Gen}}
\newcommand{\snarkprover}{\mathsf{SNARK.Prover}}
\newcommand{\snarkverify}{\mathsf{SNARK.Verify}}
\newcommand{\sigmasnark}{\sigma}
\newcommand{\tausnark}{\tau}
\newcommand{\relU}{\mathcal{R}_\mathcal{U}}
\newcommand{\langU}{\mathcal{L}_\mathcal{U}}
\newcommand{\extractor}{\mathcal{E}}

\subsection{SNARKs}\label{sec:SNARKs}
We follow the notation from \cite{BCCT13}.
The universal relation $\relU$ \cite{BG08} is the set of instance-witness pairs $(y,w) =((M,x,t), w)$, where $|y|, |w| \leq t$ and $M$ is a random-access machine, such that $M$ accepts $(x, w)$ after running at most $t$ steps. The universal language $\langU$ is the language corresponding to $\relU$.

A \emph{non-interactive argument system} for $\relU$ is a triple of algorithms $(\snarkgen$, $\snarkprover$, $\snarkverify)$ with the following syntax:
\begin{itemize}
	\item $\snarkgen(1^\secParam,B)\to(\sigmasnark,\tausnark)$: on input the security parameter $\secParam$ and a time bound $B\in\NN$, the generation algorithm outputs a common reference string $(\sigmasnark,\tausnark)$ consisting of a prover reference string $\sigmasnark$ and a verification state $\tausnark$.
	\item $\snarkprover(\sigmasnark,y,w)\to\pi$: given a prover reference string $\sigmasnark$, an instance $y=(M,x,t)$ with $t\leq B$, and a witness $w$ such that $(y,w)\in R$, the algorithm produces a proof $\pi$.
	\item $\snarkverify(\tausnark,y,\pi)\to b$: given a verification state $\tausnark$, an instance $y$, and a proof $\pi$, the verifier algorithm outputs a bit $b$.
\end{itemize}

\begin{definition}[SNARK]\label{def:SNARKs}
A non-interactive argument system $(\snarkgen$, $\snarkprover$, $\snarkverify)$ for $\relU$ is a \emph{SNARK}
if the following conditions are satisfied:
\begin{enumerate}
\item \emph{Completeness:}
For every large enough security parameter $\secParam\in\NN$, every time bound $B \in\NN$, and every instance-witness pair $(y,w)=((M,x,t),w)\in \relU$ with $t\leq B$,
\[
\pr{
\begin{tabular}{l|l}
$\snarkverify(\tausnark,y,\pi)=1$ &
\Centerstack[l]{$(\sigmasnark,\tausnark)\gets\snarkgen(1^\secParam,B)$\\
$\pi\gets\snarkprover(\sigmasnark,y,w)$}\\
\end{tabular}
}
=1.
\]

\item \emph{Proof of Knowledge:}
For every polynomial-size prover $P^*$, there exists a polynomial-size extractor $\extractor_{P^*}$ such that for every security parameter $\secParam\in\NN$, every auxiliary input $\aux \in\zo^{\poly(\secParam)}$, and every time bound $B \in \NN$,
\[
\pr{
\begin{tabular}{l|l}
\Centerstack[l]{$\snarkverify(\tausnark, y, \pi) = 1$ \\ $(y,w)\notin\relU$}
&
\Centerstack[l]{$(\sigmasnark,\tausnark) \gets\snarkgen(1^\secParam, B)$\\
$(y,\pi) \gets P^*(\aux,\sigmasnark)$\\
$w\gets\extractor_{P^*}(\aux,\sigmasnark)$} \\
\end{tabular}
}
\leq \negl(\secParam).
\]

\item \emph{Efficiency:}
There exists a universal polynomial $p(\cdot)$ such that, for every large enough security parameter $\secParam\in\mathbb{N}$, every time bound $B\in\mathbb{N}$, and every instance $y=(M,x,t)$ with $t\leq B$,
\begin{itemize}[leftmargin=*]
	\item The generator algorithm $\snarkgen(1^\secParam,B)$ runs in time $p(k+B)$ for a fully succinct SNARK (and in time $p(k+\log B)$ for a preprocessing SNARK).
	\item The prover algorithm $\snarkprover(\sigmasnark,y,w)$ runs in time $p(\secParam+|M|+t+\log B)$ for a fully succinct SNARK (and in time $p(k+|M|+|x|+B)$ for a preprocessing SNARK).
	\item The verifier algorithm $\snarkverify(\tausnark,y,\pi)$ runs in time $p(\secParam+|M|+|x|+\log B)$.
	\item An honestly generated proof has size $p(\secParam + \log B)$.
\end{itemize}
\end{enumerate}
\end{definition}

\subsection{Proof-Carrying Data}\label{sec:PCD}
A \emph{proof-carrying data system} (PCD system) is a cryptographic primitive introduced by \citet{CT10}. Informally speaking, given a predicate $\compliance$, consider a distributed system where nodes perform computations; each computation takes as input messages and generates a new output message. The security goal is to ensure that each output message is compliant with the predicate $\compliance$. Proof-carrying data ensures this goal by attaching short and easy-to-verify proofs of $\compliance$-compliance to each message.

Concretely, a generator $\PCDGen$ first sets up a reference string and a verification state. Anyone can then use the
prover algorithm $\PCDProve$, which is given as input the reference string, prior messages $z_\inputvar$ with proofs $\pi_\inputvar$, and an output message
$z_\outputvar$, to generate a proof $\pi_\outputvar$ attesting that $z_\outputvar$ is $\compliance$-compliant. Anyone can use the verification algorithm $\PCDVerify$, which is given as input
the verification state, a message $z$, and a proof $\pi$, to verify that $z$ is $\compliance$-compliant.

Crucially, the running time of proof generation and proof verification are ``history independent'': the first only depends on the time to execute $\compliance$ on input a node's messages, while the second only on the message length.

We now formally define the notions associated with a PCD system as defined in \cite{BCCT13}. We refer the reader to \cite{BCCT13,BCTV14} for a detailed discussion.

\begin{definition}
A \textsf{(distributed computation) transcript} is a triplet $\trans = (\graph, \linp, \data)$, where $\graph = (\vertex,\edge)$ is a directed acyclic graph, $\linp: \vertex\to\zo^\ast$ are local inputs (node labels), and $\data: \edge \to \zo^\ast$ are edge labels (messages sent on the edge). The output of $\trans$, denoted $\outputvar(\trans)$, is equal to $\data(\tilde{u},\tilde{v})$ where $(\tilde{u},\tilde{v})$ is the lexicographically first edge such that $\tilde{v}$ is a sink.
\end{definition}

Syntactically a proof-carrying transcript is a transcript where messages are augmented by proof strings, \ie a function $\prove : \edge\to\zo^\ast$ provides for each edge $(u, v)$ an additional label prove $(u, v)$, to be interpreted as a proof
string for the message $\data(u, v)$
\begin{definition}
A \textsf{proof-carrying (distributed computation) transcript} PCT is a pair $(\trans, \prove)$ where $\trans$ is a transcript and $\prove : \edge\to\zo^\ast$ is an edge label.
\end{definition}

Next, we define what it means for a distributed computation to be compliant, which as defined in \cite{BCCT13} is the notion of ``correctness with respect to a given local property.'' Compliance is captured via an efficiently computable compliance predicate $\compliance$, which must be locally satisfied at each vertex; here, ``locally'' means with respect to a node's local input, incoming data, and outgoing data. For convenience, for any vertex $v$, we let $\child(v)$ and $\parent(v)$ be the vector of $v$'s children and parents respectively, listed in lexicographic order.

\begin{definition}
Given a polynomial-time predicate $\compliance$, we say that a distributed computation transcript $\trans = (\graph, \linp, \data)$ is $\compliance$-\textsf{compliant} (denoted by $\compliance(\trans) = 1$) if for every $v \in \vertex$ and $w \in \child(v)$ it holds that
\[
\compliance(\data(v, w); \linp(v), \inputs(v)) = 1,
\]
where $\inputs(v) \coloneqq \data(u_1, v),\ldots, \data(u_c, v)$ and $(u_1, \ldots , u_c)\coloneqq\parent(v)$. Furthermore, we say that a message $z$ from node $v$ to $w$ is $\compliance$-compliant if $\compliance(\data(v, w); \linp(v), \inputs(v)) = 1$ and there is a transcript $\trans$ such that $v$ is the sink and $\compliance(\trans) = 1$.
\end{definition}

\begin{definition}
Given a distributed computation transcript $\trans = (\graph, \linp, \data)$ and any edge $(v, w) \in \edge$, we denote by $t_{\trans,\compliance}(v, w)$ the time required to evaluate $\compliance(\data(v, w); \linp(v), \inputs(v))$. We say that $\trans$ is $B$-\textsf{bounded} if $t_{\trans,\compliance}(v, w) \leq B$ for every edge $(v, w)$.
\end{definition}

\begin{definition}
The \textsf{depth of a transcript} $\trans$, denoted $d(\trans)$, is the largest number of nodes on a source-to-sink path in $\trans$ minus 2 (to exclude the source and the sink). The depth of a compliance predicate $\compliance$, denoted $d(\compliance)$, is defined to be the maximum depth of any transcript $\trans$ compliant with $\compliance$. If $d(\compliance)\coloneqq\infty$ (\ie paths in $\compliance$-compliant transcripts can be arbitrarily long) we say that $\compliance$ has unbounded depth.
\end{definition}

We note that for our application in \cref{sec:app:BA}, we can assume that for every $v\in\vertex$, the label input is $\linp(v)=\bot$.

We now give a formal definition of a PCD system.
\begin{definition}\label{def:PCD}
A \textsf{proof-carrying data (PCD) system} for a class of compliance predicates $C$ is a triple of algorithms $(\PCDGen,\PCDProve,\PCDVerify)$ that work as follows:
\begin{itemize}[leftmargin=*]
	\item
    $\PCDGen(1^\secParam,\compliance)\to(\PCDsigma,\PCDtau)$: on input the security parameter $\secParam$ and compliance predicate $\compliance\in C$, the (probabilistic) generator $\PCDGen$ outputs a reference string $\PCDsigma$ and a corresponding verification state $\PCDtau$.
	\item
    $\PCDProve(\PCDtau,z_\inputvar,\pi_\inputvar,\linp,z_\outputvar)\to\pi_\outputvar$: given a reference string $\PCDtau$, inputs $z_\inputvar$ with corresponding proofs $\pi_\inputvar$, a local input $\linp$, and an output $z_\outputvar$, the (honest) prover algorithm $\PCDProve$ produces a proof $\pi_\outputvar$ attesting to consistency of $z_\outputvar$ with a $\compliance$-compliant transcript.
	\item
    $\PCDVerify(\PCDtau, z_\outputvar, \pi_\outputvar)\to b$: given the verification state $\PCDtau$, an output $z_\outputvar$, and a proof string $\pi_\outputvar$, the verifier algorithm $\PCDVerify$ accepts if it is convinced that $z_\outputvar$ is consistent with some $\compliance$-compliant transcript.
\end{itemize}
\end{definition}

After the generator $\PCDGen$ is run to obtain $\PCDsigma$ and $\PCDtau$, the prover $\PCDProve$ is used (along with $\PCDsigma$) at each node of a distributed computation transcript to dynamically compile it into a proof-carrying transcript by generating and adding a proof to each edge. Each of these proofs can be checked using the verifier $\PCDVerify$ (along with $\PCDtau$). A PCD system $(\PCDGen,\PCDProve,\PCDVerify)$  must satisfy the following properties:

\paragraph{Completeness:}
An honest prover can convince a verifier that the output of any compliant transcript is indeed compliant. Namely, for every security parameter $\secParam$, compliance predicate $\compliance$, and distributed-computation generator G (described below),
\[
\pr{
\begin{tabular}{l|l}
\Centerstack[l]{$\trans$ is $B$-bounded \\ $\compliance(\trans)=1$ \\ $\PCDVerify(\PCDtau, z, \pi) \neq 1$}
&
\Centerstack[l]{$(\PCDsigma,\PCDtau) \gets\PCDGen(1^\secParam, \compliance)$\\
$(z,\pi,\trans) \gets \ProofGen(\compliance, \PCDsigma, G, \PCDProve)$}\\
\end{tabular}
}
\leq \negl(\secParam).
\]

Above, $\ProofGen$ is an interactive protocol between a distributed-computation generator $\dcgen$ and the PCD prover $\PCDProve$, in which both are given the compliance predicate $\compliance$ and the reference string $\PCDsigma$. Essentially, at every time step, $\dcgen$ chooses to do one of the following actions: (1) add a new unlabeled vertex to the computation transcript so far (this corresponds to adding a new computing node to the computation), (2) label an unlabeled vertex (this corresponds to a choice of local data by a computing node), or (3) add a new labeled edge (this corresponds to a new message from one node to another). In case $\dcgen$ chooses the third action, the PCD prover $\PCDProve$ produces a proof for the $\compliance$-compliance of the new message, and adds this new proof as an additional label to the new edge. When $\dcgen$ halts, the interactive protocol outputs the distributed computation transcript $\trans$, as well as $\trans$'s output and corresponding proof. Intuitively, the completeness property requires that if $\trans$ is compliant with $\compliance$, then the proof attached to the output (which is the result of dynamically invoking $\PCDProve$ for each message in $\trans$, as $\trans$ was being constructed by $\dcgen$) is accepted by the verifier.
	
\paragraph{Proof of knowledge (and soundness):}
Loosely speaking, if the verifier accepts a proof for a message, the prover ``knows'' a compliant transcript $\trans$ with output $z$.
For every polynomial-size prover $P^*$ there exists a polynomial-size extractor $\extractor_{P^*}$ such that for every polynomial-size compliance predicate $\compliance\in C$ and every auxiliary input $\aux \in\zo^{\poly(\secParam)}$,
\[
\pr{
\begin{tabular}{l|l}
\Centerstack[l]{$\PCDVerify(\PCDtau, z, \pi) = 1$ \\ $\outputvar(\trans)\neq z\vee \compliance(\trans)\neq 1$}
&
\Centerstack[l]{$(\PCDsigma,\PCDtau) \gets\PCDGen(1^\secParam, \compliance)$\\
  $(z,\pi) \gets P^*(\PCDsigma,\aux  )$\\
  $\trans\gets\extractor_{P^*}(\PCDsigma,\aux)$} \\
\end{tabular}
}
\leq \negl(\secParam).
\]

\paragraph{Succinctness:}
There exists a universal polynomial $p(\cdot)$ such that for every compliance predicate $\compliance\in C$, every time bound $B\in\mathbb{N}$, and every $B$-bounded distributed computation transcript $\trans$,
\begin{itemize}[leftmargin=*]
	\item
	The computation time of $\PCDProve(\PCDsigma,z_\inputvar, \pi_\inputvar, \linp,z_\outputvar)$ is $p(\secParam + |\compliance| + B)$.
	\item The verification algorithm $\PCDVerify(\PCDtau, z, \pi)$  runs in time $p(\secParam + |\compliance| + |z| + \log B)$
	\item An honestly generated proof has size $p(\secParam + \log B)$.
\end{itemize}

\begin{theorem}[\cite{BCCT13}] \label{PCD_from_snarks}
Let the size of a compliance predicate $\compliance$, denoted by $s(\compliance)$, be the largest number of nodes in any transcript compliant with $\compliance$. Assuming the existence of SNARKs with linear extraction (\ie $\ssize{\mathbb{E}_{\mathcal{P}^*}}\leq c|\mathcal{P^*}|$ for some constant $c$), there exist PCD systems for logarithmic-depth and polynomial-size compliance predicates.	
\end{theorem}

\subsection{Merkle Hash Proof System}\label{sec:merkle-hash}

A Merkle hash proof system~\cite{Merkle89} corresponding to a hash function $\hash:\zo^\secParam \times \zo^\lambda\to\zo^{\lambda/2}$ is defined by a tuple of algorithms $(\msetup,\mhash,\mproof,\mverify)$ as follows:
\begin{itemize}
	\item $\msetup(1^\secParam)$: On input the security parameter, the setup algorithm samples and outputs a random $\seed\gets\zo^\secParam$ for the hash function.
	\item $\mhash(\seed,x_1,\ldots,x_n)$: On input the seed and a vector $x_1,\ldots,x_n$, the Merkle hash algorithm computes a hash using a Merkle tree as follows:
	\begin{itemize}
		\item For each $i\in[n]$, compute $y_i^0=\hash(\seed,x_i)$.
		\item For each $\ell\in[\log(n)]$ and $i\in[n/2^\ell]$,\footnote{For simplicity, we assume that $n$ is a power of 2. The general case follows by including additional elements $0^\lambda$, such that the length of the resulting input string becomes a power of 2.} compute $y^\ell_i=\hash(\seed,y^{\ell-1}_{2i-1}||y^{\ell-1}_{2i})$.
	\end{itemize}
Output $y=y_1^{\log(n)}$.

\item $\mproof(\seed,x_1,\ldots,x_n,x_i)$: On input the seed, a vector $x_1,\ldots,x_n$, and an element $x_i$, the Merkle proof algorithm computes and outputs a proof $p$ as follows:
\begin{itemize}
	\item For each $k\in[n]$, compute $y_k^0=\hash(\seed,x_k)$.
	\item For each $\ell\in[\log(n)]$ and $k\in[n/2^\ell]$ compute $y^\ell_k=\hash(\seed,y^{\ell-1}_{2k-1}||y^{\ell-1}_{2k})$.
	\item Initialize the proof $p=\sset{(i,\sibling(y^ 0_{i}))}$ and for each level $\ell\in[\log(n)]$, set $p=p\cup\sset{(\ceil{i/2^\ell},\sibling(y^\ell_{\ceil{i/2^\ell}}))}$.
\end{itemize}

\item $\mverify(\seed,x_i,y,p)$: On input the seed, an input element $x_i$, Merkle hash $y$, and a Merkle proof $p$, the Merkle verification algorithm parses $p=((i_0,x^0),\ldots,(i_{\log(n)},x^{\log(n)}))$ and proceed as follows:
    \begin{itemize}
    \item If $i_0$ is an even number, compute $y^1=\hash(\seed,\hash(\seed,x_i)||x^0)$, else compute $y^1=\hash(\seed,x^0||\hash(\seed,x_i))$.
    \item For each $\ell\in[\log(n)]$, if $i_\ell$ is an even number, compute $y^\ell=\hash(\seed,\hash(\seed,y^{\ell-1})||x^\ell)$, else compute $y^\ell=\hash(\seed,x^\ell||\hash(\seed,y^{\ell-1}))$.
    \end{itemize}
If $y^{\log(n)}=y$, output 1; else, output 0.
\end{itemize}

The Merkle Hash Proof System has the following properties.
\begin{theorem}[Merkle hash proof system]\label{thm:merkle-hash}
Assuming existence of a length-halving, seeded, collision resistant hash function $\hash:\zo^\secParam \times \zo^\lambda\to\zo^{\lambda/2}$, the Merkle hash proof system $(\msetup,\mhash,\mproof,\mverify)$ satisfies the following properties:
\begin{itemize}
\item \textbf{Completeness:} For any input string $x_1,\ldots,x_n\in\zo^{n\lambda}$ and $i\in[n]$, it holds that:
\[
\pr{
	\begin{tabular}{l|l}
	~ & $\seed\gets\msetup(1^\secParam)$\\
	$\mverify(\seed,x_i,y,p)=1$ & $y=\mhash(\seed,x_1,\ldots,x_n)$\\
	~ & $p=\mproof(\seed,x_1,\ldots,x_n,x_i)$ \\
	\end{tabular}
}
=1.
\]
\item \textbf{Soundness:} No PPT adversary $\Adv$, can win the following game with more than negligible probability (in $\secParam$):
\begin{enumerate}
	\item The challenger samples $\seed\gets\msetup(1^\secParam)$ and sends to $\Adv$.
	\item $\Adv$ responds with $(i, \sset{x_{j}}_{j\in[n]\setminus\sset{i}})$.
	\item The challenger samples $x_i\gets\zo^\lambda$, computes $\mhash(\seed,x_1,\ldots,x_n)=y$ and sends $(x_i,y)$ to $\Adv$.
	\item $\Adv$ responds with a pair $(x',p)$, and wins if $\mverify(\seed,x',y,p)=1 $ and $x'\neq x_i$ for every $i\in[n]$.
	\end{enumerate}
\end{itemize}
\end{theorem}

\subsection{Multi-signatures}\label{sec:ms}
In a multi-signature scheme, a single short object---the \emph{multi-signature}---can take the place of $n$ signatures by $n$ signers, all on the same message.\footnote{Note that multi-signatures are a special case of \emph{aggregate} signatures \cite{BGLS03}, which in contrast allow combining signatures from $n$ different parties on $n$ \emph{different} messages.}
The first formal treatment of multi-signatures was given by \citet*{MOR01}. We consider a variant of this model due to \citet{Boldyreva03} that is also used by \citet{LOSSW13}. In this model, the adversary is given a single challenge verification key $\vk$, and a signing oracle for that key. His goal is to output a forged multi-signature $\sigma^*$ on a message $m^*$ under keys $\vk_1,\ldots,\vk_\ell$, where at least one of these keys is a challenge verification key (without loss of generality, $\vk_1$). For the forgery to be nontrivial, the adversary must not have queried the signing oracle at $m^*$.
The adversary is allowed to choose the remaining keys, but must prove knowledge of the private keys corresponding to them.

\begin{definition} \label{def:multisig}
A \emph{multi-signature} scheme is a tuple of algorithms
\begin{itemize}[leftmargin=*]
	\item $\MSSetup(1^\secParam)\to\pp$: On input the security parameter, the setup algorithm outputs public parameters $\pp$.
	\item $\MSGen(\pp)\to(\vk,\sk)$: On input the public parameters $\pp$, the key-generation algorithm outputs a pair of verification/signing keys $(\vk,\sk)$.
	\item $\MSSign(\pp,\sk,m)\to \sigma$: On input $\pp$, a signing key $\sk$, and a message $m$, the signing algorithm outputs a signature $\sigma$.
	\item $\MSVerify(\pp,\vk,\sigma,m)\to b$: On input $\pp$, a verification key $\vk$, a signature $\sigma$, and a message $m$, the verification algorithm outputs a bit $b\in\zo$.
	\item $\MSCombine(\pp,\sset{\vk_i,\sigma_i}_{i=1}^\ell, m)\to\sigma$: On input $\pp$, a collection of signatures (or multi-signatures), and a message $m$, the combine algorithm outputs a combined multi-signature $\sigma$, with respect to the union of verification keys.
	\item $\MSMverify(\pp,\sset{\vk_1,\ldots,\vk_n},S, m, \sigma)\to b$: On input $\pp$, the set of all verification keys, a subset $\S\subseteq[n]$, a message $m$, and a multi-signature $\sigma$, the multi-signature verification algorithm outputs a bit $b\in\zo$.
\end{itemize}
\end{definition}

\noindent
We require the following properties from a multi-signature scheme.

\paragraph{Correctness:}
The correctness requirement of digital signatures must hold for $(\MSSetup$, $\MSGen, \MSSign, \MSVerify)$.
In addition, for any message $m$, any collection of honestly generated signatures $\sset{\sigma_i \gets \MSSign(\pp,\sk_i,m)}_{i \in S}$ on $m$ (for some $S \subseteq [n]$), the combined multi-signature formed by $\bar\sigma \gets \MSCombine(\pp,\{\vk_i, \sigma_i\}_{i \in S}, m)$
will properly verify with overwhelming probability, \ie
$\pr{ 1 \gets \MSMverify(\pp,\sset{\vk_1,\ldots,\vk_n},S, m, \bar\sigma)} \ge 1 - \negl(k)$.

\paragraph{Unforgeability:}
For any PPT adversary $\Adv$, the probability that the challenger outputs 1 when interacting with $\Adv$ in the following game is negligible in the security parameter $\secParam$:

\begin{enumerate}
    \item \emph{Setup.}
    $\Adv$ selects a proper subset $\IS \subseteq [n]$ (corresponding to corrupted parties). The challenger samples a pair of verification/signing keys $(\vk_i, \sk_i) \gets \MSGen(\pp)$ for every $i \in [n]\setminus \IS$, and gives $\Adv$ all verification keys $\sset{\vk_i}_{i \in [n]\setminus \IS}$.
    Next, $\Adv$ chooses keys $\sset{\sk_i,\vk_i}_{i \in \IS}$ for the corrupted parties and sends them to the challenger.
	
    \item \emph{Signing queries.}
    $\Adv$ can make polynomially many adaptive signature queries of the form $(m, \vk_i)$.
    For each query, the challenger responds with a signature $\sigma \gets \MSSign(\pp,\sk_i,m)$ on the message $m$ with respect to the signing key $\sk_i$ corresponding to $\vk_i$.
	
    \item \emph{Output.}
    $\Adv$ outputs a triple $(\bar\sigma^*, m^*, \sset{\vk_i}_{i \in S})$.
    The challenger outputs 1 if at least one of the provided verification keys $\vk_i$ corresponds to a challenge (honest party) key, the message $m^*$ was not queried to the signature oracle with this verification key $\vk_i$, and the provided forgery $\sigma^*$ is a valid multi-signature, \ie $1 \gets \MSMverify(\pp,\sset{\vk_1,\ldots,\vk_n},S, m^*, \sigma^*)$.
\end{enumerate}

\input{lossw}

%% file: lossw.tex
\subsection{The Multi-Signatures Scheme of \texorpdfstring{\citet{LOSSW13}}{Lg}}\label{lossw}
In this section we describe the LOSSW multi-signature scheme that is used in \cref{sec:succinct_arguments}.
We will let $\mathbb{G}$ and $\mathbb{G}_T$ are multiplicative groups of prime order $p$, and denote $g$ a generator of $\mathbb{G}$. In addition, let $e: \mathbb{G}\times\mathbb{G}\rightarrow\mathbb{G}_T$ be an efficiently computable non-degenerate bilinear map.
The multi-signature scheme of \citet{LOSSW13} is based on the Bilinear Computational Diffie-Hellman (BCDH) assumption. The message space is $\zo^k$ for some fixed $k$. The following is taken verbatim from \cite{LOSSW13}:
\begin{itemize}[leftmargin=*]
	\item $\MSSetup(1^\secParam)$:
    Sample random elements $u',u_1,\ldots,u_k\in \mathbb{G}$ and output the public parameters $\ppms$, consisting of descriptions of $\GG, \GG_T, p, e$, $u',u_1,\ldots,u_k$ and the generator $g$ of $\GG$.
	\item $\MSGen(\ppms)$:
    Sample a random signing key $\sk\in\ZZ_p$ and set the corresponding verification key $\vk$ as $e(g,g)^{\sk}$.
	\item $\MSSign(\ppms,\sk,m)$:
    Parse the message $m$ as $(m_1,\ldots,m_k)\in\zo^k$, sample $r\gets\ZZ_p$, and compute $\sigma=(\sig_1,\sig_2)$ as follows:
	\[
    \sig_1=g^\sk\cdot \bigg(u'\cdot\prod_{i=1}^k u_i^{m_i}\bigg)^r~\text{and}~\sig_2=g^r.
    \]
	\item $\MSVerify(\ppms, \vk, m, \sigmams):$
    Parse the message $m$ as $(m_1,\ldots,m_k)\in\sset{0,1}^k$ and $\sigmams=(\sig_1,\sig_2)$, and outputs 1 if and only if
	\[
	e(\sig_1,g)\cdot e\bigg(\sig_2,u'\cdot\prod_{i=1}^k u_i^{m_i}\bigg)^{-1}= \vk.
	\]
    \item $\MSCombine(\ppms, \sset{\vk_i,\sigma_i}_{i\in S},m)$:
    Parse each $\sigma_i$ as $(\sig_1^{(i)},\sig_2^{(i)})$ and compute the combined multi-signature $\sigmams=(\sig_1,\sig_2)$ as follows:
	\[
	\sig_1=\prod_{i\in S} \sig_1^{(i)}\text{ and }\sig_2=\prod_{i\in S}\sig_2^{(i)}.
	\]
    \item $\MSMverify(\ppms,\sset{\vk_1,\ldots,\vk_n},S,m,\sigmams)$:
    Output 1 if and only if
	\[
	e\left(\sig_1,g\right)\cdot e\bigg(\sig_2,u'\cdot\prod_{i=1}^k u_i^{m_i}\bigg)^{-1}=\prod_{i\in S} \vk_{i}.
	\]
\end{itemize}

%% file: balanced_ba_cont.tex
\section{Balanced Communication-Efficient BA (Cont'd)}\label{sec:app:BA_cont}
In this section, we provide supplementary material for \cref{sec:app:BA}.
\subsection{Balanced Byzantine Agreement from SRDS (Cont'd)}\label{sec:ba_from_srds_cont}

In this section, we give the proof of \cref{lem:ba} and discuss applications of our Byzantine agreement protocol.

\input{app_ba_proof}

%% file: app_ba_proof.tex
\begin{proof}
Let $\Adv$ be a PPT adversary for $\protba$. We construct a simulator $\Sim$ as follows.
The simulator $\Sim$ starts by simulating the setup for the protocol, while allowing adaptive corruptions by \Adv (in a similar way to the robustness and unforgeability games). First, $\Sim$ runs the setup algorithm as $\pp\gets\TSCR(1^\secParam,1^{n\cdot z})$, and for every $i\in[n]$ and $j\in[z]$ computes $(\vk_{i,j},\sk_{i,j})\gets\TSGen(\pp)$.
Next, \Sim sends $(1^\secParam, 1^{n\cdot z},\pp,\sset{\vk_{i,j}}_{i\in[n],j\in[z]})$ to $\Adv$.
As long as $\ssize{\IS}\leq \beta\cdot n$ and $\Adv$ requests to corrupt a party $\Party_i$, the simulator sends $\sset{\sk_{i,j}}_{j\in[z]}$ to $\Adv$ and receives back $\sset{\vk_{i,j}'}_{j\in[z]}$; in the bare-PKI mode, $\Sim$ updates each $\vk_{i,j}=\vk_{i,j}'$.
Let $\sset{\vk_{i,j}}_{i\in[n],j\in[z]}$ be the PKI keys at the end of this process.

The simulator \Sim proceeds to simulate the protocol execution towards \Adv. Initially, \Sim receives from $\fba$ the input bits of all honest parties $\sset{x_i}_{i\notin\IS}$.
To simulate $\faecomm$ in Step \ref{step:faecomm_first}, the simulator receives from \Adv the communication-tree $\commtree$ defining the set of isolated parties $\DS$.
The simulator simulates sending the output to every corrupted party. Let $\CS$ denote the supreme committee (the parties assigned to the root).

To simulate $\fba$ for the supreme committee in Step~\ref{step:fba}, \Sim sends to $\Adv$ the input bit $x_i$ for every $i\in\CS\setminus\IS$ and receives inputs $\sset{x_i}_{i\in\IS\cap\CS}$. If $2/3$ of the honest committee members' bits are the same, denote this value by $y$; otherwise, let \Adv determine $y$. Output the value $y$ to every corrupted party in $\CS$. To simulate $\fct$ in Step~\ref{step:fct}, sample a random $s\in\zo^\secParam$ and send $s$ to $\Adv$ for every $\Party_i$ for $i\in\IS\cap\CS$.

To simulate the call to $\faecomm$ in Step~\ref{step:faecomm_second}, receive inputs from $\Adv$ on behalf of corrupted supreme-committee members, and send $(y,s)$ to $\Adv$ for every $i\in\IS$. In addition, receive $(y_i,s_i)$ for every $i\in\DS$ from $\Adv$.

Next, for every honest party $\Party_i$ for $i\notin\IS$ do the following:
\begin{itemize}
    \item
    For $i\notin\DS$, compute $\sigma_{i,j}\gets\TSSignShare(\pp,\mapping(i,j),\sk_{i,j},(y,s))$ for each $j\in[z]$.
    \item
    For $i\in\DS$, compute $\sigma_{i,j}\gets\TSSignShare(\pp,\mapping(i,j),\sk_{i,j},(y_i,s_i))$ for each $j\in[z]$.
\end{itemize}

To simulate Step~\ref{step:sign}, for every $i\in[n]\setminus\IS$, let $L_i=\sset{\leafnode_{i_1},\ldots,\leafnode_{i_z}}\subseteq\vertexCommtree$ be the subset of leaves assigned to $\Party_i$. For each $j\in[z]$, send $\sigma_{i,j}$ to all corrupted parties assigned to the leaf node~$\leafnode_{i_j}$ on behalf of $\Party_i$. In addition, for every $\Party_i$ assigned to a leaf node $v$, receive a signature $\sigma_{j,k}$ from every corrupt $\Party_j$ for which $v=\leafnode_{j_k}\in L_j$.

To simulate Step~\ref{step:recurse}, for each level $\ell=1,\ldots,\ell^*$ of the tree and each node $v$ on level $\ell$:
\begin{enumerate}
	\item For each $i\in \party(v)\setminus \IS$, prepare the set of signatures received in Step~\ref{step:recurse:firstlevel} as follows:
	\begin{itemize}[leftmargin=*]
	\item
    For $\ell=1$: let $S^{\ell,i,1}_\sig$ be the set of following signatures. For every honest $\Party_j$ with $v=\leafnode_{j_k}\in L_j$, the signature $\sigma_{j,k}$ simulated in the previous step. For every corrupt $\Party_j$ with $v=\leafnode_{j_k}\in L_j$, the signature $\sigma^i_{j,k}$ received from the adversary (note that the adversary might send different signatures to different parties).
 	\item
    For $\ell>1$: let $S^{\ell,i,1}_\sig$ be the set of following signatures. For each child node $u\in\child(v)$ and each $j\in \party(u)\setminus \IS$, the signature $\sigma_{u}$ (that was simulated for level $\ell-1$). For each $j\in \party(u)\cap \IS$, the signature $\sigma^i_u$ received from the adversary $\Adv$ (note that the adversary might send different signatures to different parties).
	\end{itemize}
    \item Next, simulate $\ssize{\party(v)}$ broadcast protocols in Step~\ref{step:recurse:ba}, where for every $i\in \party(v)$, party $\Party_i$ broadcasts $S^{\ell,i,1}_\sig$. Let $S^{\ell,i,2}_\sig$ be the union of the sets of the broadcasted signatures.
    \item
    To simulate Step~\ref{step:recurse:secondlevel}, for each party assigned to the node, \ie for each $i\in\party(v)\setminus \IS$, compute $S^{\ell,i,3}_\sig\gets \TSAggr_1(\pp,\sset{\vk_{1,1},\ldots,\vk_{n,z}},(y,s),S^{\ell,i,2}_\sig)$.
   	If $\ell=1$, for each  $\signsig$ in $S^{i,\ell,3}_\sig$ check if $\indexMin(\signsig)=\indexMax(\signsig)$ and if $\indexMin(\signsig)\in\Range(v)$, whereas if $\ell>1$ check if $\exists v'\in\child(v)$ such that the range  $[\indexMin(\signsig),\indexMax(\signsig)]$ falls within the range $\Range(v')$. If this check fails for any $\signsig$, it updates $S^{i,\ell,3}_\sig=S^{i,\ell,3}_\sig\setminus\sset{\signsig}$.
    To simulate $\faggrsig$, for every $i\in\party(v)\cap \IS$, receive from $\Adv$ a message $((\tilde y_i,\tilde s_i),\tilde S^{\ell,i,3}_\sig)$.
    If $\ssize{\party(v)\setminus\sset{\IS\cup\DS}}\geq2\ssize{\party(v)}/3$ (\ie the node is good), compute
	\[
	\sigma_v\gets\TSAggr_2\Big(\pp, (y,s), S^{\ell,i,3}_\sig\Big).
	\]
    Else (\ie the node is bad), get $\sigma_v$ from $\Adv$. Finally, send $\sigma_v$ to $\Adv$ as the output of $\faggrsig$.
	\item
    If $\ell<\ell^*$, send for every $\sigma_v$ from each honest party in $\party(v)$ to every corrupt party in $\party(u)$, where $u=\parent(v)$. In addition, receive from $\Adv$ a signature $\sigma_v'$ from every corrupt party in $\party(v)$ to every honest party in $\party(u)$.
\end{enumerate}

To simulate the call to $\faecomm$ in Step~\ref{step:faecomm_third}, receive inputs from $\Adv$ on behalf of corrupted supreme-committee members, and send $(y,s,\sigma_\rroot)$ to $\Adv$ for every $i\in\IS$. In addition, receive $(y'_j,s'_j,\sigma_j')$ for $j\in\DS$ from $\Adv$.
Finally, to simulate Step~\ref{step:prf}, for every $i\notin\IS\cup\DS$ evaluate $\CS_i=F_s(i)$ and simulate party $\Party_i$ sending $(y,s,\sigma_\rroot)$ to every party $\Party_j$ for $j\in\IS\cap\CS_i$. For every $i\in\DS$ evaluate $\CS_i=F_{s'_j}(i)$ and simulate party $\Party_i$ sending $(y'_j,s'_j,\sigma'_j)$ to every party $\Party_j$ for $j\in\IS\cap\CS_i$.

To conclude the simulation, the simulator sends the value $y$ to the ideal functionality $\fba$ as the ``tie-breaker'' value and outputs whatever \Adv outputs.

Note that \Sim simulates a random honest execution towards the adversary, with only the syntactic difference that \Sim simulates the ideal functionalities computing $\faecomm$, $\fba$, $\fct$ and $\faggrsig$ (rather than using trusted parties).
Thus, the view of the adversary is perfectly distributed in the real and ideal worlds. What remains to prove is that conditioned on the view of the adversary, the output of the honest parties is correct and identical in the
real and ideal worlds. In other words, we need to show that this Byzantine agreement protocol satisfies both agreement and validity.

\begin{claim}[Agreement] \label{claim:agreement}
For any adversarial strategy of $\Adv$, all honest parties output the same value, except for negligible probability.
\end{claim}
We show that our protocol satisfies agreement in three main steps; (1) We start by showing that with an overwhelming probability, every isolated party receives a message from at least one non-isolated honest party in the last round. (2) Next, we show that the aggregate signature $\sigma_\rroot$ obtained by the end of Step~\ref{step:recurse} is a valid SRDS on $(y,s)$, where $y$ and $s$ are the outputs of $\fba$ and $\fct$ in Step~\ref{step:fct}, respectively. We prove this by showing a reduction to the robustness property of the SRDS scheme. Thereby showing that every honest party receives a valid SRDS on the same message $(y,s)$. (3) Finally we prove that every honest party only receives one valid SRDS (which is on $(y,s)$). We prove this by showing that an adversary cannot forge a valid SRDS on any other message by relying on the unforgeability of the SRDS scheme. Thus, each honest party outputs the same value $y$.
Now we proceed to the formal proof.
\begin{proof}[Proof of \cref{claim:agreement}.]
Let $F_s$ be a truly random function. Then the set $\CS_i$ defined by $F_s(i)$ is chosen randomly for each $i\in[n]$. Therefore, in expectation, each party $P_j$ appears in $\polylog(n)$ sets. From Chernoff bound, except with some negligible probability (in $n$), each party receives messages from $\polylog(n)\pm\delta$ for $\delta=O(1)$ other parties.
Similarly, except with negligible probability, each party receives messages from at least one non-isolated honest party. Therefore, each isolated party $P_i$ for $i\in\DS$ in the initial phase of the protocol, receives a message from at least one non-isolated honest party $P_j\in[n]\setminus\sset{\IS\cup\DS}$. If this is true for a truly random function, the same must also hold for a pseudorandom $F_s$ with overwhelming probability (in $\secParam$) over a random seed $s$.
Recall that the message sent by $P_j$ to $P_i$ is $(y,s,\sigma_\rroot)$ (where $y$ is the output of $\fba$ in Step~\ref{step:fba}).
It remains to show the following:
\begin{enumerate}
	\item Except with some negligible probability, $\sigma_\rroot$ is a valid SRDS on $(y,s)$.
	\item Except with some negligible probability, no adversary can compute a valid SRDS on any message other than $(y,s)$.
\end{enumerate}

\paragraph{1. Receiving valid signatures on $(y,s)$.}
Let us assume for the sake of contradiction that $\sigma_\rroot$ is \emph{not} a valid SRDS on $(y,s)$. We now construct an adversary $\BS$ that can break \emph{robustness} of the SRDS scheme. The adversary $\BS$ interacts with the challenger of the SRDS scheme and the adversary $\Adv$ and proceeds as follows:
\begin{itemize}[leftmargin=*]
	\item $\BS$ maps each corrupt virtual party to a party in the SRDS robustness game, i.e., elements in the set $[n\cdot z]$. In other words, the challenger of the SRDS scheme runs the setup algorithm as $\pp\gets\TSCR(1^\secParam,1^{n\cdot z})$, and for every $i\in[n]$ and $j\in[z]$ computes $(\vk_{i,j},\sk_{i,j})\gets\TSGen(1^\secParam)$.
	Next, it sends $(1^\secParam, 1^{n\cdot z},\pp,\sset{\vk_{i,j}}_{i\in[n],j\in[z]})$ to $\BS$, which it forwards to $\Adv$. For each $i\in\IS$ that $\Adv$ requests to corrupt, $\BS$ chooses to corrupt the corresponding parties $\sset{(i,j)}_{j\in[z]}$ in the SRDS robustness game and receives $\sset{\sk_{i,j}}_{j\in[z]}$ from the challenger, which it forwards to $\Adv$.
Next, $\BS$ receives verification keys $\sset{\vk_{i,j}'}_{i\in\IS, j\in[z]}$ of the corrupted parties from $\Adv$.
 \item
For the bare-PKI mode, $\BS$ updates $\vk_{i,j}=\vk_{i,j}'$ for each $i\in\IS$ and $j\in[z]$.
	\item $\BS$ then simulates step \ref{step:faecomm_first}, as described in the simulator to receive the $(n,\IS)$-\textsf{almost-everywhere-communication tree with repeated parties} $\tree$ from $\Adv$. It  transforms this tree into an $(n\cdot z,\{(i,j)\}_{i\in\IS,j\in[z]})$-\textsf{almost-everywhere-communication tree} by augmenting it with level $0$ comprising of $n\cdot z$ nodes (representing the $n\cdot z$ virtual parties in the SRDS game), and adding an edge between each of these nodes and the leaf node that it (\ie the party that they represent) is assigned to. It forwards this transformed tree to the challenger of the SRDS game.

	\item $\BS$ then proceeds to simulate steps \ref{step:fba}, \ref{step:fct}, and \ref{step:faecomm_second} as described in the simulator and learns $(y,s)$ and $(y_i,s_i)$ for each $i\in\DS$. $\BS$ sets $m=(y,s)$, $m_{(i,j)}=(y_{i},s_i)$ for each $i\in\DS$, $j\in[z]$ and for each $(i,j)\in \goodP\setminus\{(i,j)\}_{i\in\DS,j\in[z]}$, where $\goodP$ is the set of all honest parties in the SRDS game that are assigned to leaf nodes that do not have a good path in the transformed tree (described in the previous step), it sets $m_{(i,j)}=(y,s)$. It sends these messages to the challenger of the SRDS game.
	
	\item $\BS$ receives signatures $\sset{\sigma_{i,j}}_{i\in [n]\setminus \IS,j\in[z]}$ of the honest parties from the challenger and forwards them to the adversary $\Adv$.

	\item
	For each level $\ell=1,\ldots,\ell^*$ of the communication tree and each node $v$ on level $\ell$, it simulates Step~\ref{step:recurse} as described in the simulator, except in Step~\ref{step:recurse:secondlevel}, if the node is good, it sets $\sigma_v$ to the partially aggregated signature sent by the challenger and if the node is bad, it forwards the partially aggregated signature $\sigma_v$ received from $\Adv$ to the challenger of the SRDS game.
\end{itemize}
Note that if for some adversarial strategy $\Adv$, the signature $\sigma_\rroot$ is not a valid SRDS on $(y,s)$, then by construction, $\BS$ wins the robustness game of the SRDS scheme. From robustness of the SRDS scheme, we know that this only happens with at most negligible probability, therefore our assumption is incorrect and with overwhelming probability, $\sigma_\rroot$ is a valid SRDS on $(y,s)$.

\paragraph{2. Not receiving valid signatures on other values.}
We now show that if the adversary $\Adv$ can forge an SRDS on any other message, then we can use this adversary to construct another adversary $\BS$ that can  break unforgeability of the SRDS scheme. The adversary $\BS$ proceeds as follows:
\begin{itemize}[leftmargin=*]
	\item $\BS$ maps each corrupt virtual party to a party in the SRDS game, i.e., elements in the set $[n\cdot z]$. In other words, the challenger of the SRDS scheme runs the setup algorithm as $\pp\gets\TSCR(1^\secParam,1^{n\cdot z})$, and for every $i\in[n]$ and $j\in[z]$ computes $(\vk_{i,j},\sk_{i,j})\gets\TSGen(1^\secParam)$.
	Next, it sends $(1^\secParam, 1^{n\cdot z},\pp,\sset{\vk_{i,j}}_{i\in[n],j\in[z]})$ to $\BS$, which it forwards to $\Adv$. For each $i\in\IS$ that $\Adv$ requests to corrupt, $\BS$ chooses to corrupt the corresponding parties $\sset{(i,j)}_{j\in[z]}$ in the SRDS game and receives $\sset{\sk_{i,j}}_{j\in[z]}$ from the challenger, which it forwards to $\Adv$.
	Next, $\BS$ receives verification keys $\sset{\vk_{i,j}'}_{i\in\IS, j\in[z]}$ of the corrupted parties from $\Adv$.
In the bare-PKI mode, $\BS$ updates $\vk_{i,j}=\vk_{i,j}'$ for each $i\in\IS, j\in[z]$.
	
	\item $\BS$ then proceeds to simulate Steps \ref{step:faecomm_first}, \ref{step:fba}, \ref{step:fct}, and \ref{step:faecomm_second} as described in the simulator.
	$\BS$ chooses $m=(y,s)$ and $\setP=\DS$ and sends it to the challenger. For each $i\in\DS$ and $j\in[z]$, it sets $m_{i,j}=(y_i,s_i)$ as received from the adversary.
	\item
	$\BS$ receives signatures $\sset{\sigma_{i,j}}_{i\in [n]\setminus \IS,j\in[z]}$ of the honest parties from the challenger and forwards them to the adversary $\Adv$.
	\item $\BS$ then simulates Steps \ref{step:sign}, \ref{step:recurse}, \ref{step:faecomm_third}, \ref{step:prf}, and \ref{step:output} as described in the simulator.
	\item Finally if $\Adv$ manages to send a valid SRDS on a message other than $(y,s)$ to any of the honest parties,
	$\BS$ forwards that to the challenger.
\end{itemize}
Clearly, $\BS$ wins the forgery game only if $\Adv$ succeeds in forging a valid SRDS on a message other than $(y,s)$. Since our SRDS scheme is unforgeable, this only happens with negligible probability.
\QED
\end{proof}	
\begin{claim}[Validity]\label{claim:validity}
For any adversarial strategy of $\Adv$, if there exists a value $x$ such that $x_i=x$ for each honest party $P_i\in[n]\setminus\IS$, then the output of all honest parties is $y=x$.
\end{claim}
\begin{proof}
From \cref{claim:agreement}, we know that with overwhelming probability, the final output $y$ of all honest parties is the same as the output of $\fba$ in Step~\ref{step:fba}. All that remains to prove now is that if there exists a value $x$, such that $x_i=x$ for each honest party $P_i\in[n]\setminus\IS$, then the output of $\fba$ in Step~\ref{step:fba} is $x$. Recall that $\fba$ in Step~\ref{step:fba} is computed over the inputs of all parties in the supreme committee $\CS$. From \cref{def:com-tree}, we know that at least $2/3$ fraction of the parties in $\CS$ are honest. Therefore, if there exists a value $x$ such that $x$ is the input of all honest parties, then the input of all honest parties in $\CS$ is also $x$. Now, irrespective of the inputs of the remaining malicious parties in $\CS$, from the validity of $\fba$, we are guaranteed that the output of $\fba$ is $y=x$.
\QED
\end{proof}

This concludes the proof of \cref{lem:ba}.
\QED
\end{proof}

%% file: SRDS_cont.tex
\section{Constructions of SRDS (Cont'd)}\label{sec:SRDS_cont}

In this section we present the proofs on the SRDS constructions from \cref{sec:SRDS:construction}.

\subsection{SRDS from One-Way Functions (Cont'd)}\label{sec:SRDS:owf_cont}

We now present the proof of \cref{thm:srds_owf}.

\medskip\noindent
\textbf{\cref{thm:srds_owf}.}
\emph{\ThmSRDSOWF}

\input{SRDS-owf_cont}

\subsection{SRDS from SNARKs (Cont'd)}\label{sec:SRDS:snark_cont}

We present the proof of \cref{thm:SNARKSconstruction}.

\medskip\noindent
\textbf{\cref{thm:SNARKSconstruction}.}
\emph{\ThmSRDSSNARKs}

\input{SRDS-snarks_cont}

%% file: SRDS-owf_cont.tex
\begin{proof}[Proof of \cref{thm:srds_owf}]
In \cref{lem:srds_owf_succinct}, we prove succinctness, in \cref{lem:srds_owf_robust}, we prove robustness, and in \cref{lem:srds_owf_forge}, we prove unforgeability.

\begin{lemma}\label{lem:srds_owf_succinct}
The construction in \cref{fig:srds_from_owf} is succinct.
\end{lemma}
\begin{proof}
We start by proving the size of the signatures is succinct.
Let $\CS=\sset{i\mid \sk_i\neq\bot}$ and let $X$ be a random variable representing $\ssize{\CS}$.
By construction, $\Exp[X]=\ell$ and $\ell=\omega(\log(n))$. Therefore, by Chernoff bound for $\mu=\ell$ and $\delta=1/2$,\footnote{The exact Chernoff bound used is $\pr{\ssize{X-\mu}\geq\delta \mu}\leq 2e^{-\mu\delta^2/3}$ for $0<\delta<1$, where $\mu=\Exp[X]$.} it holds that
\[
\pr{\ssize{X-\ell}\geq \ell/2} \leq 2e^{-\ell/12} = \negl(n).
\]
We therefore conclude that $\ell/2\leq \ssize{\CS}\leq 3\ell/2$ with overwhelming probability (in $n$).
By definition of digital signatures, every signature in the support of $\DSSign$ is polynomial in $\secParam$. Therefore, every $\sigma$ in the support of $\TSSignShare$ (of the SRDS scheme) is also polynomial in $\secParam$.
By construction, unless an adversary is able to successfully break the obliviousness of the signature scheme (which only happens with negligible probability in $\secParam$), an aggregate signature only consists of $\ssize{\CS}$ ``base'' signatures from the parties in $\CS$.
Further, in the negligible event where the aggregate signature consists of more than $\ssize{\CS}$ base signatures, the output is $\bot$.
Therefore, the length of an aggregated signature is bounded by $\alpha(n,\secParam)\in\poly(\log{n},\secParam)$.

Proving decomposability is immediate. Since the aggregation algorithm is deterministic, it can be entirely captured by the first algorithm $\TSAggr_1$, which outputs a set of $\polylog(n)$ signatures (since there are at most $\ssize{\CS}$ signatures, with all but negligible probability). The second algorithm $\TSAggr_2$ simply outputs the same set of signatures.
\QED
\end{proof}

\begin{lemma}\label{lem:srds_owf_robust}
The construction in \cref{fig:srds_from_owf} is $\beta n$-robust.
\end{lemma}
\begin{proof}
Let \Adv be a PPT adversary. We will show that \Adv can win the game $\ExptTSrobust_{\trpki,\Pi,\Adv}(\secParam,n,\beta n)$ (with the trusted PKI mode) with at most negligible probability.
The game begins when the challenger computes $\pp\gets\TSCR(1^\secParam, 1^n)$ and $(\vk_i,\sk_i)\gets \TSGen(\pp)$ for every $i\in[n]$. Denote $\CS=\sset{i\mid \sk_i\neq\bot}$. Next, the adversary adaptively selects the set of corrupted parties; denote by $\IS$ the set of corrupted parties.

In the \emph{robustness challenge} phase, the adversary \Adv chooses an $(n,\IS,\mathsf{robust})$-\textsf{almost-everywhere communication tree} $\tree=(\vertex,\edge)$ (see Definition \ref{def:rob-com-tree}). It also chooses a message $m\in \MS$ and $\{m_i\}_{i\in \goodP}$, where $\goodP$ is the set of honest parties that are assigned to leaf nodes that do not have a good path to the root.

Recall that there are $n/\log^5 n$ leaf nodes in this tree out of which all but $3/\log n$ fraction have a good path to the root. In other words, the signatures of the parties assigned to ``good'' leaf nodes are guaranteed to be part of the final aggregate signature. Total number of parties assigned to the good leaf nodes are $\log^5 n \left(1-\frac{3}{\log n}\right)\frac{n}{\log^5 n}=n\frac{\log n-3}{\log n}$. Let us use $S$ to denote this set of parties. We proceed to show that with overwhelming probability (in $n$), there are more than $\ell'/3$ honest parties in $S$, that have a valid signing key, where $\ell'=\ell/2$.

\begin{claim}\label{claim:robusthonestcommittee}
$\pr{\ssize{\CS\cap(S\setminus\IS)}\leq\ell'/3}\leq \negl(n)$.
\end{claim}
\begin{proof}
We know that $\ssize{S}=n\frac{\log n-3}{\log n}>2n/3$. In order to maximize its chance of winning the robustness game, an adversary who is allowed to arbitrarily choose the set $S$, will without loss of generality include all the corrupted parties in $S$.
Denote by $\HS_S=S\setminus\IS$ the set of honest parties in $S$. Since $\ssize{\IS}=(1/3-\epsilon)\cdot n$ (where $\epsilon=1/3-\beta$), it holds that
\[
\ssize{\HS_S}> \frac{2}{3}\cdot n - \left(\frac{1}{3}-\epsilon\right)\cdot n=\left(\frac{1}{3}+\epsilon\right)\cdot n.
\]
Thus, there are more than $(1/3+\epsilon)\cdot n$ honest parties in the set $S$.
Given the information with the adversary and the fact that the set of parties with valid signing keys are chosen at random, he will get the same success probability for any arbitrary choice of $\CS$.
Let $X$ be a random variable representing the number of honest parties in $S$ who have a valid signing key, \ie $\ssize{\CS\cap\HS_S}$.
If $\pr{\ssize{\CS\cap\HS_S}\leq\ell'/3}\leq \negl(n)$ holds for $\ssize{\CS}=\ell'$, it will also hold for any $\ssize{\CS}>\ell'$.
By \cref{lem:srds_owf_succinct}, we know that $\ssize{\CS}\geq \ell'$ with an overwhelming probability.
Therefore, we can assume that $\ssize{\CS}=\ell'$; in this case it holds that
\[
\Exp[\ssize{\CS\cap \HS_S}]=\frac{\ell'}{n} \cdot \left(\frac{1}{3}+\epsilon\right)\cdot n=\left(\frac{1}{3}+\epsilon\right)\cdot\ell'.
\]
By Chernoff bound for $\mu=\ell'\left(1/3+\epsilon\right)$ and $\delta=3\epsilon/(1+3\epsilon)$,\footnote{The exact Chernoff bounds used is $\pr{X\leq (1-\delta)\mu}\leq e^{-\frac{\delta^2}{2}\mu}$ for $0<\delta<1$, where $\mu=\mathbb{E}[X]$.} it holds that
\begin{align*}
\pr{X\leq (1-\delta)\mu}
&=
\pr{X\leq\left(1-\frac{3\epsilon}{1+3\epsilon}\right)\cdot\ell'\cdot\left(\frac{1}{3}+\epsilon\right)}\\
&=
\pr{X\leq\left(\frac{1}{1+3\epsilon}\right)\cdot\ell'\cdot\left(\frac{1+3\epsilon}{3}\right)}\\
&=
\pr{X\leq\ell'/3}\\
&\leq
e^{-\frac{9\epsilon^2}{2(1+3\epsilon)^2}\ell'(1/3+\epsilon)}\\
&=
e^{-\frac{3\epsilon^2}{2(1+3\epsilon)}\ell'}.
\end{align*}
Since $\epsilon>0$ is constant, we conclude that
\[
\pr{X\leq\ell'/3}\leq e^{-\omega(\log{n})}=\negl(n).
\]
Hence, for any arbitrary strategy deployed by the adversary, the probability that less than $\ell'/3$ honest parties with a valid signing key are chosen in the set $S$ is negligible.
\QED
\end{proof}

The robustness phase proceeds with the challenger signing the message $m$ on behalf of all the honest parties $\sset{\sigma_i}_{i\in [n]\setminus (\IS\cup \goodP)}$ and signing the respective messages $m_i$ on behalf of parties in $\goodP$ and handing their signatures to \Adv who responds with signatures for corrupted parties $\sset{\sigma_i}_{i\in\sset{\IS}}$ (potentially also for parties whose signing key is $\bot$). As described in \cref{fig:exp_robust}, using these ``base'' signatures $\sset{\sigma_i}_{i\in [n]}$, the challenger then interacts with the adversary according to $\tree=(\vertex,\edge)$ to compute the aggregate signature $\sigma$.

\begin{claim}\label{claim:robustness}
$\pr{\TSVerify(\pp,\sset{\vk_1,\ldots,\vk_n},m,\sigma)=0}\leq\negl(\secParam,n)$.
\end{claim}
\begin{proof}
An accepting signature on a message $m$ consists of at least $\ell'/3$ valid signatures of the form $\sigma_i=(i,m,\signsig_i)$, satisfying $\DSVerify(\vk_i,m,\signsig_i)=1$. As proved earlier in \cref{lem:srds_owf_succinct}, since $\ell\in\omega(\log n)$ it holds with overwhelming probability that $\ell/2\leq \ssize{\CS}\leq 3\ell/2$; therefore, by the obliviousness of the signature scheme that the aggregate signature can consist of at most $\ssize{\CS}$ base signatures.

The aggregate algorithm then checks if the ``base'' signatures contain a valid signature on $m$. We rely on the correctness of the underlying digital signature scheme to ensure that only valid signatures from the adversary (\ie by committee members) get aggregated with an overwhelming probability (in $\secParam$).

Additionally, in the case where the adversary does not provide sufficiently many valid signatures, from  \cref{claim:robusthonestcommittee} we know that the number of honest parties in $S$ with a valid signing key is more than $\ell'/3$ with an overwhelming probability (in $n$). Therefore, the signatures of these honest parties are sufficient for generating an accepting signature.
\QED
\end{proof}

This concludes the proof of \cref{lem:srds_owf_robust}.
\QED
\end{proof}

\begin{lemma}\label{lem:srds_owf_forge}
The construction in \cref{fig:srds_from_owf} is $\beta n$-unforgeable.
\end{lemma}
\begin{proof}
Let \Adv be a PPT adversary.
We will show that \Adv can win the game $\ExptTSforge_{\trpki,\Pi,\Adv}(\secParam,n,\beta n)$ with at most negligible probability. The game begins when the challenger computes $\pp\gets\TSCR(1^\secParam,1^n)$ and $(\vk_i,\sk_i)\gets \TSGen(\pp)$ for every $i\in[n]$. Next, the adversary adaptively selects the set of corrupted parties;
denote by $\IS$ the set of corrupted parties.

In the \emph{forgery challenge} phase, the adversary $\Adv$ chooses a subset $S\subseteq[n]\setminus\IS$ such that  $\ssize{S\cup\IS}< (1/3-\epsilon')n$ for some constant $0<\epsilon'<\epsilon$, where $\epsilon=1/3-\beta$,
and messages $m$ and $\sset{m_i}_{i\in S}$ from $\MS$. We now prove that with an overwhelming probability (in $n$), the fraction of parties who have a valid signing key in $S\cup\IS$ is less than a third.

\begin{claim}\label{claim:unforgeablehonestcommittee}
The number of parties with a valid signing key in a set $S\cup\IS$ is less than $\ell'/3$ with an overwhelming probability in $n$, \ie
\[
\Pr\left[ \ssize{\CS\cap(S\cup\IS)} \geq \ell'/3\right] \leq \negl(n).
\]
\end{claim}
\begin{proof}
The parties with a valid signing key are chosen at random, and the information about whether a party has a valid signing key is not revealed to the adversary $\Adv$, unless it chooses to corrupt that party or it sees a signature from that party.
The adversary chooses the honest set $S$ only based on the knowledge of corrupted parties and their signing keys. Given this information with the adversary and the fact that the parties with valid signing keys are chosen at random, he will get the same success probability for any arbitrary choice of $S$.

Let $X$ be a random variable representing the number of parties in $\CS\cap(S\cup\IS)$.
If for $\ssize{\CS}=3\ell/2$ it holds that $\pr{\ssize{\CS\cap (S\cup\IS)}\leq\ell'/3}\leq \negl(n)$, it will also hold for any $\ssize{\CS}<3\ell/2$. By \cref{lem:srds_owf_succinct}, we know that $\ssize{\CS}<3\ell/2$ with an overwhelming probability.
Therefore, we an assume that $\ssize{\CS}=3\ell/2=3\ell'$; in this case it holds that $\Exp[X]=3(1/3-\epsilon'')\ell'$ for some $\epsilon''>\epsilon'$.
By Chernoff bound for $\mu=3\ell'\left(1/3-\epsilon''\right)$ and $\delta=\frac{9\epsilon''-2}{3-9\epsilon''}$ (note that $\delta>0$ since $0<\epsilon''<1/3$),\footnote{The exact Chernoff bound used is $\pr{X\geq (1+\delta)\mu}\leq e^{-\frac{\delta^2}{(2+\delta)}\mu}$ where $\mu=\mathbb{E}[X]$} it holds that
\begin{align*}
\pr{X\geq(1+\delta)\mu}
&=
\pr{X\geq\left(1+\frac{9\epsilon''-2}{3-9\epsilon''}\right)\cdot\ell'\cdot\left(\frac{1}{3}-\epsilon''\right)\cdot 3}\\
&=
\pr{X\geq\left(\frac{1}{3-9\epsilon''}\right)\cdot\ell'\cdot\left(\frac{3-9\epsilon''}{3}\right)}\\
&=
\pr{X\geq\ell'/3}\\
&\leq
e^{-\frac{\delta^2}{2+\delta}\mu}\\
&=
e^{-\frac{(9\epsilon''-2)^2/(3-9\epsilon'')^2}{(4-9\epsilon'')/(3-9\epsilon'')}3\ell'(1/3-\epsilon'')}\\
&=
e^{-\frac{(9\epsilon''-2)^2}{3(4-9\epsilon'')}\ell'}.
\end{align*}
Since $0<\epsilon''<1/3$ is a constant, it holds that $4-9\epsilon''>0$, hence we conclude that
\[
\pr{X\geq \ell'/3}\leq e^{-\omega(\log{n})}=\negl(n).
\]
Hence, the probability that for any arbitrary strategy deployed by the adversary, the probability that more than $\ell'/3$ of the parties with a valid signing key are in $S\cup\IS$ is negligible.
\QED
\end{proof}

The \emph{forgery challenge} phase proceeds when for each $i\in S$, the challenger signs the message $m_i$ on behalf of honest $\Party_i$, and signs the message $m$ on behalf of all the remaining honest parties $i\not\in S\cup\IS$.
Next, the challenger hands these signatures $\sset{\sigma_i}_{i\in[n]\setminus\IS}$ to $\Adv$ who responds with an aggregate signature $\sigma'\in\XS$ and a message $m'\in\MS$.
\begin{claim}
$\Pr[(\TSVerify(\pp,\sset{\vk_1,\ldots,\vk_n}, m',\sigma')=1) \wedge (m'\neq m)]\leq\negl(\secParam,n)$.
\end{claim}
\begin{proof}
An accepting signature on any message $m'\neq m$ consists of at least $\ell'/3$ valid signatures of the form $\sigma_i=(i,m',\signsig_i)$, satisfying  $\DSVerify(\vk_i,m',\signsig_i)=1$.
	
By \cref{claim:unforgeablehonestcommittee}, the number of parties with a valid signing key in $S\cup\IS$ is less than $\ell'/3$ with an overwhelming probability (in $n$). Essentially, the adversary receives valid signatures on a message other than $m$ only from less than $\ell'/3$  parties (in $\CS\cap(S\cup\IS)$). Hence, the only way \Adv can produce more than $\ell'/3$ valid signatures on any message other than $m$ is by forging a valid signature for a corrupt party whose signing key is $\bot$ or by forging a signature for an honest party. Since the verification keys of the parties whose signing keys are $\bot$ correspond to oblivious keys, we rely on the obliviousness of these keys (see \cref{def:Okeygen}) to ensure that this only happens with negligible probability (in $\secParam$). Similarly we can rely on the unforgeability of a digital signature scheme to ensure that an adversary will be able to forge a valid signature for an honest party with a valid signing key only with a negligible probability (in $\secParam$).	
Hence, except with negligible probability $\negl(\secParam,n)$, the adversary is unable to forge an accepting SRDS signature.
\QED
\end{proof}
This concludes the proof of \cref{lem:srds_owf_forge}
\QED
\end{proof}
This concludes the proof of \cref{thm:srds_owf}.
\QED
\end{proof}

%% file: SRDS-snarks_cont.tex
\begin{proof}[Proof of \cref{thm:SNARKSconstruction}]
In \cref{lem:SNARKSsuccinct} we will show that the construction in \cref{fig:srds_from_snarks} is succinct, in \cref{lem:SNARKSrobust}, we will show robustness and in \cref{lem:SNARKSunforge}, we will show unforgeability.

\begin{lemma}\label{lem:SNARKSsuccinct}
The construction in \cref{fig:srds_from_snarks} is succinct.
\end{lemma}

\begin{proof}
We start by proving the size of the signatures is succinct.
Each SRDS signature consists of a ``truncated transcript'' $z'$ of size $(\ssize{m}+\ssize{c}+\ssize{\indexMax}+\ssize{\indexMin}+\ssize{\gamma})$ along with a proof $\pi$.
For ``base'' SRDS signatures, $\gamma$ corresponds to a digital signature, and in all other cases $\gamma=\bot$.
By definition, the size of each digital signature is $\poly(\secParam)$. Hence, the total size of each truncated transcript $z'$ is $\poly(\secParam)+\log n+\log n+\log n+\log n=\poly(\secParam)+O(\log(n))$.
Since $\pi=\bot$ for base signatures, the total size of each base SRDS signature (truncated transcript + digital signature) is $\poly(\secParam)+O(\log(n))$, and is thus succinct.

In each aggregate SRDS signature, this proof corresponds to the output of $\PCDProve$. In our construction, the size of PCD transcript $z$ is $\ssize{z'}+\ssize{\hash_\vk}+\ssize{k}+\ssize{p}=\poly(\secParam)+O(\log(n))$.
The Merkle verification algorithm runs in time $\poly(\secParam)+\polylog(n)$; therefore, by construction, the size of the compliance predicate is $\poly(\secParam)+\polylog(n)$ and the bound $B$ on its running time is $\ssize{\setsig}\cdot(\secParam+\polylog(n))$, where $\ssize{\setsig}\leq q\leq n$. Therefore, by the \emph{succinctness} property of PCD systems (see \cref{sec:PCD}), the size of each proof is  $\poly(k+\log B)=\poly(\secParam)\cdot\polylog(n)$. Hence, the total size of each aggregate signature is $\poly(\secParam)\cdot\polylog(n)$.

The time required to verify validity of each ``base'' signature in this construction is $\poly(\log n, \secParam)$ (here $\polylog(n)$ appears because of the binary representation of indices).  The time required to verify a PCD proof in our construction is $\poly(\secParam+|C|+|z|+\log B)=\poly(\secParam+\log n)$ (\cref{def:PCD}). Finally, the time required to generate an aggregate signature is equal to the time required to compute $z_\outputvar$ and the time to run $\PCDProve$. The time required to generate $z_\outputvar$ includes the time required to compute Merkle hash on all the verification keys, which is $\poly(\secParam,n)$, and the time required to verify in the incoming transcripts and proofs, which is $q\cdot\poly(\secParam+\log n)$. Therefore, the running time of $\TSAggr_1$ is $q\cdot\poly(\secParam,n)$.
The time required to run $\TSAggr_2$ includes the time required for computing $z_\outputvar$ given the above information, which is $\ssize{\setsig}\cdot O(\log n)$ and the time required to run $\PCDProve$, which is $O(\log n)+\poly(\secParam+|C|+\log B)=\poly(\secParam+\log n)$ (see \cref{def:PCD}). Therefore, the total time required to run $\TSAggr_2$ is $\ssize{\setsig}\cdot\poly(\secParam+\log n)=\poly(\log n,k)$ (since $\|\setsig\|$ is bounded by $\alpha(n,\secParam)\in\poly(\log n, \secParam)$ as enforced by the check in $\TSAggr_1$).
\QED
\end{proof}

\begin{lemma}\label{lem:SNARKSrobust}
The construction in \cref{fig:srds_from_snarks} is $t$-robust.
\end{lemma}

\begin{proof}
Let \Adv be a PPT adversary.
We will show that \Adv can win the game $\ExptTSrobust_{\bbpki,\Pi,\Adv}(\secParam,n,t)$ with at most negligible probability.
The game begins when the challenger computes $\pp=(1^\secParam,\PCDsigma,\PCDtau,\seed)\gets\TSCR(1^\secParam)$ and $(\vk_i,\sk_i)\gets \TSGen(\pp)$ for every $i\in[n]$. Next, the adversary adaptively selects the set of corrupted parties $\IS$ and determines their verification keys.

In the \emph{robustness challenge} phase, the adversary \Adv chooses an $(n,\IS,\mathsf{robust})$-\textsf{almost-everywhere communication tree} $\tree=(\vertex,\edge)$ (See Definition \ref{def:rob-com-tree}). It also chooses a message $m\in \MS$ and $\{m_i\}_{i\in \goodP}$, where $\goodP$ is the set of honest parties that are assigned to leaf nodes that do not have a good path to the root.
Recall that there are $n/\log^5 n$ leaf nodes in this tree out of which all but $3/\log n$ fraction have a good path to the root. In other words, the signatures of the parties assigned to ``good'' leaf nodes are guaranteed to be part of the final aggregate signature. Total number of parties assigned to the good leaf nodes are $\log^5 n \left(1-\frac{3}{\log n}\right)\frac{n}{\log^5 n}=n\frac{\log n-3}{\log n}$. Let us use $S$ to denote this set of parties. Then $\ssize{S}=n\frac{\log n-3}{\log n}>2n/3$. Since $\ssize{\IS}<n/3$, it holds that the number of honest parties $\HS_S$ in the set $S$ is at least $\ssize{\HS_S}\geq 2n/3-\ssize{\IS}>n/3$.

Next, the adversary gets signatures $\sset{\sigma_i}_{i\in [n]\setminus\IS}$ of all the honest parties on the respective messages (i.e., on message $m_i$ for $i\in \goodP$ and on message $m$ for $i\in[n]\setminus(\mathcal{I}\cup \goodP)$) and it computes signatures of corrupted parties $\sset{\sigma_i}_{i\in \IS}$.
As described in \cref{fig:exp_robust}, the challenger then interacts with the adversary according to $\tree=(\vertex,\edge)$, to compute the aggregate signature $\sigma$.

Recall that the aggregation algorithm first checks the validity of incoming transcripts and proofs and only aggregates transcripts with a convincing proof. Starting from the ``base'' signatures, if the adversary does not provide valid signatures on  $m$ on behalf of the corrupted parties, they will not pass the validity check at level $\ell=2$ (this follows from the \emph{correctness} of the digital signature scheme). The aggregation algorithm on the remaining ``verified'' base signatures mimics the interactive protocol $\ProofGen$  (as described in the \emph{completeness} definition of PCD in \cref{sec:PCD}).
The tree $\tree$ chosen by the adversary acts as the \emph{distributed-computation generator} $G$ (see \cref{def:PCD}). For each ``good'' node in $\tree$, the reconstruction algorithm aggregates the signatures (\ie computes a $\compliance$-compliance transcript and PCD proof)  from its incoming edges and labels the outgoing edges from the node with this partially aggregated signature. For every ``bad'' node in $\tree$, the adversary can provide an arbitrary signature of its choice. From soundness of PCDs, it follows that the adversary cannot give a faulty proof/partially aggregate signature that verifies. The aggregation algorithm halts at the root node and outputs the corresponding truncated transcript and proof (\ie the aggregated signature $(z_{\outputvar}',\pi_\outputvar)$). From this construction, we now have that the output transcript is compliant with $\compliance$, and even if the adversary does not provide valid partially aggregate signatures for bad nodes, since there were at least $n/3$ honest signatures from the honest parties that also had a good path to the root node, from the correctness of the digital signature scheme and completeness of the Merkle hash proof system, it follows that $c_{\outputvar}\geq n/3$. Robustness now follows from the \emph{completeness} and \emph{succinctness} of the PCD system.
\QED
\end{proof}

\begin{lemma}\label{lem:SNARKSunforge}
The construction in \cref{fig:srds_from_snarks} is $t$-unforgeable.
\end{lemma}
\begin{proof}
Let \Adv be a PPT adversary.
We will show that \Adv can win the game $\ExptTSforge_{\bbpki,\Pi,\Adv}(\secParam,n,t)$ with at most negligible probability. The game begins when the challenger computes $\pp=(1^\secParam,\PCDsigma,\PCDtau,\seed)\gets\TSCR(1^\secParam)$ and $(\vk_i,\sk_i)\gets\TSGen(\pp)$ for every $i\in[n]$. Next, the adversary adaptively selects the set of corrupted parties and determines their verification keys; denote by $\IS$ the set of corrupted parties.

In the \emph{forgery challenge} phase, the adversary $\Adv$ chooses a subset $S\subseteq[n]\setminus\IS$, such that $\ssize{S\cup\IS}<n/3$, and messages $m$ and $\sset{m_i}_{i\in S}$ from $\MS$. Subsequently, for each $i\in S$, the challenger signs the message $m_i$ on behalf of honest $\Party_i$, and signs the message $m$ on behalf of all the remaining honest parties $i\not\in S\cup\IS$.
Next, the challenger hands these signatures $\sset{\sigma_i}_{i\in[n]\setminus\IS}$ to $\Adv$ who responds with an aggregate signature $\sigma'\in\XS$ and a message $m'\in\MS$.

Let us assume for the sake of contradiction that the adversary manages to generate an aggregate signature $\sigma'=(z',\pi)$, such that $\TSVerify(\pp,\sset{\vk_1,\ldots,\vk_n},m',\sigma')=1$ and  $m'\neq m$. From the proof of knowledge property of the PCD system, we know that given a verifying proof from a polynomial-size prover, there exists a polynomial-size extractor $\mathbb{E}_{\PCDProve}$ that can extract the witness. Recall that given a vector of input transcripts $\vect{z}_\inputvar$ and an output transcript $z_\outputvar$, the compliance predicate in our construction checks if the maximas and minimas of the input and output transcripts are ordered properly, the value of counter $c$ in the output transcript is equal to the sum of the counter values in the input transcripts and that the same Merkle hash of keys is used in all transcripts. Additionally, if any of the input transcripts correspond to base signatures, the compliance predicate also checks that the signature is valid \wrt the verification key specified in that transcript and also verifies the Merkle proof corresponding to this key and the Merkle hash. We now design an adversary $\BS$ that uses this extractor to either break unforgeability of the digital signature scheme or break soundness of the Merkle hash proof system.
The adversary $\BS$ starts by computing $z=z'||(\hash_\vk,\bot,\bot)$, where $\hash_\vk=\mhash(\seed,(1||\vk_1),\ldots,(n||\vk_n))$, initializing $\validset=\emptyset$ and running the following recursive algorithm $\BS_{\mathsf{ext}}(\PCDsigma,z)$:
\begin{enumerate}
	\item Compute $\trans\gets\mathbb{E}_{\PCDProve}(\PCDsigma,z)$.
	\item If $\compliance(\trans)=1$, for each valid input ``base'' transcript in $\trans$ of the form $(z_i,\bot)$ on $m'$ with $z_i=(m',1,i,i,\gamma_i,\hash_\vk,k_i,p_i)$ and $\gamma_i\neq\bot$, set  $\validset=\validset\cup\sset{(z_i,\pi_i)}$. For each partially aggregated signature on $m'$ in $\trans$ of the form $(z_i,\pi_i)$ with $z_i=(m',\cdot,\cdot,\cdot,\cdot,\cdot,\cdot,\cdot)$, check whether $\PCDVerify(\PCDtau, z_i, \pi_i)=1$ and if so, run $\BS_{\mathsf{ext}}(\PCDsigma,z_i)$.
\end{enumerate}
If $\ssize{\validset}\geq n/3$, the adversary $\BS$ succeeds in extracting at least $n/3$ transcripts of the form $(m',1,i,i,\gamma_i,\hash_\vk,k_i,p_i)$, each with a distinct $i$ (as enforced by the checks on the maximas and minimas) such that the following holds for each of these transcripts:
\begin{enumerate}[label=(\alph*)]
	\item $\DSVerify(\vk_i,m',\gamma_i)=1$.
	\item $\mverify(\seed,(i||k_i),\hash_\vk,p_i)=1$.
\end{enumerate}
Since $\hash_\vk$ was computed honestly by $\BS$, it holds for each extracted ``base'' transcript that either $\gamma_i$ is a valid signature \wrt $k_i=\vk_i$, or if $k_i\neq \vk_i$, then the adversary $\Adv$ has managed to break the soundness of the Merkle proof hash proof system. However, from \cref{thm:merkle-hash}, we know that this only happens with  at most negligible probability (in $\secParam$).
Now, since each $i$ (and thereby each $k_i$) is distinct in the extracted ``base'' transcripts, adversary $\BS$ has managed to extract at least $n/3$ valid signatures ($\gamma_i$) on $m'$. Since the adversary only had access to signatures on $m'$ from less than $n/3$ parties, this would imply that it has successfully forged signatures of some honest parties in the set $[n]\setminus S$. From \emph{unforgeability} of the digital signature scheme, we know that this can only happen with at most negligible probability (in $\secParam$).
\QED
\end{proof}

\cref{lem:SNARKSsuccinct,lem:SNARKSrobust,lem:SNARKSunforge} rely on PCD systems for logarithmic-depth and polynomial-size compliance predicates. By \cref{PCD_from_snarks}, such PCD systems exist assuming the existence of SNARKs with linear extraction.
This concludes the proof of \cref{thm:SNARKSconstruction}.
\QED
\end{proof}

%% file: sym-poly.tex
\subsection{Proof of \texorpdfstring{\cref{thm:sym-poly-nphard}}{Lg}}\label{sec:sym-poly}

\noindent{\bf \cref{thm:sym-poly-nphard}.}
\emph{\ThmSymPolyNPHard}

\medskip
\begin{proof}
We divide the proof as follows: (1) First, we show that for any ring $R=\field^n$ with Hadamard product satisfying $|R|=2^{\Theta(n)}$, then for any elementary symmetric polynomial $\phi_2$, the $R$-Subset-$\phi_2$ problem (see \cref{def:subset-f}) is NP-complete by showing a reduction to 3-SAT. (2) Second, we show the same for any $\phi_\ell$, where $\ell\geq 3$. (3) Finally, we show how these reductions can be modified to prove the existence of $s\in\Theta(n)$, for which $(s,R)$-Subset-$\phi_{\ell}$ (see \cref{def:subset-f}) is also NP-complete.
\medskip

\noindent\textbf{For $R$-Subset-$\phi_2$: }Given a 3-CNF formula $\Phi$ over variables $x_1,\ldots,x_N$ with clauses $C_1,\ldots,C_m$, each containing exactly three distinct literals, the reduction algorithm constructs an instance $x=(a_1,\ldots,a_{2+2N+3m},t)$ of the $R$-Subset-$\phi_2$ problem such that $\Phi$ is satisfiable if and only if there exists a subset $S\subseteq[2+2N+3m]$, such that $\phi_2(\sset{a_i}_{i\in S})=t$. The reduction algorithm constructs elements in $R=\field^{1+N+m}$ as follows:
\begin{enumerate}
	\item A special element $a_1=\alpha_0\in R$, whose first entry is 1 and all other entries are 0.
	\item A special element $a_2=\alpha_1\in R$ whose first $n+1$ entries correspond to $1$.
	\item For each variable $x_i$ (for $i\in [N]$), define two elements $a_{2+2i+1}=v_i\in R$ and $a_{2+2i+2}=v_i'\in R$ such that the \iith{(1+i)} entry of these elements is set to $1$.
	\item Define three elements $a_{2+2N+3j+1}=c_j^1,a_{2+2N+3j+2}=c_j^2,$ and $a_{2+2N+3j+3}=c_j^3$ corresponding to each clause $C_j$ (for $j\in[m]$). The \iith{(1+N+j)} entry in $c_j^1$ corresponds to 9, the \iith{(1+N+j)} entry in $c_j^2$ corresponds to 4 and the \iith{(1+N+j)} entry in $c_j^3$ corresponds to 2.
	The remaining entries in each of these correspond to $0$.
	\item The target element $t$ is also a  vector of $1+N+m$ elements in $\field$. The first $1+N$ entries in $t$ are set to $1$, while the remaining entries are set to 9.
\end{enumerate}

We now prove completeness and soundness of this reduction:
\paragraph{Completeness.}
Suppose $\Phi$ has a satisfying assignment $X$. We will construct a subset $S\subseteq [2+2N+3m]$ such that $\phi_2(\sset{a_i}_{i\in S})=t$. For each variable $x_i$, if $x_i$ is set to 1 in $X$, we include $a_{2+2i+1}=v_i$ in $S$, else we include $a_{2+2i+2}=v_i'$ in $S$.
We also include the two special elements $a_1=\alpha_0$ and $a_2=\alpha_1$ in $S$. Note that, $\alpha_0$ and $\alpha_1$ are the only elements whose first entry is 1, the first entry of all other elements is 0. This ensures that we have exactly 2 elements with value 1 in the first column.
Thus, the first entry of $t$ is guaranteed to be 1.
Also, apart from $\alpha_1$, for each $1\leq i\leq N$, there are only two other elements $v_i$ and $v_i'$ whose \iith{(1+i)} entry is set to 1. Including one of these for each $1\leq i\leq n$ along with $\alpha_1$ ensures that there are exactly two elements with value 1 in the \iith{(1+i)} column. Therefore, we are guaranteed to get 1 in each of the first $1+N$ entries of $t$.
	
Since $X$ is a satisfying assignment, each clause must contain at least one literal with the value 1. For each clause $C_j$, if there is exactly one literal with value 1 in the satisfying assignment $X$, we include $c_j^1$.
Note that $S$ now has exactly one element whose \iith{(1+N+j)} entry is set to 1 and exactly one element with 9 in this column.
All other elements in the subset have 0's in this position. This ensures that the \iith{(1+N+j)} entry of $t$ adds up to 9. If there are exactly two literals with value 1, we include $c_j^2$. In this case, there are exactly two elements that have value 1 in the \iith{(1+N+j)} column and exactly one element that has a value of 4 in this position. All other elements in the subset have 0's in this position. This ensures that the \iith{(1+N+j)} entry of $t$ adds up to
\[
(1\cdot 1) + (1\cdot 4) + (1\cdot 4)=9.
\]
Finally, if there are exactly three literals with value 1, we include $c_j^3$.  In this case, there are exactly three elements that have value 1 in \iith{(1+N+j)} column and exactly one element that has a value of 2 in this position. All other elements in the subset have 0's in this position. This ensures that the \iith{(1+N+j)} entry in $t$ adds up to
\[
(1\cdot 1) + (1\cdot 1)+(1\cdot 1)+(1\cdot 2)+(1\cdot 2) + (1\cdot 2)=9.
\]
Thus, the last $m$ entries in $t$ all add up to 9.

\paragraph{Soundness.}
Suppose there exists a subset $S\subseteq [2+2N+3m]$ whose pairwise sum of products is $t$. We show that this implies that there must be a satisfying assignment for $\Phi$. Note that $\alpha_0$ and $\alpha_1$ are the only elements whose first entry is 1, while the first entry of all other elements is set to 0. Since the first entry in $t$ is required to be 1, both $\alpha_0$ and $\alpha_1$ must be included in the set $S$.
	
For each $1\leq i\leq N$, there are exactly three elements $\alpha_1$, $v_i$ and $v_i'$ whose \iith{(i+1)} entry is 1. Since we have already included $\alpha_1$ in $S$, if we include both $v_i$ and $v'_i$, then the \iith{(i+1)} entry in result of $\phi_2$ applied over $S$ will be $(1\cdot1)+(1\cdot1)=2$. Since the characteristic of the field $\field$ is at least $63$, we know that $2\neq 1$. Therefore, we are assured that only one of $v_i$ or $v_i'$ can be included, but not both. Therefore, for each $1\leq i\leq n$, the set $S$ contains either $v_i$ or $v_i'$. If $v_i\in S$, we set $x_i=1$; else we set $x_i=0$.

We want the last $m$ entries in $t$ to all add up to 9 each. We note that for each $1\leq j\leq m$, there must be at least one element of the form $v_i$ or $v_i'$ in the subset $S$ that has its \iith{(1+n+j)} entry set to 1. This is because none of the combinations of $c_j^1,c_j^2,c_j^3$ that have $9,4,2$ in this position, respectively, can add up to give 9 when all other elements have 0 in this position:
\begin{itemize}
	\item If only one of either $c_j^1$ or $c_j^2$ or $c_j^3$ are included in $S$, then the \iith{(1+n+j)} entry in the result is trivially 0.
	\item If any two of $c_j^1$, $c_j^2$ and $c_j^3$ are included in $S$, then the \iith{(1+n+j)} entry in the result is $(9\cdot 4)=36$ or $(9\cdot 2)=18$ or $(4\cdot 2)=8$, depending on which $c_j$ values are included. Since the characteristic of the field $\field$ is at least $63$, we know that $36,18,8$ are all different than 9.
	\item If all three of $c_j^1$, $c_j^2$ and $c_j^3$ are included in $S$, then the \iith{(1+n+j)} entry in the result is $(9\cdot 4)+(9\cdot 2)+(4\cdot 2)=62$. As before, since the characteristic of the field $\field$ is at least $63$, we know that $62\neq 9$.
\end{itemize}
Therefore, there is at least one literal in each clause $C_j$ whose value is 1 and $\Phi$ has a satisfying assignment.

\medskip
Having proved NP-completeness of $R$-Subset-$\phi_2$, we proceed to prove the general case of $R$-Subset-$\phi_\ell$ for $\ell\geq 3$.

\noindent\textbf{For $R$-Subset-$\phi_{\ell}$, where $\ell\geq 3$:}
The reduction algorithm for reducing a given 3-CNF formula $\Phi$ with $N$ variables $x_1,\ldots,x_N$ and $m$ clauses $C_1,\ldots,C_m$, each containing exactly three distinct literals to an instance of $R$-Subset-$\phi_{\ell}$ and the proof of soundness for that reduction has already been discussed in the proof sketch of \cref{thm:sym-poly-nphard} in \cref{sec:snark-2}. Here we only prove the completeness for that reduction.

\paragraph{Completeness.}
Completeness follows similarly to the previous case. For a satisfying assignment $X$ for $\Phi$, for each $i\in[N]$, either $v_i$ or $v_i'$ is included in subset $S$. Since each monomial is a combination of $\ell$ numbers, we include all the special elements $\alpha_0,\alpha_1,\ldots, \alpha_{\ell-1}$ to get the value 1 in the first column $\ell$ times. This guarantees that the first $N+1$ entries in $t$ are all 1.
Since $X$ is a satisfying assignment, each clause contains at least one literal with the value 1. For each clause $C_j$ (for $j\in[m]$), if there is exactly one literal with value 1, we include all the $\ell-1$ elements $c_j$. If there are exactly two literals with value 1, we include $\ell-2$ elements $c_j$. And if there are exactly three literals with value 1, we include $\ell-3$ elements $c_j$. As before, this ensures that the value 1 appears exactly $\ell$ times in the last $m$ columns and $\phi_{\ell}$ will evaluate to the target value 1 in these positions.

\medskip
\noindent\textbf{For $(s,R)$-Subset-$\phi_{\ell}$ for some $s\in\Theta(n)$:}
Let $\Phi$ be a given $3$-CNF formula with $N$ variables $x_1,\ldots,x_N$ and $m$ clauses $C_1,\ldots,C_m$. It is easy to see that this instance can be reduced to another 3-CNF instance $\Phi'$ with $n'=\mathsf{max}(m,N)$ variables and $n'=\mathsf{max}(m,N)$ clauses by adding ``dummy'' variables and clauses.  We can then use the reduction algorithms discussed above to reduce $\Phi'$ to an instance of $R$-Subset-$\phi_{\ell}$ with $n=\ell+2n'+(\ell-1)n'$ elements in $R$. Recall that this reduction is such that for a satisfying assignment $X'$ for $\Phi'$, the corresponding witness $S$ for the $R$-Subset-$\phi_{\ell}$ instance contains the following:
\begin{itemize}
	\item $\ell$ elements: It contains elements $\alpha_0,\ldots,\alpha_{\ell-1}$.
	\item $n'$ elements: For each $i\in[n']$, it either contains $v_i$ or $v_i'$.
	\item At least $(\ell-3)n'$ elements: Depending on how many literals have value 1, in clause $C_j$ (for $j\in[n']$), $S$ contains at least $\ell-3$ elements $c_j$.
\end{itemize}
As a result, the subset $S$ for the $R$-Subset-$\phi_{\ell}$ instance  contains at least $\ell+n'+(\ell-3)n'$ out of $n=\ell+2n'+(\ell-1)n'$ elements, \ie $s=|S|\in\Theta(n)$ for each $\ell\in[n]$. In other words,  there exists $s\in\Theta(n)$, for which $(s,R)$-Subset-$\phi_{\ell}$ is NP-complete.

\medskip
This concludes the proof of \cref{thm:sym-poly-nphard}.
\QED
\end{proof}

%% file: general_multisigs.tex
\subsection{SNARG-Compliant Multi-Signatures and Subset-\texorpdfstring{$\phi_\ell$}{Lg}}\label{sec:snarg-compliant}

In this section, we identify the properties of multi-signatures used in \cref{lem:lossw-snargs} to provide the connection with average-case SNARGs. We call multi-signature schemes that satisfy these properties as \emph{SNARG-compliant} multi-signature schemes.
\begin{definition}[SNARG-compliant multi-signatures]\label{def:SRDScompliant}
A multi-signature scheme $(\MSGen$, $\MSSign$, $\MSVerify$, $\MSCombine$, $\MSMverify)$ is \textsf{SNARG compliant} if it satisfies the following properties:
\begin{enumerate}
	\item The algorithm $\MSMverify$ is deterministic.
	\item Verification keys are independently and uniformly sampled from a ring $R=\field^k$ (for some $k$) with Hadamard Product. 	
	\item There exist polynomial-time algorithms $\MSVerifyAgg$ and $\fagg$, such that given a multi-signature $\sigmams\in\XSms$ on a message $m\in\MS$, corresponding to a set of keys $\sset{\vk_{i}}_{i\in S}$ for some subset $S\subseteq[n]$, the algorithm $\MSMverify(\pp_\ms, \{\vk_{i}\}_{i\in[n]},S,m,\sigma_{\ms})$ can be decomposed as follows:
    \begin{enumerate}
        \item $\vkagg=\fagg(\sset{\vk_i}_{i\in S})$.
        \item $b=\MSVerifyAgg(\ppms, \vkagg, m,\sigmams)$.
    \end{enumerate}
    \item There exists a PPT algorithm $\MSVerifyInv$ that on input the public parameters $\ppms$, a message $m$ and a multi-signature $ \sigmams$, outputs $\vk\in R$.

        We require that for $\vkagg=\fagg(\sset{\vk_i}_{i\in S})\in R$ and $\MSVerifyAgg(\ppms,\vkagg, m,\sigmams)=1$, it holds that $\MSVerifyInv$ computes the corresponding unique and well-defined key $\vkagg$, \ie
        \[
        \MSVerifyInv(\ppms,m,\sigmams)=\fagg(\sset{\vk_i}_{i\in S}).
        \]
    \item There exist degenerate keys $\skdeg$ and $\vkdeg$, and a PPT algorithm $\MSSignDeg$ such that $\sigmams\gets\MSSignDeg(\ppms,\skdeg,m)$ satisfies $\MSVerifyAgg(\pp_{\ms},\vkdeg,m,\sigma_{\ms})=1$.
    \end{enumerate}	
\end{definition}

We now show that an SRDS scheme based on a SNARG-compliant multi-signature scheme with key-aggregation function $\fagg=\phi_{\ell}$, implies SNARGs for average-case Subset-$\phi_{\ell}$. This reduction can be viewed as a generalization of \cref{lem:lossw-snargs}.

\begin{lemma} \label{lem:snarg-compliant}
Let $\field$ be a field, let $R=\field^k$ (for some $k$) be a ring with Hadamard product, let $\phi_\ell$ (for some $\ell\in\NN$, $\ell>1$) be an elementary symmetric polynomial over $R$, let $0<\alpha<1$ be a constant, and let $s(n)=\alpha \cdot n$. Assume that $\ssize{\field}=n^{\omega(1)}$ and that $n/\log|R|<1$.

If there exists an SRDS scheme based on a SNARG-compliant multi-signature scheme with key-aggregate function $\fagg=\phi_{\ell}$, then there exist SNARGs for average-case $(\subsetsize,R)$-Subset-$\phi_{\ell}$.
\end{lemma}
\begin{proof}
We give a construction of average-case SNARGs for $(\subsetsize,R)$-Subset-$\phi_{\ell}$ using an SRDS scheme based on an SRDS-compliant multi-signature scheme as per Definitions \ref{def:SRDS_multisig} and \ref{def:SRDScompliant}.
\begin{enumerate}
    \item
    $\ProofSetup(1^\secParam,1^n):$ Run the setup of the SRDS scheme $\TSSetup(1^\secParam,1^n)$ to output $\crs=(\ppms,\pp_2)$.
    \item $\ProofProve(\crs,x,w):$
    Given an average-case $\yes$ instance-witness pair $(x,w)\gets\yesD(1^n)$ of the form $x=(a_1,\ldots,a_n,t)$ and $w=S$, proceed as follows:
    \begin{itemize}[leftmargin=*]
        \item
        Let $\alpha=\phi_{\ell-1}(\sset{a_i}_{i\in S})$ and let $\vkdeg$ be the degenerate aggregate verification key.
        If $\alpha$ does not have an inverse in $R$, output $\bot$ and terminate. Else, compute
        \[
        a_{n+1}=(\vkdeg-\phi_\ell(\sset{a_i}_{i\in S}))\cdot \alpha^{-1}=(\vkdeg-t)\cdot \alpha^{-1}.
        \]
        Parse $\crs=(\ppms,\pp_2)$ and interpret the set $\sset{a_1,\ldots,a_n,a_{n+1}}$ as a set of $n+1$ verification keys $\sset{\vk_1,\ldots,\vk_{n+1}}$. Note that $\phi_{\ell}(\sset{\vk_i}_{i\in S'})=\vkdeg$ for $S'=S\cup\sset{n+1}$.
        \item Choose an arbitrary $m\in\MS$ and use $\MSSignDeg$ (as defined in \cref{def:SRDScompliant}) to compute
        \[
        \sigmams\gets\MSSignDeg(\ppms,\skdeg,m).
        \]
        \item Use the algorithm $\progP$ (that exists from \cref{def:SRDS_multisig}) to compute
        \[
        \pi\gets\progP(\crs,\vk_1,\ldots,\vk_{n+1}, S', m,\sigmams).
        \]
        \item Finally, output $(m,\sigmams,\pi)$.
    \end{itemize}
    \item $\ProofVerify(\crs,x,\pi):$ Parse $\crs=(\ppms,\pp_2)$ and $x=(a_1,\ldots,a_n,t)$, and proceed as follows:
    \begin{itemize}[leftmargin=*]
    	\item Compute $a_{n+1}$ as in the prover algorithm.
    	Interpret the set $\sset{a_1,\ldots,a_n,a_{n+1}}$ as a set of $n+1$ verification keys $\sset{\vk_1,\ldots,\vk_{n+1}}$.
    	\item
        Compute $\vk=\MSVerifyInv(\ppms,m,\sigmams)$ and check if $\vk$ equals the degenerate verification key $\vkdeg$ (that, by construction, satisfies $\vkdeg=\phi_{\ell}(\sset{\vk_i}_{i\in S'})$). Set $b'=1$ if $\vk=\vkdeg$ and $b'=0$ otherwise.
    	\item Run the verification algorithm
    	of the SRDS scheme $$b\gets\TSVerify((\ppms,\pp_2),\vk_1,\ldots,\vk_n,m,(\sigmams,\pi)).$$
    	\item Finally output $b\wedge b'$.
    \end{itemize}
\end{enumerate}

We now argue succinctness, completeness, and average-case soundness for this construction:

\smallskip\noindent\textbf{Succinctness.}
Succinctness follows from the succinctness of the SRDS scheme.

\smallskip\noindent\textbf{Completeness.}
Recall that each of the values $(a_1,\ldots,a_n)$ in an average case $\yes$ instance is sampled uniformly at random; hence, the output of an elementary symmetric polynomial on a randomly chosen subset $S$ of these values is also uniformly distributed. Given any average-case $\yes$ instance-witness pair $(x,w)\gets\yesD(1^n)$ of the form $x=(a_1,\ldots,a_n,t)$ and $w=S$, the probability that $\phi_{\ell-1}(\sset{a_i}_{i\in S})$ has an inverse in $R$ is $1-k/|\field|$.\footnote{We note that all elements of $R=\field^k$, except for the ones with a 0 in any of its vector coordinates, have an inverse in $R$.} Since our proof system only works for such instances, the rest of this argument assumes that this is the case.
Given $x=(a_1,\ldots,a_n,t)$ and $w=S$, it holds that $\fagg(\sset{a_i}_{i\in S})=\phi_{\ell-1}(\sset{a_i}_{i\in S})=t$ or equivalently, it holds for $S'=S\cup\sset{n+1}$ that
\[
\phi_\ell(\sset{a_i}_{i\in S'})=\phi_\ell(\sset{a_i}_{i\in S})+ a_{n+1}\cdot\phi_{\ell-1}(\sset{a_i}_{i\in S}) =\vkdeg.
\]
Recall in an SRDS-compliant multi-signature scheme, it holds that
\[
\MSVerifyInv(\ppms,m,\sigmams)=\phi_{\ell}(\sset{\vk_i}_{i\in S'}) =\vkdeg.
\]
Hence, $\MSVerifyAgg(\ppms, \vkdeg, m, \sigmams)=1$, \ie $\sigmams$ is a valid multi-signature on $m$ with respect to $\vkdeg$.
Since the multi-signature satisfies $\MSVerifyInv(\ppms, m, \sigmams)=\vkdeg$, completeness of SRDS based on an SRDS-compliant multi-signature scheme (see \cref{def:SRDS_multisig}) implies that the output of $\progP$, given this signature and $S'$ will be a valid SRDS signature. Completeness now holds with an overwhelming probability since $\phi_{\ell-1}(\sset{a_i}_{i\in S})$ has an inverse in $R$ with an overwhelming probability of $1-k/|\field|$.

\smallskip\noindent\textbf{Average-case soundness.}
Recall that each of the values $(a_1,\ldots,a_n,t)$ in $x\gets\noD(1^n)$ is sampled uniformly at random.
Let $\alpha=\phi_{\ell-1}(\sset{a_i}_{i\in S})$ and assume that $\alpha^{-1}$ exists.
Since $t$ is a randomly sampled value, so is $a_{n+1}=(\vkdeg-t)\cdot\alpha^{-1}$ for any $S\subseteq[n]$. We interpret the set of $n+1$ verification keys as  $\vk_i=a_i$ for $i\in[n+1]$; thus, the verification keys $\sset{\vk_1,\ldots,\vk_{n+1}}$ are uniformly distributed over $R$.
Since $n/\log |R|<1$ and the output of elementary symmetric polynomials is uniformly distributed, then with overwhelming probability (bounded by ${2^{n+1}}/{|R|}$),
there does not exist a subset $S'\subseteq[n+1]$ of size $\subsetsize+1$, such that $\phi_{\ell}(\sset{a_i}_{i\in S'})=\vkdeg$.

Given $(m,\sigmams,\pi)$, we check if: (1) $\sigmams$ is a valid multi-signature on $m$ with respect to $\vkdeg$ and (2) if $(\sigmams,\pi)$ is a valid SRDS on $m$. Recall that in a SNARG-compliant multi-signature scheme, given a multi-signature $\sigma_{\ms}$,  a message $m$, and public parameters $\pp_{\ms}$, there exists a unique aggregate verification key $\vkagg$ with respect to which $\sigma_{\ms}$ verifies, \ie
\[
\MSVerifyInv(\pp_{\ms},m,\sigma_{\ms})=\vkagg.
\]
Therefore, if check (1) goes through, then $\vkagg=\vkdeg$ is the only aggregate verification key for which $\sigma_{\ms}$ is a valid multi-signature on $m$. As argued earlier, with a high probability there does not exist a subset $S'\subseteq[n+1]$ such that $\phi_{\ell}(\{\vk_i\}_{i\in S'})=\vkdeg$. Also, from the soundness of SRDS based on a multi-signature scheme (\cref{def:SRDS_multisig}), we know that if there does not exist a subset $S'\subseteq[n+1]$ of size $s+1$, such that $\sigmams$ is a valid multi-signature on $m$ with respect to $\sset{\vk_i}_{i\in S'}$, then the probability of an adversary computing a valid SRDS $(\sigmams,\pi)$ on a message $m$ is negligible. Soundness now follows from the soundness of SRDS based on a multi-signature scheme.
\QED
\end{proof} 